\documentclass{article}
\usepackage{graphicx}
\usepackage{subfig}
\usepackage{amsfonts}
\usepackage{amsmath}
\usepackage{bm}
\usepackage{booktabs}
\usepackage{array}
\usepackage{amssymb}
\usepackage{pifont}
\usepackage{url}
\usepackage{fancyhdr}
\usepackage{indentfirst}
\usepackage{enumerate}
\usepackage{needspace}
\usepackage{longtable}
\usepackage{float} 
\usepackage{placeins}

\usepackage{dsfont}
\usepackage{setspace}
\usepackage[colorlinks=true,citecolor=blue]{hyperref}
\usepackage{amsthm}
\usepackage{color}
\usepackage{natbib}
\usepackage{comment}
\usepackage{capt-of}
\usepackage{multirow}

\def\d{\mathrm{d}}

\newcommand{\ES}{\mathrm{ES}}

\newcommand{\X}{\mathcal {X}}

\newcommand{\VaR}{\mathrm{VaR}}
\newcommand{\VaRp}{\mathrm{VaR}^+}
\newcommand{\RVaR}{\mathrm{RVaR}}

\newcommand{\R}{\mathbb{R}}

\newcommand{\N}{\mathbb{N}}
\newcommand{\p}{\mathbb{P}}

\newcommand{\M}{\mathcal{M}}

\newcommand{\U}{\mathrm{U}}
\newcommand{\id}{\mathds{1}}

\renewcommand{\ge}{\geqslant}
\renewcommand{\le}{\leqslant}

\renewcommand{\epsilon}{\varepsilon}

\renewcommand{\cdots}{\dots}

\theoremstyle{plain}
\newtheorem{theorem}{Theorem}

\newtheorem{lemma}{Lemma}
\newtheorem{proposition}{Proposition}
\theoremstyle{definition}

\newtheorem{example}{Example}

\theoremstyle{remark}
\newtheorem{remark}{Remark}

\newcommand{\cet}{\begin{center}}
	\newcommand{\ecet}{\end{center}}

\usepackage{refcount}
\setrefcountdefault{1}

\newcounter{saveexample}

\allowdisplaybreaks[4]

\begin{document}
\sloppy

\title{Robust distortion riskmetrics under Wasserstein ambiguity}

\author{	
       Yang Liu\thanks{\scriptsize 
       School of Science and Engineering, The Chinese University of Hong Kong (Shenzhen), Shenzhen, Guangdong 518172, China. Email: \texttt{yangliu16@cuhk.edu.cn}}		
\and
       Qiuqi Wang\thanks{\scriptsize  
		Maurice R.~Greenberg School of Risk Science, Robinson College of Business, Georgia State University, Atlanta, Georgia 30303, U.S.A. Email: \texttt{qwang30@gsu.edu}}		
\and
       Yihan Wang\thanks{\scriptsize  
		Corresponding author.
		School of Science and Engineering, The Chinese University of Hong Kong (Shenzhen), Shenzhen, Guangdong 518172, China. Email: \texttt{yihanwang6@link.cuhk.edu.cn}}
}

\maketitle

\begin{abstract}
Risk evaluation under distributional ambiguity is central to decision
making in finance, economics, and operations research. Wasserstein balls
provide a natural way to describe uncertainty around a reference
distribution.
We solve a natural yet open problem of robust optimization for the class of distortion riskmetrics with Wasserstein distance as the sole ambiguity constraint.
This chosen objective class does not require convexity, monotonicity, and continuity of distortion functions, encompassing many common risk measures and deviation measures.
First, we
characterize conditions under which direct convexification preserves the worst-case value.
Second, we develop a constructive method for exact worst-case evaluation when the direct convexification conditions fail.
Third, we construct explicit approximate worst-case distributions and
provide computable error bounds to assess their accuracy without
solving the exact problem. We apply these results to distributionally
robust portfolio selection and use numerical experiments to assess
approximation accuracy and the resulting portfolio decisions.
\vskip 10pt



\vskip 10pt

\end{abstract}

\section{Introduction}\label{sec:intro}

Risk measures assign numerical values to uncertain outcomes and thereby
provide criteria for assessing risk, comparing alternatives, and expressing
preferences. They are used across finance, insurance, economics, statistics,
and operations research. Measures of risk and variability capture different
aspects of a distribution, such as the severity of losses and their
dispersion; see \citet{ADEH99} and \citet{RUZ06}.

In decision problems, risk measures appear as objectives or constraints that
guide the choice of portfolios, insurance allocations, and operational
plans. Their use requires a probability model for the
uncertain outcomes, but this model is rarely known precisely. Distributionally robust
optimization (DRO) addresses this ambiguity by evaluating a decision under
the least favorable distribution in a prescribed ambiguity set.
Moment constraints \citep{P07,DY10} and probability constraints
\citep{WKS14} are common specifications of ambiguity sets; more general
formulations and their connections to risk aversion are surveyed by
\citet{RM22}.
DRO has been studied for expected utility \citep{NSU10,CHZ11}. The two most common risk measures,
Value-at-Risk (VaR) \citep{EOO03,NPS08} and Expected Shortfall (ES)
\citep{ZF09}, have also been widely studied under distributional ambiguity.
Under moment constraints, worst-case bounds have been obtained for
coherent risk measures \citep{L18}, Range Value-at-Risk (RVaR)
\citep{LSWY18}, and general distortion risk measures \citep{ZS18,SZ24}.
For the latter class, \citet{CLM25} develop tractable reformulations
under moment and support constraints, whereas \citet{PWW25} show a unified framework to optimize distortion riskmetrics with a general ambiguity set closed under concentration.

When a reference distribution is available, Wasserstein distance provides
a natural way to describe nearby alternatives. It measures the cost of
transporting probability mass and allows the alternative distributions to
have supports different from that of the reference model.
\citet{PW07} maximize worst-case expected portfolio returns over
Wasserstein balls under risk constraints.
\citet{MK18} compute worst-case expected costs over Wasserstein balls
around empirical distributions using convex programming.
\citet{BM19} derive dual formulas for worst-case expectations under
Wasserstein and more general transport constraints.
\citet{GK23} characterize distributions attaining worst-case expectations
over Wasserstein balls.
These expectation objectives are linear in the probability distribution.
For risk measures, \citet{W14} derives closed-form robust convex risk measures under Kantorovich/Wasserstein ambiguity. \citet{LMWW22} obtain explicit worst-case values for
coherent distortion risk measures and VaR over Wasserstein balls.
\citet{CHLW23} derive sharp RVaR bounds under $W_1$ ambiguity with
support information. \citet{BLLW25} obtain sharp bounds for robust VaR and RVaR aggregation problems with full dependence uncertainty. 

A natural and fundamental question is how to solve the corresponding
DRO problem for a nonconvex risk measure under a $p$-Wasserstein ($W_p$) ball alone, which is the question we study.
Specifically, let $\mathbf X\in\R^d$ be the random loss,
$P_0$ be the reference distribution, and $\bm\omega\in\mathcal W$
be a feasible decision. For a risk measure $\rho$, the robust problem is
\begin{equation}
\inf_{\bm\omega\in\mathcal W}
\sup_{F_{\mathbf X}\in\mathcal M_{r,p}^d(P_0,\epsilon)}
\rho\bigl(\bm\omega^\top\mathbf X\bigr),
\label{eq:dro_multi}
\end{equation}
where $\mathcal M_{r,p}^d(P_0,\epsilon)$ is the $p$-Wasserstein ball of
radius $\epsilon>0$ and order $1\le p<\infty$, with transportation
cost induced by the $\ell^r$-norm.
Despite its simple formulation, this question has not been fully
addressed.
The challenge is to determine how a single transportation budget is
allocated among changes in the mean, dispersion, and shape of the
distribution. Unlike moment constraints, which only restrict distributional summaries such as the mean and the variance, the transportation cost depends on the full reference distribution. This makes the natural problem in \eqref{eq:dro_multi} highly nontrivial to solve technically.

\citet{BPV24} and \citet{LVX25} study distortion objectives under a $W_2$
constraint together with fixed mean and variance. Fixing the second
moment makes the squared $W_2$ constraint linear in the quantile
function. In the above sense, their problem is essentially different from ours with only
Wasserstein ambiguity, where moment constraints provide extra essential distributional information; \citet{CHLW23} also emphasize this distinction.
\citet[Section~4.2]{BPV24} discuss the Wasserstein-only extension of their framework but consider only concave distortions.

To accommodate a broad range of nonconvex criteria, we choose our objective risk measure $\rho$ as a \emph{distortion riskmetric} introduced by \citet{WWW20a} and \citet{WWW20b}, which is defined as a signed Choquet integral. This class of risk measures connects
$L$-functionals in statistics \citep{HR09}, Yaari's dual utilities in
decision theory \citep{Y87}, distorted insurance premium principles
\citep{WYP97}, and spectral risk measures in finance \citep{A02}.
It includes VaR, ES, and RVaR as special cases.
Moreover, they include inverse-S-shaped distortion risk measures, describing probability weighting in cumulative prospect theory \citep{TK92}. Distortion riskmetrics allow nonmonotone distortions and thus include
common deviation measures \citep{RUZ06}, such as inter-quantile differences, Gini
deviations, and mean-median deviations \citep[see e.g.,][]{WWW20b}. A more important advantage of allowing nonmonotonicity is that we are able to solve the best-case problem from the worst-case one, noting that $\inf \rho_{h}=-\sup\rho_{-h}$ (see Remark \ref{rem:bc}).

\begin{table}[!t]
\centering
\small
\setstretch{1.0}
\setlength{\tabcolsep}{2.8pt}
\renewcommand{\arraystretch}{1.25}
\caption{Comparing literature of distributionally robust optimization problems for risk measures.}
\label{tab:literature}
\begin{tabular}{@{}>{\centering\arraybackslash}m{1.80cm}
>{\centering\arraybackslash}m{3.10cm}
>{\centering\arraybackslash}m{2.15cm}
>{\centering\arraybackslash}m{2.50cm}
>{\centering\arraybackslash}m{1.80cm}
>{\centering\arraybackslash}m{2.25cm}
>{\centering\arraybackslash}m{1.75cm}@{}}
\toprule
Literature & Ambiguity set & {\footnotesize Wasserstein ambiguity alone} & Objective risk measure & {\footnotesize Nonconvexity allowed} & {\footnotesize Nonmonotonicity allowed} & {\footnotesize Discontinuity of $h$ allowed} \\
\midrule
EOO03, NPS08, CHZ11, LSWY18 & Moment constraints & \ding{53} & VaR, ES, or RVaR & --- & --- & --- \\
\addlinespace[4pt]
NSU10, L18 & Moment constraints & \ding{53} & Convex risk measures & \ding{53} & \ding{53} & --- \\
\addlinespace[4pt]
ZS18, SZ24, CLM25 & Moment constraints & \ding{53} & Distortion risk measures & \ding{52} & \ding{53} & \ding{52} \\
\addlinespace[4pt]
ZF09 & Probability constraints & \ding{53} & ES & --- & --- & --- \\
\addlinespace[4pt]
PWW25 & Concentration-closed sets & \ding{53} & Distortion riskmetrics & \ding{52} & \ding{52} & \ding{52} \\
\addlinespace[4pt]
CHLW23 & $W_1$ balls or moment constraints + support information & \ding{53} & RVaR & --- & --- & --- \\
\addlinespace[4pt]
BPV24 & $W_2$ balls + moment constraints & \ding{53} & Distortion risk measures & \ding{52} & \ding{53} & \ding{53} \\
\addlinespace[4pt]
LVX25 & $W_2$ balls + moment constraints & \ding{53} & Distortion riskmetrics & \ding{52} & \ding{52} & \ding{52} \\
\addlinespace[4pt]
W14 & $W_p$ balls & \ding{52} & Convex risk measures & \ding{53} & \ding{53} & --- \\
\addlinespace[4pt]
LMWW22 & $W_p$ balls & \ding{52} & Distortion risk measures & \ding{53} & \ding{53} & \ding{53} \\
\midrule
\textbf{This paper} & \textbf{$W_p$ balls} & \ding{52} & \textbf{Distortion riskmetrics} & \ding{52} & \ding{52} & \ding{52} \\
\bottomrule
\end{tabular}
\par\vspace{5pt}
\begin{minipage}{0.98\textwidth}
\footnotesize
\setstretch{1.05}
\textit{Notes.}
EOO03: \citet{EOO03}; NPS08: \citet{NPS08}; CHZ11: \citet{CHZ11}; LSWY18: \citet{LSWY18}; NSU10: \citet{NSU10}; 
L18: \citet{L18}; 
ZS18: \citet{ZS18}; SZ24: \citet{SZ24}; CLM25: \citet{CLM25};
ZF09: \citet{ZF09}; PWW25: \citet{PWW25}; CHLW23: \citet{CHLW23}; BPV24: \citet{BPV24}; LVX25: \citet{LVX25}; W14: \citet{W14};
LMWW22: \citet{LMWW22}. We only record convexity, monotonicity, and continuity assumptions for general classes of risk measures. For specific ones, $\VaR$ is nonconvex and has discontinuous distortion function; $\ES$ is convex and has continuous distortion function; $\RVaR_{\alpha,\beta}$ is nonconvex for $0<\alpha\le\beta<1$ and has continuous distortion function for $0<\alpha<\beta\le 1$; $\VaR$, $\ES$, and $\RVaR$ are all monotone.
\end{minipage}
\end{table}

We solve the DRO problem under Wasserstein ambiguity for
distortion riskmetrics, without requiring convexity or monotonicity of objective risk measures, or absolute continuity of distortion functions. To our best knowledge, such a problem has not been addressed even for general distortion risk measures. Table~\ref{tab:literature} compares our formulation with existing literature on DRO
problems in terms of their objective classes, ambiguity sets, and
required assumptions.

Our main contribution is three-fold.
First, we determine conditions under which the common idea of direct convexification \citep[see e.g.,][]{ZKR13,CLM25} can
solve the original problem. Building on the concentration principle of \citet{PWW25},
Theorem~\ref{thm:flatness} characterizes those conditions when replacing a risk measure by its
convex counterpart (i.e., replacing the general nonconcave distortion function by its concave envelope) preserves the worst-case value. This clarifies cases when
existing convex methods give a sharp risk assessment and when such a
replacement can be overly conservative.
Second, we obtain the exact worst-case solution as the limit of solutions
to a sequence of solvable problems with regularized distortions
(Theorem~\ref{thm:regularization}).
This provides a constructive method for evaluating the original risk measure when direct convexification conditions fail. We also give a discrete implementation with convergence
guarantees.
Third, we provide explicit approximations for settings where calculating the exact worst-case values by taking the limit would be costly and hence unavailable. Theorem~\ref{thm:radius_bound} gives
feasible distributions with computable risk values and bounds on their
approximation error. These bounds show how accuracy depends on the size
of the ambiguity set and allow a decision maker to assess the approximation
without first solving the exact problem. This offers a practical tradeoff
between computational cost and guaranteed accuracy when risk evaluation
is part of a larger decision problem.

We apply these three approaches to robust portfolio selection under
elliptical benchmarks, obtaining exact and approximate solutions
and bounds on the resulting decision error. Spherical benchmarks admit
explicit portfolio weights. Numerical experiments assess portfolio
performance and approximation accuracy.

The remainder of the paper is organized as follows.
Section~\ref{sec:setting} introduces the setup and preliminary results.
Section~\ref{sec:convexification} characterizes the conditions under which direct convexification
preserves the worst-case value.
Section~\ref{sec:regularization} solves the worst-case problem through
regularization and presents a discrete implementation.
Section~\ref{sec:gap_bounds} constructs explicit approximate worst-case
distributions and bounds their approximation error.
Section~\ref{sec:portfolio} applies the results to portfolio selection,
and Section~\ref{sec:numerics} presents the numerical experiments.
Concluding remarks discuss the implications for robust risk evaluation
and decision making.

\section{Preliminaries and problem setup}\label{sec:setting}

\subsection{Wasserstein distances, ambiguity, and dimension reduction}
\label{subsec:wasserstein}

Let $(\Omega,\mathcal F,\mathbb P)$ be an atomless probability space. Let $\X$ be a set of single-dimensional random variables and $\M$ be a set of distributions on $\R$. For $p\in[1,\infty)$ and $d\in\N$, let $\mathcal M_p^d$ be the set of distributions on $\mathbb R^d$
with finite $p$-th moment in each component and let $L^p$ be the set of random variables with finite $p$-th moment. Let $L^\infty$ and $\M^d_\infty$ be those with bounded support. For
$d=1$, write $\mathcal M_p=\mathcal M_p^1$. For $x,y\in\mathbb R$, write $x\vee y=\max\{x,y\}$ and $x_+=x\vee0$. Throughout, we adopt the convention $0/0=0$. All terms of ``increasing" and ``decreasing" are in the nonstrict sense.

For $F\in\mathcal M_1$, define its left- and right-quantile functions,
respectively, by
\[
F^{-1}(u)=\inf\{x\in\mathbb R:F(x)\ge u\},
\qquad
F^{-1+}(u)=\inf\{x\in\mathbb R:F(x)>u\},
\qquad u\in[0,1],
\]
with the convention $\inf\emptyset=\infty$. Both quantiles are finite
for $0<u<1$; their endpoint values may be infinite.
For a random variable $X$, write $F_X$ as its distribution function.

For $F,G\in\mathcal M_p$, the \emph{one-dimensional $p$-Wasserstein distance} is defined as
\[
W_p(F,G)
=
\left(\int_0^1|F^{-1}(u)-G^{-1}(u)|^p\,\mathrm d u\right)^{1/p}.
\]
For a benchmark $G\in\mathcal M_p$ and a radius $\epsilon\ge0$, the Wasserstein distance induces the ambiguity described by a \emph{one-dimensional $p$-Wasserstein ball}, defined as
\[
\mathcal M_p(G,\epsilon)
=
\{F\in\mathcal M_p:W_p(F,G)\le\epsilon\}.
\]
For the multivariate problem, fix $r\ge1$ and let $\|\cdot\|_r$ be the
$\ell^r$-norm on $\mathbb R^d$. The associated \emph{$p$-Wasserstein distance} is
\[
W_{r,p}^d(F,G)
=
\inf_{F_{\mathbf X}=F,\,F_{\mathbf Y}=G}
\left(\mathbb E\|\mathbf X-\mathbf Y\|_r^p\right)^{1/p},
\qquad F,G\in\mathcal M_p^d.
\]
For $P_0\in\mathcal M_p^d$ and $\epsilon\ge0$, the corresponding \emph{$p$-Wasserstein ball} is
\[
\mathcal M_{r,p}^d(P_0,\epsilon)
=
\{F\in\mathcal M_p^d:W_{r,p}^d(F,P_0)\le\epsilon\}.
\]


Fix a benchmark random vector $\mathbf X_0\in \X^d$ with distribution $P_0$ and a nonempty
decision set $\mathcal W\subseteq\mathbb R^d$. A decision
$\bm\omega\in\mathcal W$ produces the aggregate loss
$\bm\omega^\top\mathbf X$ from the random vector $\mathbf X\in\mathbb R^d$ (e.g., the future outcome of assets). Let
$\rho:\mathcal M_p\to\R$ be a general objective functional
of the aggregate loss distribution, which is not necessarily convex. Define
\begin{equation}\label{eq:wc_multi}
V_\rho(\bm\omega)
=
\sup_{F_{\mathbf X}\in\mathcal M_{r,p}^d(P_0,\epsilon)}
\rho\bigl(F_{\bm\omega^\top\mathbf X}\bigr).
\end{equation}
The dependence on $\epsilon$ is suppressed in $V_\rho$.
The main decision problem~\eqref{eq:dro_multi} is therefore
$\inf_{\bm\omega\in\mathcal W}V_\rho(\bm\omega)$.

Let the aggregated reference distribution be $G_{0,\bm\omega}=F_{\bm\omega^\top\mathbf X_0}$ and the radius be $\epsilon_{\bm\omega}=\|\bm\omega\|_{r'}\epsilon$.
Theorem~5 of \citet{MWW25} gives an exact linear projection of a Wasserstein ball \citep[see also its use in Example~12 of][]{
PWW25}:
\begin{align}
\left\{
F_{\bm\omega^\top\mathbf X}:
F_{\mathbf X}\in\mathcal M_{r,p}^d(P_0,\epsilon)
\right\}=
\mathcal M_p\bigl(G_{0,\bm\omega},\epsilon_{\bm\omega}\bigr).
\label{eq:linear_projection}
\end{align}
The result \eqref{eq:linear_projection} shows that for every $\bm\omega\in\mathcal W$ for which the supremum is well defined, the ambiguity characterized by a multi-dimensional Wasserstein ball boils down to a one-dimensional Wasserstein ball with reference distribution $G_{0,\bm\omega}$ and radius $\epsilon_{\bm\omega}$. Consequently, the original worst-case functional in \eqref{eq:wc_multi} becomes
\begin{equation}\label{eq:wc_one}
V_\rho(\bm\omega)
=
\sup_{F\in\mathcal M_p(G_{0,\bm\omega},\epsilon_{\bm\omega})}\rho(F),
\end{equation}
and the original multivariate problem in \eqref{eq:dro_multi} is exactly
\[
\inf_{\bm\omega\in\mathcal W}
\sup_{F\in\mathcal M_p(G_{0,\bm\omega},\epsilon_{\bm\omega})}\rho(F).
\]

\subsection{Main robust optimization problem with distortion riskmetrics}
\label{subsec:riskmetrics}
\label{subsec:envelopes}

We now specialize the general objective functional $\rho$ to a specific general class of risk measures. 
Let
\[
\mathcal H
=
\{h:[0,1]\to\mathbb R:h\text{ is of bounded variation and }h(0)=0\}.
\]
For a \emph{distortion function} $h\in\mathcal H$, a \emph{distortion riskmetric} \citep{WWW20a,WWW20b} is a signed Choquet integral $\rho_h:\X\to\R$ given by
\[
\rho_h(X)
=
\int_0^\infty h\bigl(\mathbb P(X>x)\bigr)\,\mathrm d x
+
\int_{-\infty}^0
\left[h\bigl(\mathbb P(X>x)\bigr)-h(1)\right]\,\mathrm d x,
\]
whenever the two integrals are finite.
Because $\rho_h(X)=\rho_h(Y)$ whenever $X\sim F_Y$, we use the shorthand $\rho_h(F_Y)=\rho_h(Y)$ for a random variable $Y$.
Throughout, whenever a Wasserstein order $p$ is fixed, we assume $L^p\subseteq \X$. As a consequence, every distortion
riskmetric appearing in an optimization problem is assumed to be real-valued
on $\mathcal M_p$, unless stated otherwise.

Adopting distortion riskmetrics as our specific objective functionals, the inner problem \eqref{eq:wc_one} becomes
\begin{equation}\label{eq:wc_decision}
    V_h(\bm\omega):=V_{\rho_h}(\bm\omega)
=
\sup_{F\in\mathcal M_p(G_{0,\bm\omega},\epsilon_{\bm\omega})}\rho_h(F).
\end{equation}
For each fixed decision $\omega$, we suppress the dependence of the projected benchmark $G_{0, \omega}$ and radius $\varepsilon_\omega$ on $\omega$, writing them as $G_0$ and $\varepsilon$, respectively, and use
\begin{equation}\label{eq:wc_distortion}
S_h(\epsilon):=
\sup_{F\in\mathcal M_p(G_{0},\epsilon)}\rho_h(F).
\end{equation}

\begin{remark}[Best-case distortion riskmetrics]\label{rem:bc}
Importantly, distortion riskmetrics are not necessarily monotone, and we have $\rho_{-h}=-\rho_h$. As a direct consequence, we can also solve the best-case problem of distortion riskmetrics by the solution to \eqref{eq:wc_distortion} with distortion function $-h$. Namely,
$$\inf_{F\in\mathcal M_p(G_{0},\epsilon)}\rho_{h}(F)=-\left\{\sup_{F\in\mathcal M_p(G_{0},\epsilon)}\rho_{-h}(F)\right\}=-S_{-h}(\epsilon).$$
In this sense, all results in the worst-case problem of distortion riskmetrics throughout the paper can naturally induce their best-case counterparts whenever $-h$ also satisfies the assumptions in the results.
\end{remark}

Our main analysis focuses on solving or approximating the inner problem
$S_h(\epsilon)$ in \eqref{eq:wc_distortion}. Once this value is available, the outer problem amounts to
minimizing the decision-dependent value
$V_h(\bm\omega)$ in \eqref{eq:wc_decision} over $\bm\omega\in\mathcal W$.
Section~\ref{sec:portfolio} develops the three approaches for the outer
portfolio problem under an elliptical benchmark.

\subsection{Properties of distortion riskmetrics and concave envelopes}

The following assertions, proved by \cite{WWW20a}, are useful for subsequent discussions: $\rho_h$ is convex if and only if $h$ is concave. Moreover, if $h$ is absolutely continuous, define
$\psi_h(u)=h'(1-u)$ for Lebesgue-a.e.~$u\in(0,1)$.
We have the following quantile representation:
\begin{equation}
\rho_h(X)
=
\int_0^1F_X^{-1}(u)\psi_h(u)\,\mathrm d u,~~X\in\X.
\label{eq:quantile_density}
\end{equation}

\begin{remark}
\label{rem:finiteness}
The finiteness assumption implies
\begin{equation}
\lim_{t\downarrow0}h(t)=h(0)=0
\qquad\mbox{and}\qquad
\lim_{t\uparrow1}h(t)=h(1).
\label{eq:h_endpoints}
\end{equation}
An endpoint jump would make a Choquet integral infinite for an unbounded
$X\in L^p$ of the appropriate sign. If the distortion function $h$ is concave, then
endpoint continuity implies absolute continuity. Moreover, because
$\psi_{h}$ is nondecreasing, it is clear that $\psi_{h}\in L^{p'}(0,1)$ by the finiteness of $\rho_{h}$ and $L^p$-duality; for $a>1$, denote its conjugate as $a'$ with $1/a+1/a'=1$ and set $a'=\infty$ when $a=1$.
\end{remark}

Because $h\in\mathcal H$ is of bounded variation, its left- and right-limits
exist. Let $h(t+)=\lim_{x\downarrow t}h(x)$, $t\in[0,1)$, and $h(t-)=\lim_{x\uparrow t}h(x)$, $t\in(0,1]$, be the right- and left-continuous versions of $h$, respectively. Write $h(1+)=h(1)$ and $h(0-)=h(0)$.
We define the upper-semicontinuous version
of $h$ as
\[
\widehat h(t)=h(t+)\vee h(t-)\vee h(t),\quad t\in(0,1),
\quad \widehat h(0)=0,\quad \widehat h(1)=h(1).
\]
Thus $\widehat h(t)=h(t)$ at every continuity point of $h$, and no
modification is made at the endpoints.
For $t\in[0,1]$, define the \emph{concave envelope} of $h$ by
\[
h^*(t)
=
\inf\{\ell(t):\ell\in\mathcal H,\ \ell\ge h,\
\ \ell\text{ is concave on }[0,1]\}.
\]
Equivalently, $h^*$ is the smallest concave distortion function dominating
$h$. By Proposition~1 of \citet{PWW25},
$h^*=(\widehat h)^*\ge\widehat h\ge h$.
Moreover, $\{t\in[0,1]:\widehat h(t)\ne h^*(t)\}$ is a union of disjoint open
intervals, and $h^*$ is linear on each such interval. We collect these intervals in the following set:
\[
\mathcal I_h
=
\left\{
(1-b,1-a):0\le a<b\le1,\
\widehat h\ne h^*\text{ on }(a,b),\
\ \widehat h(a)=h^*(a),\
\ \widehat h(b)=h^*(b)
\right\}.
\]

Distortion riskmetrics include many common risk measures and deviation measures \citep[see e.g.,][]{WWW20a}. Especially, if $h$ is increasing and $h(1)=1$, then $\rho_h$ is a distortion risk measure, which includes Value-at-Risk and Expected Shortfall, the most popular risk measures in finance and insurance.
For a random loss $X$ and a probability level $\alpha$, define
\emph{Value-at-Risk} (VaR) and \emph{upper Value-at-Risk} ($\VaR^+$) by
\[
\VaR_\alpha(X)=F_X^{-1}(\alpha),\quad \alpha\in(0,1],
\qquad
\VaRp_\alpha(X)=F_X^{-1+}(\alpha),\quad \alpha\in[0,1),
\]
and, for integrable $X$, define \emph{Expected Shortfall} (ES) \citep{RU02} by
\[
\ES_\alpha(X)
=
\frac{1}{1-\alpha}\int_\alpha^1\VaR_u(X)\,\mathrm d u,
\qquad \alpha\in(0,1).
\]
For $\alpha\in(0,1)$, let $h_\alpha(t)=\id_{(1-\alpha,1]}(t)$, $t\in[0,1]$.
By definition, we have
$$\widehat h_\alpha(t)=\id_{[1-\alpha,1]}(t)~~\mbox{and}~~ h_\alpha^*(t)=(\widehat h_\alpha)^*(t)=\left(\frac{t}{1-\alpha}\right)\wedge1,~~t\in[0,1].$$
We can directly check that $\rho_{h_\alpha}=\VaR_\alpha$, $\rho_{\widehat h_\alpha}=\VaRp_\alpha$, $\rho_{h_\alpha^*}=\ES_\alpha$, 
and $\mathcal I_{h_\alpha}=\mathcal I_{\widehat h_\alpha}=\{(\alpha,1)\}$.

For $0<\alpha<\beta\le1$, define \emph{Range Value-at-Risk} (RVaR) \citep{CDS10,ELW18} by
\[
\RVaR_{\alpha,\beta}(F)
=\frac1{\beta-\alpha}\int_\alpha^\beta F^{-1}(u)\,\mathrm d u.
\]
Its distortion and quantile density are
\[
h_{\alpha,\beta}(t)
=\left(\frac{(t-1+\beta)_+}{\beta-\alpha}\right)\wedge1,
\quad 0\le t\le1,
\qquad
\psi_{h_{\alpha,\beta}}(u)
=\frac{\mathbf1_{(\alpha,\beta)}(u)}{\beta-\alpha},
\quad 0<u<1.
\]
For $\beta<1$, $h_{\alpha,\beta}^*=h_\alpha^*$ and
$\mathcal I_{h_{\alpha,\beta}}=\{(\alpha,1)\}$.
At the endpoints, $\RVaR_{\alpha,1}=\ES_\alpha$ and
$\RVaR_{\alpha,\beta}(F)\to\VaRp_\alpha(F)$ as $\beta\downarrow\alpha$.

\section{Worst-case distortion riskmetrics by direct convexification}\label{sec:convexification}

A common way to handle a nonconvex worst-case objective $\rho_h$ is to equivalently replace it by the worst case of its
convex counterpart $\rho_{h^*}$. In this section, we examine the conditions when such an equivalence holds under Wasserstein ambiguity.

Let
$\mathcal I$ be a countable collection of pairwise disjoint, non-degenerate intervals in $[0,1]$, and let $\lambda$ denote Lebesgue measure. For
$F\in\mathcal M_1$, define the \emph{$\mathcal I$-concentration} of $F$ as the distribution $F^{\mathcal I}$ whose left-quantile
function is the left-continuous modification of
\[
F^{-1}(u)
\id_{\{u\notin\bigcup_{C\in \mathcal I}C\}}
+
\sum_{C\in \mathcal I}
\frac{\int_CF^{-1}(v)\,\mathrm d v}{\lambda(C)}
\id_{\{u\in C\}},
\qquad 0<u<1.
\]
An ambiguity set $\mathcal M\subset\mathcal M_1$ is \emph{closed under
concentration within $\mathcal I$} if $F^\mathcal I\in\mathcal M$ for every
$F\in\mathcal M$.
Under such closeness under concentration, \cite{PWW25} showed that there is an equivalence between the worst-case distortion riskmetric $\rho_h$ and its convex counterpart $\rho_{h^*}$. We present the result as the following lemma, which is useful for subsequent discussions.
\begin{lemma}[Theorem 1 of \citealp{PWW25}]
\label{lem:concentration}
Let $\mathcal M\subset\mathcal M_1$ be an ambiguity set and
$h\in\mathcal H$. Suppose that $h=\widehat h$ and $\mathcal M$ is closed
under concentration within $\mathcal I_h$. Then
\begin{equation}
\sup_{F\in\mathcal M}\rho_h(F)
=
\sup_{F\in\mathcal M}\rho_{h^*}(F).
\label{eq:concentration_equal}
\end{equation}
Moreover, if the supremum on the right-hand side of
\eqref{eq:concentration_equal} is attained at some
$F^*\in\mathcal M$, then $(F^*)^{\mathcal I_h}$ attains both suprema in
\eqref{eq:concentration_equal}.
\end{lemma}

For $F\in\mathcal M_1$, we say that $F^{-1}$ is
\emph{flat on $\mathcal I_h$} if, for every $C\in\mathcal I_h$,
there exists $c_C\in\mathbb R$ such that
$F^{-1}(u)=c_C$ for Lebesgue-a.e.~$u\in C$. Example~11 of
\citet{PWW25} shows that benchmark flatness on $\mathcal I_h$ is
sufficient for $\mathcal M_p(G_0,\epsilon)$ to be closed under concentration
within $\mathcal I_h$. Therefore, if $h=\widehat h$ and $G^{-1}_0$ is flat on
$\mathcal I_h$, then
Lemma~\ref{lem:concentration} gives
\begin{equation}
S_h(\epsilon)=S_{h^*}(\epsilon).
\label{eq:flatness_equal}
\end{equation}
If the concave-envelope supremum is attained at some
$F^*\in\mathcal M_p(G_0,\epsilon)$, then $(F^*)^{\mathcal I_h}$ attains both
$S_h(\epsilon)$ and $S_{h^*}(\epsilon)$.
The implication in \eqref{eq:flatness_equal} is not a new result of this
paper. It is the direct specialization of the concentration theorem of
\citet{PWW25} to a one-dimensional Wasserstein ball.

Inspired by \eqref{eq:flatness_equal}, we characterize the conditions under which such a conversion holds under Wasserstein ambiguity, so that the original inner problem \eqref{eq:wc_distortion} can be solved by established convex optimization methods.

\begin{theorem}
\label{thm:flatness}
Let $1\le p<\infty$, $G_0\in\mathcal M_p$, and $h\in\mathcal H$ with
$h=\widehat h$ and $h^*\not\equiv0$.
\begin{enumerate}[(i)]
\item If $p>1$, then, for every $\epsilon>0$,
\begin{equation}
S_h(\epsilon)=S_{h^*}(\epsilon)
\quad\Longleftrightarrow\quad
G^{-1}_0\text{ is flat on }\mathcal I_h.
\label{eq:flatness_iff}
\end{equation}
\item If $p=1$, flatness of $G_0^{-1}$ on $\mathcal I_h$ implies
$S_h(\epsilon)=S_{h^*}(\epsilon)$ for every $\epsilon>0$.
If $G_0^{-1}$ is not flat on $\mathcal I_h$, then
$S_h(\epsilon)<S_{h^*}(\epsilon)$ for all sufficiently small $\epsilon>0$.
\end{enumerate}
\end{theorem}

Theorem~\ref{thm:flatness} shows that one can only convert the original worst-case distortion riskmetric problem to its equivalent concave-envelope counterpart when the reference quantile of the Wasserstein ball is flat on $\mathcal I_h$ for $p>1$ (flatness remains sufficient for $p=1$). Thus the sufficient condition discussed by \cite{PWW25} for the equivalence to hold is also necessary. This indicates that the method of solving the nonconvex worst-case distortion riskmetrics by direct convexification is limited in practice.

\begin{example}[VaR, RVaR, and ES]\label{ex:vares}
Fix $0<\alpha<\beta<1$
and define
\begin{equation}\label{eq:wc_vares}
\begin{aligned}
S_{\VaR_\alpha}&:=\sup_{F\in\mathcal M_p(G_0,\epsilon)}\VaR_\alpha(F),
& S_{\VaRp_\alpha}&:=\sup_{F\in\mathcal M_p(G_0,\epsilon)}\VaRp_\alpha(F),\\
S_{\RVaR_{\alpha,\beta}}&:=\sup_{F\in\mathcal M_p(G_0,\epsilon)}\RVaR_{\alpha,\beta}(F),
& S_{\ES_\alpha}&:=\sup_{F\in\mathcal M_p(G_0,\epsilon)}\ES_\alpha(F).
\end{aligned}
\end{equation}
Taking $h=\widehat h_\alpha$ and $h=h_{\alpha,\beta}$ in
Theorem~\ref{thm:flatness}(i) gives, for $p>1$ and every $\epsilon>0$,
\begin{equation}
S_{\VaRp_\alpha}=S_{\ES_\alpha}
\quad\Longleftrightarrow\quad
S_{\RVaR_{\alpha,\beta}}=S_{\ES_\alpha}
\quad\Longleftrightarrow\quad
G^{-1}_0\text{ is flat on }(\alpha,1).
\label{eq:var_es_flatness}
\end{equation}
Moreover, $S_{\VaR_\alpha}=S_{\VaRp_\alpha}$ for every $p\ge1$
\citep[Proposition~4]{LMWW22}.
For $p=1$, Theorem~\ref{thm:flatness}(ii) gives the same sufficient
condition, and both gaps from $S_{\ES_\alpha}$ are strictly positive for all
sufficiently small radii when flatness fails. Equality at a single radius
need not imply flatness: with $G_0=\U[0,1]$, $\alpha=1/2$, and
$\epsilon=1/8$, all four worst-case values equal $1$ for every
$\beta\in(1/2,1)$, although $G_0^{-1}(u)=u$ is not flat on $(1/2,1)$.
\end{example}

For increasing and normalized concave distortion functions $h^*$, the Wasserstein worst-case
$\rho_{h^*}$, a coherent distortion risk measure, is already solved in the literature \citep[see e.g.,][]{LMWW22}.
Our distortion riskmetric $\rho_{h^*}$ differs slightly from coherent distortion risk measures because it needs not satisfy monotonicity and normalization. We present the following result for the worst-case value $S_{h^*}(\epsilon)$ of the distortion riskmetric $\rho_{h^*}$ for theoretical completeness.

\begin{proposition}
\label{prop:wc_hstar}
Let $1\le p<\infty$, $G_0\in\mathcal M_p$, $\epsilon>0$, and
$h\in\mathcal H$. We have
\begin{equation}
S_{h^*}(\epsilon)=\rho_{h^*}(G_0)+\epsilon\|\psi_{h^*}\|_{p'}.
\label{eq:bench_hstar}
\end{equation}
If $h^*\not\equiv0$, the supremum is attained if and only if $p>1$, or
$p=1$ and $\{|\psi_{h^*}|=\|\psi_{h^*}\|_\infty\}$ has positive Lebesgue measure.
In these cases, an optimizer $F^*$ is given by the left-continuous
modification of $(F^*)^{-1}(u)=G_0^{-1}(u)+r^*(u)$ for $0<u<1$, where
\begin{equation}
r^*(u)=
\begin{cases}
\displaystyle
\epsilon\frac{\operatorname{sgn}(\psi_{h^*}(u))|\psi_{h^*}(u)|^{p'-1}}
{\|\psi_{h^*}\|_{p'}^{p'-1}}, & p>1,\\[3mm]
\displaystyle
\epsilon\frac{\operatorname{sgn}(\psi_{h^*}(u))
\mathbf 1_{\{|\psi_{h^*}(u)|=\|\psi_{h^*}\|_\infty\}}}
{\int_0^1\mathbf 1_{\{|\psi_{h^*}(v)|=\|\psi_{h^*}\|_\infty\}}\,\mathrm d v}, & p=1.
\end{cases}
\label{eq:bench_shift}
\end{equation}
It satisfies $W_p(F^*,G_0)=\epsilon$ and is unique when $p>1$.
\end{proposition}

\begin{remark}
\label{rem:unique_hstar}
For $p=1$, uniqueness need not hold for Proposition \ref{prop:wc_hstar}, and attainment
can fail: $h(t)=h^*(t)=t-t^2/2$ gives $\psi_{h^*}(u)=u$, which never attains its
essential supremum on $(0,1)$.
For $p=1$ and $h^*\not\equiv0$, let $P_n$, $n\ge1$, be the distributions
whose left-quantile functions are the left-continuous modifications of
\begin{equation}
P_n^{-1}(u)=G_0^{-1}(u)
+\epsilon\frac{\operatorname{sgn}(\psi_{h^*}(u))|\psi_{h^*}(u)|^n}
{\int_0^1|\psi_{h^*}(v)|^n\,\mathrm d v},\qquad 0<u<1.
\label{eq:bench_sequence}
\end{equation}
Each $P_n$ has $W_1(P_n,G_0)=\epsilon$ and
$\rho_{h^*}(P_n)\to S_{h^*}(\epsilon)$ as $n\to\infty$.
\end{remark}

\section{Exact worst-case value of distortion riskmetrics}
\label{sec:regularization}

\subsection{Absolutely continuous distortions and regularization}

We first solve the worst-case problem for absolutely continuous distortions
and then extend the analysis to general distortions by taking their limits.
The following result gives a direct dual formulation for every
$1\le p<\infty$ and a continuous min--max characterization for $p>1$.

\begin{proposition}
\label{prop:isotonic}
Let $1\le p<\infty$, $G_0\in\mathcal M_p$, and $\epsilon>0$.
Suppose that $h\in\mathcal H$ is absolutely continuous and
$\psi_h\in L^{p'}(0,1)$.
The following hold.
\begin{enumerate}[(i)]
\item The value $S_h(\epsilon)$ is finite and
\begin{equation}
S_h(\epsilon)
=\min_{\lambda\ge0}
\left\{
\frac{\lambda\epsilon^p}{p}
+\sup_{F\in\M_p}
\left[\rho_h(F)-\frac{\lambda}{p}W_p^p(F,G_0)\right]
\right\}.
\label{eq:reg_dual}
\end{equation}
An optimal distribution for $S_h(\epsilon)$ exists if $p>1$, or if $p=1$ and
\begin{equation}
\lim_{\delta\downarrow0}
\left\|\psi_h\mathbf1_{(0,\delta)\cup(1-\delta,1)}\right\|_\infty=0.
\label{eq:psi_endpoints}
\end{equation}
A distribution $F\in\M_p(G_0,\epsilon)$ is optimal if and only if it
attains the inner supremum at some $\lambda^*\ge0$ satisfying
$\lambda^*\bigl[W_p^p(F,G_0)-\epsilon^p\bigr]=0$.
Denote any such optimal distribution by $P_{\lambda^*}$.
If $h(1)\ne0$, every optimal pair satisfies $\lambda^*>0$ and
$W_p(P_{\lambda^*},G_0)=\epsilon$.

\item For $p>1$ and $\lambda>0$, the inner maximizer $P_\lambda$ is unique.
Its quantile is the left-continuous version of
\begin{equation}
P_\lambda^{-1}(u)
=\sup_{0\le a<u}\ \inf_{u<b\le1}z_{a,b}(\lambda),
\qquad 0<u<1,
\label{eq:reg_quantile}
\end{equation}
where $z_{a,b}(\lambda)$ is the unique root of
\begin{equation}
\lambda\int_a^b|z-G_0^{-1}(v)|^{p-2}(z-G_0^{-1}(v))\,\mathrm d v
=\int_a^b\psi_h(v)\,\mathrm d v,~~0\le a<b\le1.
\label{eq:block_root}
\end{equation}
If $h(1)\ne0$, the optimal pair
$(P_{\lambda^*},\lambda^*)$ is unique.
\end{enumerate}
\end{proposition}
For $p=2$, \eqref{eq:block_root} gives
\begin{equation}
z_{a,b}(\lambda)
=\frac1{b-a}\int_a^b
\left(G_0^{-1}(v)+\frac{\psi_h(v)}{\lambda}\right)\mathrm d v.
\label{eq:l2_block}
\end{equation}
Thus the $p=2$ optimal quantile is the so-called isotonic projection of
$G_0^{-1}+\psi_h/\lambda$, as in \citet{P22}. If $h(1)\ne0$, the multiplier
is chosen to meet the Wasserstein radius. Absolute continuity alone does not ensure
attainment for $p=1$; Remark~\ref{rem:unique_hstar} gives a counterexample.

We now allow interior discontinuities and a singular continuous component
of $h$.
Fix $h\in\mathcal H$, $1\le p<\infty$, $G_0\in\M_p$, and $\epsilon>0$.
Enumerate the interior discontinuities of $h$ as $1-\xi_k$, $k\in\N$,
so that the distinct $\xi_k\in(0,1)$ are the corresponding quantile
levels. All sums over $k$ run over these discontinuities. Set
\[
c_k^-=h((1-\xi_k)+)-h(1-\xi_k),
\qquad
c_k^+=h(1-\xi_k)-h((1-\xi_k)-).
\]
It follows that
\begin{equation}
\rho_h(F)
=
\int_{(0,1)}F^{-1}(u)\,\nu^c(\mathrm d u)
+\sum_k\left[c_k^-F^{-1}(\xi_k)+c_k^+F^{-1+}(\xi_k)\right],
\label{eq:quantile_bv}
\end{equation}
where $\nu^c$ is the reflected Lebesgue--Stieltjes measure of the continuous
part of $h$. It is non-atomic, and $|\nu^c|$ denotes its total variation measure.
The theorem below uses the endpoint integrability condition
\begin{equation}
\int_{(0,1)}(u\wedge(1-u))^{-1/p}\,|\nu^c|(\mathrm d u)
+\sum_k(|c_k^-|+|c_k^+|)(\xi_k\wedge(1-\xi_k))^{-1/p}<\infty.
\label{eq:endpoint_integral}
\end{equation}
For $\eta>0$ and $v\in[0,1]$, define
\begin{equation}
\begin{aligned}
\psi_\eta(v)
=
\int_{(0,1)}
\frac{\mathbf 1_{(u-\delta_\eta(u),u]}(v)}{\delta_\eta(u)}
\,\nu^c(\mathrm d u)
+
\sum_k
\frac{c_k^-}{\delta_{k,\eta}^-}
\mathbf 1_{(\xi_k-\delta_{k,\eta}^-,\xi_k]}(v)
+
\sum_k
\frac{c_k^+}{\delta_{k,\eta}^+}
\mathbf 1_{(\xi_k,\xi_k+\delta_{k,\eta}^+]}(v),
\end{aligned}
\label{eq:reg_density}
\end{equation}
where
\begin{equation}
\delta_\eta(u)=\eta\wedge\frac{u}{2},~~u\in[0,1],
\qquad
\delta_{k,\eta}^-=\eta\wedge\frac{\xi_k}{2},
\qquad
\delta_{k,\eta}^+=\eta\wedge\frac{1-\xi_k}{2}.
\label{eq:reg_widths}
\end{equation}
Let $h_\eta(t)=\int_{1-t}^1\psi_\eta(v)\,\mathrm d v$, $0\le t\le 1$.

Intuitively, the idea of the construction in \eqref{eq:reg_density} is to replace the original distortion $h$ by its averages over short intervals nearby to eliminate singular and discontinuous points of $h$. It is clear that $h_\eta\to h$ pointwise as $\eta\downarrow 0$.
Under \eqref{eq:endpoint_integral}, $\psi_\eta\in L^{p'}(0,1)$ and
$h_\eta(1)=h(1)$. When $p=1$, $\psi_\eta$ also satisfies
\eqref{eq:psi_endpoints}. Hence Proposition~\ref{prop:isotonic} applies
to $h_\eta$ for every $\eta>0$; write
$(P_{\lambda_\eta^*},\lambda_\eta^*)$ for an optimal pair.
The following theorem recovers the exact worst-case value of $\rho_h$ and, under upper
semicontinuity, its exact worst-case distribution, by taking limits of the worst-case $\rho_{h_\eta}$. We call this method a \emph{regularization}.

\begin{theorem}
\label{thm:regularization}
Suppose that \eqref{eq:endpoint_integral} holds.
\begin{enumerate}[(i)]
\item For $1\le p<\infty$,
\begin{equation}
S_h(\epsilon)
=\lim_{\eta\downarrow0}S_{h_\eta}(\epsilon)
=\lim_{\eta\downarrow0}
\max_{F\in\M_p(G_0,\epsilon)}\rho_{h_\eta}(F)<\infty.
\label{eq:reg_value}
\end{equation}
\item Suppose that $h=\widehat h$.
For every sequence $\eta\downarrow0$, there exists a subsequence $\eta_n\downarrow 0$, such that
\[
P^*_n\xrightarrow{\mathrm{w}} F^*\in\M_p(G_0,\epsilon)~~\mbox{as}~n\to\infty,
\qquad \mbox{and}\qquad\rho_h(F^*)=S_h(\epsilon),
\]
where we write $P^*_n=P_{\lambda^*_{\eta_n}}$.\footnote{Here, $\xrightarrow{\mathrm{w}}$ denotes weak convergence (convergence in distribution).}
\item Under the assumptions of (ii), if $h(1)\ne0$, then $W_p(F^*,G_0)=\epsilon$ and $W_p(P^*_n,F^*)\to0$ as $n\to\infty$.
\item Under the assumptions of (iii), if $p>1$, the worst-case distribution
$F^*$ is unique and $W_p(P_{\lambda^*_\eta},F^*)\to0$ for the full sequence as $\eta\downarrow 0$.
\end{enumerate}
\end{theorem}

Without $h=\widehat h$, the original supremum need not be attained;
the left-VaR criterion in Example~\ref{ex:vares} is an example.
The next result gives convergence rates of the regularized worst-case distortion riskmetric to the exact worst-case value. Here, additional assumptions of the Lipschitz continuity of the continuous part of $h$ are required.

\begin{proposition}
\label{prop:reg_rate}
Suppose that $h$ has finitely many discontinuities and its continuous
part is Lipschitz on $[0,1]$.
There exists $C\in\R_+$ such that, for all sufficiently small $\eta>0$,
the following hold.
\begin{enumerate}[(i)]
\item If $p>1$, then
\begin{equation}
|S_{h_\eta}(\epsilon)-S_h(\epsilon)|
\le C\eta^{1-1/p}.
\label{eq:reg_rate}
\end{equation}
\item Suppose that $h$ has a Lipschitz continuous derivative near $0$ and $1$.
If $p=1$ and \eqref{eq:endpoint_integral} holds, then
\begin{equation}
|S_{h_\eta}(\epsilon)-S_h(\epsilon)|\le C\eta.
\label{eq:reg_linear}
\end{equation}
For $p>1$, the same bound \eqref{eq:reg_linear} holds if
$h'(0+)=h'(1-)=0$.
\end{enumerate}
\end{proposition}

\begin{example}[RVaR and the VaR limit]\label{ex:rvar}
Let $0<\alpha<\beta<1$. For every $1\le p<\infty$,
\citet[Proposition~4]{LMWW22} gives
\begin{equation}\label{eq:wc_var}
S_{\VaR_\alpha}=S_{\VaRp_\alpha}
=\inf\left\{v\in\mathbb R:
\int_\alpha^1(v-G_0^{-1}(u))_+^p\,\mathrm d u\ge\epsilon^p
\right\}.
\end{equation}
A distribution attaining $S_{\VaRp_\alpha}$ has quantile
\begin{equation}
P_{\VaRp_\alpha}^{-1}(u)
=
\begin{cases}
G_0^{-1}(u), & 0<u\le\alpha,\\
G_0^{-1}(u)\vee S_{\VaRp_\alpha}, & \alpha<u<1.
\end{cases}
\label{eq:var_optimizer}
\end{equation}

When $p=1$,
for fixed $\alpha$, $G_0$, and $\epsilon$, the upper-VaR optimizer
$P_{\VaRp_\alpha}$ also attains $S_{\RVaR_{\alpha,\beta}}$ for every $\beta\in(\alpha,1)$.
Indeed, \eqref{eq:wc_var} gives $W_1(P_{\VaRp_\alpha},G_0)=\epsilon$.
With $\lambda^*=[(\beta\vee G_0(S_{\VaRp_\alpha}))-\alpha]^{-1}$,
monotonicity of $F^{-1}$ gives, by averaging over
$(\alpha,\beta\vee G_0(S_{\VaRp_\alpha}))$,
\[
\RVaR_{\alpha,\beta}(F)-\RVaR_{\alpha,\beta}(P_{\VaRp_\alpha})
\le\lambda^*\bigl[W_1(F,G_0)-\epsilon\bigr],
\qquad F\in\mathcal M_1.
\]
For every feasible $F$, the right-hand side is nonpositive. Hence
\begin{equation}
S_{\RVaR_{\alpha,\beta}}
=\RVaR_{\alpha,\beta}(P_{\VaRp_\alpha})
=\frac1{\beta-\alpha}\int_\alpha^\beta
\bigl(G_0^{-1}(u)\vee S_{\VaRp_\alpha}\bigr)\,\mathrm d u.
\label{eq:wc_rvar}
\end{equation}
This recovers \citet[Theorem~3]{CHLW23} without a support constraint. The RVaR optimizer need not be unique.

When $p>1$,
by Proposition~\ref{prop:isotonic}, the unique RVaR optimizer
$P_{\lambda^*}$ is given by \eqref{eq:reg_quantile}, with
$z_{a,b}(\lambda)$ the unique root of
\[
\lambda\int_a^b|z-G_0^{-1}(v)|^{p-2}(z-G_0^{-1}(v))\,\mathrm d v
=\frac{((b\wedge\beta)-(a\vee\alpha))_+}{\beta-\alpha},
\qquad 0\le a<b\le1,
\]
where $\lambda=\lambda^*>0$ is chosen so that
$W_p(P_{\lambda^*},G_0)=\epsilon$.
Here $P_{\lambda^*}$ and $P_{\VaRp_\alpha}$ need not coincide.

For $h=\widehat h_\alpha$ and
$0<\eta<(1-\alpha)/2$, the regularization in
\eqref{eq:reg_density} is exactly $h_\eta=h_{\alpha,\alpha+\eta}$.
Thus Theorem~\ref{thm:regularization} gives, for every $1\le p<\infty$,
\[
\lim_{\beta\downarrow\alpha}S_{\RVaR_{\alpha,\beta}}=S_{\VaRp_\alpha}=S_{\VaR_\alpha}.
\]
\end{example}

\subsection{Computing worst-case values on a finite grid}
\label{subsec:finite_grid}

We present a finite-grid method for computing the worst-case
value of the distortion riskmetric $\rho_h$.
Let $1\le p<\infty$, $G_0\in\mathcal M_p$, and $\epsilon>0$.
First suppose that $h$ is absolutely continuous and $\psi_h\in L^{p'}(0,1)$.
This includes every regularized distortion $h_\eta$ in
Theorem~\ref{thm:regularization}.
Choose $0=t_0<t_1<\cdots<t_N=1$, with
$I_i=(t_{i-1},t_i]$ and $\Delta_i=t_i-t_{i-1}$, and define
\[
G_N^{-1}(u)=\frac1{\Delta_i}\int_{I_i}G_0^{-1}(v)\,\mathrm d v,
\qquad u\in I_i.
\]
The decision vector $\mathbf x=(x_1,\ldots,x_N)$, with
$x_1\le\cdots\le x_N$, defines $P_{\mathbf x}$ by
$P_{\mathbf x}^{-1}(u)=x_i$ on $I_i$.
Thus $P_{\mathbf x}$ assigns mass $\Delta_i$ to $x_i$, with masses combined
when values coincide. If $G_0^{-1}$ is already constant on each $I_i$,
then $G_N=G_0$.
The finite-grid problem is
\begin{equation}
\begin{aligned}
S_{h,N}(\epsilon)=\max_{x_1\le\cdots\le x_N}
&\quad\sum_{i=1}^N x_i\int_{I_i}\psi_h(v)\,\mathrm d v\\
\text{subject to}
&\quad\sum_{i=1}^N\Delta_i|x_i-G_N^{-1}(t_i)|^p\le\epsilon^p.
\end{aligned}
\label{eq:grid_problem}
\end{equation}
This maximizes $\rho_h(P_{\mathbf x})$ subject to
$W_p(P_{\mathbf x},G_N)\le\epsilon$.

Let $h_N$ be the piecewise linear interpolation satisfying
$h_N(1-t_i)=h(1-t_i)$ for $0\le i\le N$. Then
\[
\psi_{h_N}(u)
=\frac1{\Delta_i}\int_{I_i}\psi_h(v)\,\mathrm d v
=\frac{h(1-t_{i-1})-h(1-t_i)}{\Delta_i},
\qquad u\in I_i,
\]
and $\rho_{h_N}(P_{\mathbf x})=\rho_h(P_{\mathbf x})$.
The following result identifies the corresponding Wasserstein problem and
justifies refining the grid.

\begin{proposition}
\label{prop:step_approx}
Under the preceding assumptions, the maximum in \eqref{eq:grid_problem}
is attained for every $N$, and
\[
S_{h,N}(\epsilon)
=\max_{F\in\mathcal M_p(G_N,\epsilon)}\rho_{h_N}(F)
\to S_h(\epsilon)
\qquad\mbox{as}~\max_i\Delta_i\to0.
\]
\end{proposition}

The finite-dimensional counterparts of the optimality conditions in
Proposition~\ref{prop:isotonic} give an explicit solution.
For $p>1$ and $\lambda>0$, consider
\begin{equation}
\operatorname*{minimize}_{x_1\le\cdots\le x_N}
\quad\sum_{i=1}^N
\left\{
\frac{\lambda\Delta_i}{p}|x_i-G_N^{-1}(t_i)|^p
-x_i\int_{I_i}\psi_h(v)\,\mathrm d v
\right\}.
\label{eq:grid_isotonic}
\end{equation}
For $1\le a\le b\le N$, let $z_{t_{a-1},t_b}(\lambda)$ be the unique root of
\begin{equation}
\lambda\sum_{i=a}^b\Delta_i|z-G_N^{-1}(t_i)|^{p-2}
(z-G_N^{-1}(t_i))
=\int_{t_{a-1}}^{t_b}\psi_h(v)\,\mathrm d v.
\label{eq:grid_block}
\end{equation}
The unique minimizer of \eqref{eq:grid_isotonic} is
\begin{equation}
x_i(\lambda)
=\max_{1\le a\le i}\ \min_{i\le b\le N}
z_{t_{a-1},t_b}(\lambda),
\qquad i=1,\ldots,N.
\label{eq:grid_minmax}
\end{equation}
If $h(1)\ne0$, choose $\lambda_N^*>0$ so that
\begin{equation}
\sum_{i=1}^N\Delta_i|x_i(\lambda_N^*)-G_N^{-1}(t_i)|^p
=\epsilon^p.
\label{eq:grid_radius}
\end{equation}
Then $x_i^*=x_i(\lambda_N^*)$ solves \eqref{eq:grid_problem}.
If $h(1)=0$, the optimal multiplier may be zero; the constrained
formulation \eqref{eq:grid_problem} remains valid.

For $p=2$, \eqref{eq:grid_isotonic} reduces to the weighted isotonic regression
\[
\mathbf x(\lambda)
=\operatorname*{argmin}_{x_1\le\cdots\le x_N}
\sum_{i=1}^N\Delta_i
\left(x_i-G_N^{-1}(t_i)-\frac{\psi_{h_N}(t_i)}{\lambda}\right)^2.
\]
It can be computed by the pool-adjacent-violators algorithm
\citep{BC90}, with bisection in $\lambda$ to enforce
\eqref{eq:grid_radius} when $h(1)\ne0$.
For $p=1$, introducing $d_i\ge0$ gives the linear program
\begin{equation}
\begin{aligned}
\operatorname*{maximize}_{x_1\le\cdots\le x_N,\ d_1,\ldots,d_N\ge0}
&\quad\sum_{i=1}^N x_i\int_{I_i}\psi_h(v)\,\mathrm d v\\
\text{subject to}
&\quad-d_i\le x_i-G_N^{-1}(t_i)\le d_i,\qquad i=1,\ldots,N,\\
&\quad\sum_{i=1}^N\Delta_i d_i\le\epsilon.
\end{aligned}
\label{eq:grid_lp}
\end{equation}

For a general $h$ satisfying \eqref{eq:endpoint_integral}, apply
\eqref{eq:grid_problem} with $h_\eta$ in place of $h$.
Proposition~\ref{prop:step_approx} and Theorem~\ref{thm:regularization}(i) give
\begin{equation}
S_h(\epsilon)
=\lim_{\eta\downarrow0}\lim_{N\to\infty}S_{h_\eta,N}(\epsilon),
\label{eq:isotonic_value}
\end{equation}
where the grids satisfy $\max_i\Delta_i\to0$.

\begin{remark}
\label{rem:h_positive}
In Proposition~\ref{prop:isotonic}, $h(1)\ne0$ may be replaced by
$h(t)>0$ for some $t\in(0,1)$, provided the stated attainment conditions
hold. For $h_\eta$, it suffices to show $h_\eta(t)>0$ at one point;
because $h_\eta(t)\to h(t)$, positivity of $h$ ensures this for all
sufficiently small $\eta$. The replacement by $h(t)>0$ also applies to
Theorem~\ref{thm:regularization}, and \eqref{eq:grid_radius} holds if
$h(1-t_i)>0$ for some $1\le i<N$.
Indeed, an optimal multiplier $\lambda^*=0$ in
Proposition~\ref{prop:isotonic} would imply
\[
S_h(\epsilon)=\sup_{F\in\mathcal M_p}\rho_h(F)
\ge\rho_h\bigl((1-t)\delta_0+t\delta_M\bigr)
=Mh(t)\longrightarrow\infty\qquad\mbox{as}~~M\to\infty,
\]
contradicting finiteness. Hence $\lambda^*>0$, and the optimality
conditions give $W_p(P_{\lambda^*},G_0)=\epsilon$.
The same argument applies to the general distortion and grid problems;
for $p>1$, strict convexity also gives uniqueness of the optimal distribution.
\end{remark}

\section{Closed-form approximations and error bounds}
\label{sec:gap_bounds}

Section \ref{sec:regularization} retains the original risk measure and recovers its
worst-case value exactly, but its auxiliary optimizers generally require the
isotonic calculation rather than a closed-form formula. We next study
the complementary tradeoff: accept an approximate worst-case distribution in
exchange for an explicit formula and a computable error guarantee.

Throughout this section, let $1\le p<\infty$, $G_0\in\mathcal M_p$,
$\epsilon>0$, and $h\in\mathcal H$ with $h^*\not\equiv0$.
We first use an explicit concave-envelope optimizer
as an approximate worst-case distribution and bound its loss under the
original risk measure by a gap computed at the benchmark.

\begin{proposition}
\label{prop:bench_gap}
Suppose that $h=\widehat h$. Whenever $F^*$ in
\eqref{eq:bench_shift} exists,
\begin{equation}
\rho_h(F^*)=\rho_h(G_0)+\epsilon\|\psi_{h^*}\|_{p'},
\label{eq:candidate_value}
\end{equation}
and
\begin{equation}
0\le S_h(\epsilon)-\rho_h(F^*)
\le\rho_{h^*}(G_0)-\rho_h(G_0).
\label{eq:gap_bound}
\end{equation}
Otherwise $p=1$, and the same statements hold with $\rho_h(F^*)$
replaced by $\lim_{n\to\infty}\rho_h(P_n)$, where $P_n$ is given by
\eqref{eq:bench_sequence}.
\end{proposition}
Together with \eqref{eq:bench_hstar}, these statements give
\begin{equation}
0\le S_{h^*}(\epsilon)-S_h(\epsilon)
\le\rho_{h^*}(G_0)-\rho_h(G_0).
\label{eq:gap_h_bound}
\end{equation}
\begin{example}[VaR, RVaR, and ES]
\label{ex:vares_gap}
Fix $0<\alpha<\beta<1$ and $1\le p<\infty$, and consider the worst-case problems defined in
\eqref{eq:wc_vares}. By Proposition~\ref{prop:wc_hstar}, the
concave-envelope optimizer for $S_{\ES_\alpha}$ is
\[
(F^*)^{-1}(u)=G_0^{-1}(u)
+\epsilon(1-\alpha)^{-1/p}\mathbf1_{(\alpha,1)}(u),
\qquad 0<u<1.
\]
It satisfies $W_p(F^*,G_0)=\epsilon$ and
$\ES_\alpha(F^*)=S_{\ES_\alpha}
=\ES_\alpha(G_0)+\epsilon(1-\alpha)^{-1/p}.$
For RVaR, Proposition~\ref{prop:bench_gap} gives
\begin{equation}
\begin{aligned}
&\RVaR_{\alpha,\beta}(F^*)
=\RVaR_{\alpha,\beta}(G_0)+\epsilon(1-\alpha)^{-1/p},\\
&0\le S_{\RVaR_{\alpha,\beta}}-\RVaR_{\alpha,\beta}(F^*)
\le\ES_\alpha(G_0)-\RVaR_{\alpha,\beta}(G_0).
\end{aligned}
\label{eq:rvar_gap}
\end{equation}
For the VaR--ES pair, we can go beyond the bound in
Proposition~\ref{prop:bench_gap} and calculate the exact gap from using
$F^*$ to approximate the worst-case VaR. Evaluating VaR and $\VaRp$ at
$F^*$ gives, for $\epsilon>0$, the gaps
\[
\begin{aligned}
&\Delta_{\VaRp_\alpha}(\epsilon)
:=S_{\VaRp_\alpha}-\VaRp_\alpha(F^*)
=S_{\VaRp_\alpha}-G_0^{-1+}(\alpha)-\epsilon(1-\alpha)^{-1/p}\le\ES_\alpha(G_0)-\VaRp_\alpha(G_0),\\
&\Delta_{\VaR_\alpha}(\epsilon)
:=S_{\VaR_\alpha}-\VaR_\alpha(F^*)
=S_{\VaRp_\alpha}-G_0^{-1}(\alpha)\le\ES_\alpha(G_0)-\VaR_\alpha(G_0)
+\epsilon(1-\alpha)^{-1/p},
\end{aligned}
\]
where $S_{\VaR_\alpha}=S_{\VaRp_\alpha}$ is given by \eqref{eq:wc_var}.
As $\epsilon\downarrow0$, \eqref{eq:wc_var} gives
$S_{\VaRp_\alpha}\downarrow G_0^{-1+}(\alpha)$, and hence
\[
\lim_{\epsilon\downarrow0}\Delta_{\VaRp_\alpha}(\epsilon)=0
\qquad\mbox{and}\qquad
\lim_{\epsilon\downarrow0}\Delta_{\VaR_\alpha}(\epsilon)
=G_0^{-1+}(\alpha)-G_0^{-1}(\alpha).
\]
The limiting gap for left VaR is zero precisely when the benchmark
quantile has no jump at $\alpha$.
\end{example}

Proposition~\ref{prop:bench_gap} demonstrates the closed-form expression of the nonconvex distortion riskmetric $\rho_h$ with the concave-envelope optimizer $F^*$. Further, \eqref{eq:gap_bound} and \eqref{eq:gap_h_bound} show that the gap between the original worst-case value $S_h(\epsilon)$ and its approximated value $\rho_h(F^*)$ shares the same upper bound as the gap between $S_h(\epsilon)$ and its concave-envelope counterpart $S_{h^*}(\epsilon)$, and the bound is given by the gap between $\rho_h$ and $\rho_{h^*}$ under the reference distribution $G_0$.
This is useful in practice, where one can estimate the guaranteed bound of the gap using the known reference distribution $G_0$ before we use $F^*$ as an approximate worst-case distribution for the original problem $S_h(\epsilon)$.

However, the benchmark bound in Proposition~\ref{prop:bench_gap}
does not depend on $\epsilon$, so this bound alone does not certify
improved accuracy as the radius changes. We next construct another
explicit approximate worst-case distribution for $\rho_h$ from
the reference distribution $G_0$, its $\mathcal I_h$-concentration
$G_0^{\mathcal I_h}$, and $r^*$ in \eqref{eq:bench_shift}.
For $\epsilon>0$, define $P_\epsilon$ by the left-continuous modification of
\begin{equation}
P_\epsilon^{-1}(u)
=G_0^{-1}(u)+\epsilon
\frac{(G_0^{\mathcal I_h})^{-1}(u)-G_0^{-1}(u)+r^*(u)}
{\|(G_0^{\mathcal I_h})^{-1}-G_0^{-1}+r^*\|_p},
\qquad 0<u<1.
\label{eq:radial_candidate}
\end{equation}
\begin{theorem}
\label{thm:radius_bound}
Suppose that $h=\widehat h$ and the envelope optimizer $F^*$ in
\eqref{eq:bench_shift} exists.
\begin{enumerate}[(i)]
\item We have $W_p(P_\epsilon,G_0)=\epsilon$,
\begin{equation}
\rho_h(P_\epsilon)
=\rho_h(G_0)+\frac{\epsilon\bigl[\rho_{h^*}(G_0)-\rho_h(G_0)
+\epsilon\|\psi_{h^*}\|_{p'}\bigr]}
{\|(G_0^{\mathcal I_h})^{-1}-G_0^{-1}+r^*\|_p},
\label{eq:radial_value}
\end{equation}
and consequently,
\begin{equation}
\begin{aligned}
0\le S_h(\epsilon)-\rho_h(P_\epsilon)
\le\rho_{h^*}(G_0)+\epsilon\|\psi_{h^*}\|_{p'}-\rho_h(P_\epsilon).
\end{aligned}
\label{eq:radial_gap}
\end{equation}
For $p>1$, the upper bound in \eqref{eq:radial_gap} is
$O\!\left(\epsilon^{-\min\{p-1,1\}}\right)$ as $\epsilon\to\infty$.
Moreover, $S_h(\epsilon)-\rho_h(P_\epsilon)\to0$ as $\epsilon\downarrow0$.
\item If $h$ is absolutely continuous and $\psi_h\in L^{p'}(0,1)$, then
\begin{equation}
0\le S_h(\epsilon)-\rho_h(P_\epsilon)
\le\rho_h(G_0)+\epsilon\|\psi_h\|_{p'}-\rho_h(P_\epsilon)
=O(\epsilon),\qquad \epsilon\downarrow0.
\label{eq:radial_holder}
\end{equation}
\end{enumerate}
\end{theorem}

The construction and bounds in Theorem~\ref{thm:radius_bound} require no
optimization or root finding. For general $h$, the upper bound in
\eqref{eq:radial_gap} tends to zero as $\epsilon\to\infty$ when $p>1$, but it need not vanish as $\epsilon\downarrow0$, although
the gap tends to zero. Under part~(ii), the minimum of the upper bounds in \eqref{eq:radial_gap} and \eqref{eq:radial_holder} tends to zero as $\epsilon\downarrow0$ for every $1\le p<\infty$, and also as $\epsilon\to\infty$ when $p>1$.
The linear small-radius rate cannot hold for all distortions.
For example, let $\rho_h=\VaRp_{1/2}$ and $G_0=\U[0,1]$.
Equation~\eqref{eq:wc_var} gives
$S_h(\epsilon)=1/2+(p+1)^{1/(p+1)}\epsilon^{p/(p+1)}$ for sufficiently
small $\epsilon$, whereas \eqref{eq:radial_candidate} gives
$\rho_h(P_\epsilon)=1/2+O(\epsilon)$. Hence
\[
S_h(\epsilon)-\rho_h(P_\epsilon)
\sim(p+1)^{1/(p+1)}\epsilon^{p/(p+1)},
\qquad\epsilon\downarrow0.
\]

If $p=1$ and $F^*$ does not exist, replace $r^*$ in
\eqref{eq:radial_candidate} by $P_n^{-1}-G_0^{-1}$ from
\eqref{eq:bench_sequence}. The resulting $P_{\epsilon,n}$ satisfies
$W_1(P_{\epsilon,n},G_0)=\epsilon$ and
\[
0\le S_h(\epsilon)-\rho_h(P_{\epsilon,n})
\le\rho_{h^*}(G_0)+\epsilon\|\psi_{h^*}\|_\infty
-\rho_h(P_{\epsilon,n}).
\]

\begin{example}
\label{ex:radius_p1}
Large-radius decay in Theorem \ref{thm:radius_bound} can fail when $p=1$, even if $F^*$ exists.
Let $p=1$ and
\[
G_0^{-1}(u)=
\begin{cases}
0,&0<u\le1/4,\\
2,&1/4<u\le1/2,\\
3,&1/2<u<1,
\end{cases}
\qquad
\psi_h(u)=
\begin{cases}
8/5,&0<u<1/4,\\
0,&1/4<u<1/2,\\
6/5,&1/2<u<1.
\end{cases}
\]
Here $h$ is increasing, $h(1)=1$,
$\psi_{h^*}=4/5$ on $(0,1/2)$, $\psi_{h^*}=6/5$ on $(1/2,1)$,
and $\mathcal I_h=\{(0,1/2)\}$.
For every $\epsilon\ge1/2$, we have
\[
S_h(\epsilon)=2+\frac65\epsilon,
\qquad S_h(\epsilon)-\rho_h(P_\epsilon)
=\frac{1+4\epsilon}{5(1+2\epsilon)}
\to\frac25\quad\text{as }\epsilon\to\infty.
\]
Indeed, the majorant $(h+h^*)/2$ has benchmark risk value $2$ and
density norm $6/5$, giving the stated upper bound on $S_h$.
The distribution with quantile $2$ on $(0,1/2]$ and $2+2\epsilon$
on $(1/2,1)$ has distance $\epsilon$ from $G_0$ and attains this bound.
Substitution of $r^*=2\epsilon\mathbf1_{(1/2,1)}$ in
\eqref{eq:radial_candidate} gives the asserted gap.
\end{example}
\section{Application to distributionally robust portfolio selection}
\label{sec:portfolio}

\subsection{Portfolio selection problem and elliptical benchmark}
\label{subsec:portfolio_model}

In this section, we study a distributionally robust portfolio selection problem focusing on the outer optimization.
Let $\mathbf X\in\X^d$ denote asset losses and let $\mathcal W\subseteq \R^d$ be a
nonempty set of feasible portfolio weights. A typical choice is the \emph{long-only
simplex} $\{\bm\omega\ge0:\mathbf 1^\top\bm\omega=1\}$, possibly with
additional return or exposure constraints. More generally, More generally, $\mathcal{W}$ is \emph{second-order-cone representable}
if it is the projection onto portfolio coordinates of a set
defined by finitely many second-order-cone and affine
constraints.\footnote{Formally, a set $\mathcal{W}\subseteq\mathbb{R}^d$ is \emph{second-order-cone representable} if there exist $n\ge d$, $N\ge 1$, $q\ge0$, $m_i\ge1$, $A_i\in\mathbb{R}^{m_i\times n}$, 
$\mathbf{b}_i\in\mathbb{R}^{m_i}$,
$\mathbf{c}_i\in\mathbb{R}^n$,
$\gamma_i\in\mathbb{R}$, $i=1,\ldots,N$,
$E\in\mathbb{R}^{q\times n}$, and $\mathbf{f}\in\mathbb{R}^q$,
such that
\[
\mathcal{W}
=\left\{
\boldsymbol{\omega}\in\mathbb{R}^d:
\begin{array}{l}
\text{there exists }\mathbf{x}\in\mathbb{R}^n
\text{ with }\mathbf{x}_{1:d}=\boldsymbol{\omega},\\[2pt]
E\mathbf{x}=\mathbf{f},\quad
\|A_i\mathbf{x}+\mathbf{b}_i\|_2
\le \mathbf{c}_i^\top\mathbf{x}+\gamma_i,\quad i=1,\ldots,N
\end{array}
\right\},
\]
where $\mathbf{x}_{1:d}$ denotes the first $d$ coordinates of
$\mathbf{x}$.}
Minimizing a linear objective over such a set is a \emph{second-order cone
program (SOCP)}; see \citet{LVBL98}.
Fix $1\le p<\infty$, $r\ge1$,
$\epsilon>0$, and $h\in\mathcal H$. With $V_h(\bm\omega)$ defined
in \eqref{eq:wc_decision}, the outer problem is
$\inf_{\bm\omega\in\mathcal W}V_h(\bm\omega)$.

Suppose that the benchmark is elliptical,
$P_0=\mathcal E(\bm\mu,\bm\Sigma,\varphi)$, with $\bm\Sigma\succeq0$ and standardized
one-dimensional projection $P_\varphi=\mathcal E(0,1,\varphi)\in\mathcal M_p$.
We have
\[
\bm\omega^\top\mathbf X_0
\overset{\mathrm d}{=}\bm\omega^\top\bm\mu+s_{\bm\omega}Z_\varphi,
\qquad s_{\bm\omega}=\sqrt{\bm\omega^\top\bm\Sigma\bm\omega},
\quad F_{Z_\varphi}=P_\varphi;
\]
see, for example, \citet[Section~7]{LMWW22}.
Throughout this section, the worst-case value $S_h(\epsilon)$ in
\eqref{eq:wc_distortion} uses the fixed benchmark $G_0=P_\varphi$.
For $s_{\bm\omega}>0$, the projection result in \eqref{eq:linear_projection} and location--scale
covariance yield
\begin{equation}
V_h(\bm\omega)
=h(1)\bm\omega^\top\bm\mu
+s_{\bm\omega}S_h\!\left(
\frac{\epsilon\|\bm\omega\|_{r'}}{s_{\bm\omega}}\right).
\label{eq:elliptical_inner}
\end{equation}
Thus every portfolio uses the same worst-case distortion riskmetric at a different radius.
For $s_{\bm\omega}=0$, the projected benchmark is the point mass
$\delta_{\bm\omega^\top\bm\mu}$; in particular, $V_h(0)=0$.

\subsection{Outer optimization by direct convexification}
\label{subsec:elliptical}

Proposition~\ref{prop:wc_hstar} gives the envelope objective
\begin{equation}
V_{h^*}(\bm\omega)
=h(1)\bm\omega^\top\bm\mu+\rho_{h^*}(P_\varphi)s_{\bm\omega}
+\epsilon\|\psi_{h^*}\|_{p'}\|\bm\omega\|_{r'},
\label{eq:elliptical_value}
\end{equation}
where $\rho_{h^*}(P_\varphi)\ge0$ by symmetry of $P_\varphi$ and convexity of
$\rho_{h^*}$. The result below determines conditions under which minimizing
\eqref{eq:elliptical_value} also solves the original outer problem.

\begin{proposition}
\label{prop:outer_flatness}
Suppose that $h=\widehat h$, $h^*\not\equiv0$, $\bm\Sigma\succ0$, and
$\mathcal W$ is compact with $0\notin\mathcal W$.
\begin{enumerate}[(i)]
\item If $p>1$, then, for every $\epsilon>0$,
\begin{equation}
\inf_{\bm\omega\in\mathcal W}V_h(\bm\omega)
=\min_{\bm\omega\in\mathcal W}V_{h^*}(\bm\omega)
\quad\Longleftrightarrow\quad
P_\varphi^{-1}\text{ is flat on }\mathcal I_h.
\label{eq:outer_flatness}
\end{equation}
\item If $p=1$, flatness of $P_\varphi^{-1}$ on $\mathcal I_h$ implies
$\inf_{\bm\omega\in\mathcal W}V_h(\bm\omega)
=\min_{\bm\omega\in\mathcal W}V_{h^*}(\bm\omega)$ for every $\epsilon>0$.
If $P_\varphi^{-1}$ is not flat on $\mathcal I_h$, then
$\inf_{\bm\omega\in\mathcal W}V_h(\bm\omega)
<\min_{\bm\omega\in\mathcal W}V_{h^*}(\bm\omega)$ for all sufficiently small
$\epsilon>0$.
\end{enumerate}
Under flatness, $V_h(\bm\omega)=V_{h^*}(\bm\omega)$ for all
$\bm\omega\in\mathcal W$ and $\epsilon>0$, so the minimizers also coincide.
\end{proposition}

The envelope minimizer can be computed as described in
Section~\ref{subsec:outer_gap}. The direct convexification method of Proposition \ref{prop:outer_flatness} only works under specific setups. It is clear that flatness does not hold for $P_\varphi$ with strictly increasing $P^{-1}_\varphi$ (e.g., Gaussian and nondegenerate Student-$t$ projections). 

\subsection{Exact outer optimization by regularization}
\label{subsec:outer_reg}

Next, we adopt the idea of Theorem~\ref{thm:regularization} to solve the exact optimal value of the robust portfolio optimization problem by taking the limit of regularized problems with absolutely continuous distortion functions.

\begin{proposition}
\label{prop:outer_reg}
Suppose that \eqref{eq:endpoint_integral} holds, $\bm\Sigma\succ0$, and
$\mathcal W$ is compact. Let $h_\eta(t)=\int^1_{1-t}\psi_\eta(v)\,\d v$, $0\le t\le 1$, with density $\psi_\eta$ given by \eqref{eq:reg_density}.
\begin{enumerate}[(i)]
\item For every $\eta>0$ and $1\le p<\infty$,
\begin{equation}
\min_{\bm\omega\in\mathcal W}V_{h_\eta}(\bm\omega)
=\min_{\delta>0,\,\bm\omega\in\mathcal W}
\Bigl\{h(1)\bm\omega^\top\bm\mu
+\bigl[S_{h_\eta}(\delta)-\delta\,\partial_+S_{h_\eta}(\delta)\bigr]
s_{\bm\omega}
+\epsilon\,\partial_+S_{h_\eta}(\delta)\|\bm\omega\|_{r'}\Bigr\},
\label{eq:reg_outer_solution}
\end{equation}
where $\partial_+S_{h_\eta}$ is the right derivative with respect to the
radius. Any minimizing pair $(\delta_\eta^*,\bm\omega_\eta^*)$ on the
right-hand side satisfies
$\bm\omega_\eta^*\in
\operatorname*{argmin}_{\bm\omega\in\mathcal W}V_{h_\eta}(\bm\omega)$.
\item Uniform convergence $\sup_{\bm\omega\in\mathcal W}
|V_{h_\eta}(\bm\omega)-V_h(\bm\omega)|\to0$ holds as
$\eta\downarrow0$.
For any choice $\bm\omega_\eta^*\in
\operatorname*{argmin}_{\bm\omega\in\mathcal W}V_{h_\eta}(\bm\omega)$,
if $\bm\omega_\eta^*\to\bm\omega^*$ as $\eta\downarrow0$, then
$\bm\omega^*\in
\operatorname*{argmin}_{\bm\omega\in\mathcal W}V_h(\bm\omega)$.
Moreover, if $\operatorname*{argmin}_{\bm\omega\in\mathcal W}V_h(\bm\omega)$
is a singleton $\{\bm\omega^*\}$, then
$\bm\omega_\eta^*\to\bm\omega^*$ as $\eta\downarrow0$.
\end{enumerate}
\end{proposition}
As a direct consequence of the uniform convergence in Proposition \ref{prop:outer_reg}, the regularized minima converge to the original problem,
$$\min_{\bm\omega\in\mathcal W}V_{h_\eta}(\bm\omega)
\to\min_{\bm\omega\in\mathcal W}V_h(\bm\omega), ~~\mbox{as}~ \eta\downarrow 0.$$
For $r=2$, $\rho_{h_\eta}(P_\varphi)\ge0$, and second-order-cone
representable $\mathcal W$, the portfolio minimization at each fixed
$\delta$ in \eqref{eq:reg_outer_solution} is the following SOCP:
\begin{equation}
\begin{aligned}
\underset{\bm\omega,t,s}{\operatorname{minimize}}\quad
&h(1)\bm\omega^\top\bm\mu
+\bigl[S_{h_\eta}(\delta)-\delta\,\partial_+S_{h_\eta}(\delta)\bigr]t
+\epsilon\,\partial_+S_{h_\eta}(\delta)s\\
\text{subject to}\quad
&\|\bm\Sigma^{1/2}\bm\omega\|_2\le t,\qquad
\|\bm\omega\|_2\le s,\qquad \bm\omega\in\mathcal W.
\end{aligned}
\label{eq:reg_socp}
\end{equation}

\paragraph{Explicit weights in the spherical case.}
When $r=2$, $\bm\Sigma=\sigma^2\mathbf I_d$ with $\sigma>0$, and $\mathcal W$ is the long-only simplex $\{\bm\omega\ge0:\mathbf 1^\top\bm\omega=1\}$, the
regularized outer problem by \eqref{eq:elliptical_inner} is
\begin{equation}
\min_{\bm\omega\in\mathcal W}V_{h_\eta}(\bm\omega)=\min_{\bm\omega\ge0,\ \mathbf 1^\top\bm\omega=1}
\left\{h(1)\bm\omega^\top\bm\mu
+\sigma S_{h_\eta}(\epsilon/\sigma)\|\bm\omega\|_2\right\}.
\label{eq:spherical_outer}
\end{equation}
If $S_{h_\eta}(\epsilon/\sigma)>0$, strict convexity and the first-order
optimality conditions give the unique optimizer
$\bm\omega_\eta^*=(\omega_{\eta,1}^*,\dots,\omega_{\eta,d}^*)$ with
\begin{equation*}
\omega_{\eta,i}^*
=\frac{(\ell_\eta- h(1)\mu_i)_+}
{\sum_{j=1}^d(\ell_\eta-h(1)\mu_j)_+},\qquad i=1,\ldots,d,
\end{equation*}
where $\ell_\eta>\min_i h(1)\mu_i$ is the unique solution of
\begin{equation*}
\sum_{i=1}^d(\ell-h(1)\mu_i)_+^2
=\sigma^2 S_{h_\eta}(\epsilon/\sigma)^2.
\end{equation*}
The minimum in \eqref{eq:spherical_outer} is $\min_{\bm\omega\in\mathcal W}V_{h_\eta}(\bm\omega)=\ell_\eta$.
If $S_{h_\eta}(\epsilon/\sigma)\le0$, then $\|\bm\omega\|_2\le1$ on
$\mathcal W$ gives
\[
\min_{\bm\omega\in\mathcal W}V_{h_\eta}(\bm\omega)
=\min_{1\le i\le d}h(1)\mu_i+\sigma S_{h_\eta}(\epsilon/\sigma).
\]
An optimal portfolio is any vertex corresponding to an asset minimizing
$h(1)\mu_i$. These are the only optimizers if
$S_{h_\eta}(\epsilon/\sigma)<0$. If $S_{h_\eta}(\epsilon/\sigma)=0$,
every mixture of minimizing assets is optimal; in particular, all
portfolios tie when also $h(1)=0$.

Under \eqref{eq:endpoint_integral}, applying
Proposition~\ref{prop:outer_reg}(ii) to these regularized optimizers yields
$$\min_{\bm\omega\in\mathcal W}V_h(\bm\omega)
=\lim_{\eta\downarrow0}V_{h_\eta}(\bm\omega_\eta^*).$$
Moreover, if $\bm\omega_\eta^*\to\bm\omega^*$ as $\eta\downarrow0$, then
$\bm\omega^*\in
\operatorname*{argmin}_{\bm\omega\in\mathcal W}V_h(\bm\omega)$.

\subsection{Closed-form outer approximation}
\label{subsec:outer_gap}

We next use the minimizer of the concave-envelope objective
to approximate the outer problem. The benchmark bound
\eqref{eq:gap_h_bound} gives the following guarantee.

\begin{proposition}
\label{prop:outer_gap}
Let $h\in\mathcal H$ with $h=\widehat h$ and $h^*\not\equiv0$, and let
$\mathcal W$ be compact. For $1\le p<\infty$, choose the envelope portfolio
$\bm\omega^{\mathrm{env}}\in
\operatorname*{argmin}_{\bm\omega\in\mathcal W}V_{h^*}(\bm\omega)$,
where $V_{h^*}$ is given by \eqref{eq:elliptical_value}. Then
\begin{equation}
\begin{aligned}
0\le V_h(\bm\omega^{\mathrm{env}})
-\inf_{\bm\omega\in\mathcal W}V_h(\bm\omega)&\le V_{h^*}(\bm\omega^{\mathrm{env}})-
\min_{\bm\omega\in\mathcal W}
\{V_{h^*}(\bm\omega)-\bigl[\rho_{h^*}(P_\varphi)-\rho_h(P_\varphi)\bigr]
s_{\bm\omega}\}\\
&\le\bigl(\rho_{h^*}(P_\varphi)-\rho_h(P_\varphi)\bigr)
\max_{\bm\omega\in\mathcal W}s_{\bm\omega}.
\end{aligned}
\label{eq:decision_bound}
\end{equation}
\end{proposition}
Equation~\eqref{eq:gap_h_bound} gives the pointwise bound
\begin{equation}
0\le V_{h^*}(\bm\omega)-V_h(\bm\omega)
\le\bigl(\rho_{h^*}(P_\varphi)-\rho_h(P_\varphi)\bigr)s_{\bm\omega}.
\label{eq:elliptical_gap}
\end{equation}
Consequently, the first upper bound in \eqref{eq:decision_bound} holds
for every feasible $\bm\omega^{\mathrm{env}}\in\mathcal W$, and the final inequality
uses its optimality for $V_{h^*}$.

For coherent distortion risk measures, the worst-case formula is
given by \citet[Proposition~4]{LMWW22}, and the portfolio reformulation
\eqref{eq:elliptical_value} appears in
\citet{MWW25}.
In our setting, $h^*$ need not be monotone or normalized.
Proposition~\ref{prop:wc_hstar} handles the signed density $\psi_{h^*}$,
including $p=1$. With the location term $h(1)\bm\omega^\top\bm\mu$, the same
convex and conic formulations follow because $\rho_{h^*}$ remains convex
and symmetry gives $\rho_{h^*}(P_\varphi)\ge0$.
These remaining modifications are minor, and their routine proofs are omitted.

For $r=2$ and second-order-cone representable $\mathcal W$, the envelope
portfolio $\bm\omega^{\mathrm{env}}$ is obtained from the following SOCP:
\begin{equation}
(\bm\omega^{\mathrm{env}},t^*,s^*)\in
\operatorname*{argmin}_{\substack{
\bm\omega\in\mathcal W,\ t,s\in\R\\
\|\bm\Sigma^{1/2}\bm\omega\|_2\le t,\ \|\bm\omega\|_2\le s}}
\left\{h(1)\bm\omega^\top\bm\mu+\rho_{h^*}(P_\varphi)t
+\epsilon\|\psi_{h^*}\|_{p'}s\right\}.
\label{eq:elliptical_socp}
\end{equation}
Conversely, every minimizer $\bm\omega^{\mathrm{env}}$ of
\eqref{eq:elliptical_value} extends to an optimal triple by taking
$t^*=s_{\bm\omega^{\mathrm{env}}}$ and $s^*=\|\bm\omega^{\mathrm{env}}\|_2$.
For $r=2$, $\bm\Sigma=\sigma^2\mathbf I_d$ with $\sigma>0$, and $\mathcal W$ equal
to the long-only simplex, the spherical formulas in
Section~\ref{subsec:outer_reg} give the unique envelope optimizer
\[
\begin{aligned}
\omega_i^{\mathrm{env}}
=\frac{(\ell-h(1)\mu_i)_+}{\sum_{j=1}^d(\ell-h(1)\mu_j)_+},
\qquad i=1,\ldots,d,\qquad \sum_{i=1}^d(\ell-h(1)\mu_i)_+^2
=\bigl[\sigma\rho_{h^*}(P_\varphi)
+\epsilon\|\psi_{h^*}\|_{p'}\bigr]^2.
\end{aligned}
\]
Here $\ell>\min_i h(1)\mu_i$ is uniquely determined by the second equation,
and $\min_{\bm\omega\in\mathcal W}V_{h^*}(\bm\omega)=\ell$.

Finally, we use the candidate in Theorem~\ref{thm:radius_bound} to approximate
the portfolio objective. With $G_0=P_\varphi$, let $P_\delta$ be given by
\eqref{eq:radial_candidate} at radius $\delta>0$, and define
\begin{equation}
\widehat V_h(\bm\omega)
=h(1)\bm\omega^\top\bm\mu
+s_{\bm\omega}\rho_h\!\left(P_{\epsilon\|\bm\omega\|_{r'}/s_{\bm\omega}}\right),
\qquad s_{\bm\omega}>0.
\label{eq:radial_outer}
\end{equation}
Equation~\eqref{eq:radial_value} gives an explicit formula for this
objective.

\begin{proposition}
\label{prop:outer_radial}
Suppose that $h=\widehat h$, $h^*\not\equiv0$, $\bm\Sigma\succ0$, and
$\mathcal W$ is compact with $0\notin\mathcal W$.
Assume that the envelope optimizer in \eqref{eq:bench_shift} exists
for $G_0=P_\varphi$. For $1\le p<\infty$, choose the approximate portfolio
$\bm\omega^{\mathrm{app}}\in\operatorname*{argmin}_{\bm\omega\in\mathcal W}
\widehat V_h(\bm\omega)$.
\begin{enumerate}[(i)]
\item We have
\begin{equation}
0\le V_h(\bm\omega^{\mathrm{app}})-\inf_{\bm\omega\in\mathcal W}V_h(\bm\omega)
\le V_{h^*}(\bm\omega^{\mathrm{app}})-\widehat V_h(\bm\omega^{\mathrm{app}}).
\label{eq:radial_decision}
\end{equation}
For $p>1$, the upper bound is
$O\!\left(\epsilon^{-\min\{p-1,1\}}\right)$ as
$\epsilon\to\infty$.
Moreover,
$V_h(\bm\omega^{\mathrm{app}})-\inf_{\bm\omega\in\mathcal W}V_h(\bm\omega)
\to0$ as $\epsilon\downarrow0$.
\item If $h$ is absolutely continuous and $\psi_h\in L^{p'}(0,1)$, then
\begin{equation}
\begin{aligned}
V_h(\bm\omega^{\mathrm{app}})-\inf_{\bm\omega\in\mathcal W}V_h(\bm\omega)
&\le h(1)(\bm\omega^{\mathrm{app}})^\top\bm\mu
+\rho_h(P_\varphi)s_{\bm\omega^{\mathrm{app}}}
+\epsilon\|\psi_h\|_{p'}\|\bm\omega^{\mathrm{app}}\|_{r'}
-\widehat V_h(\bm\omega^{\mathrm{app}})\\
&\le\epsilon\|\psi_h\|_{p'}\|\bm\omega^{\mathrm{app}}\|_{r'}
=O(\epsilon),\qquad \epsilon\downarrow0.
\end{aligned}
\label{eq:radial_decision_ac}
\end{equation}
\end{enumerate}
\end{proposition}

The explicit objective \eqref{eq:radial_outer} need not be convex.
In the spherical case of Section~\ref{subsec:outer_reg}, it becomes
\[
\widehat V_h(\bm\omega)
=h(1)\bm\omega^\top\bm\mu
+\sigma\rho_h(P_{\epsilon/\sigma})\|\bm\omega\|_2.
\]
The weight formulas there therefore apply with
$\rho_h(P_{\epsilon/\sigma})$ in place of
$S_{h_\eta}(\epsilon/\sigma)$.

\section{Numerical experiments}
\label{sec:numerics}

\renewcommand{\topfraction}{0.90}
\renewcommand{\bottomfraction}{0.80}
\renewcommand{\textfraction}{0.10}
\renewcommand{\floatpagefraction}{0.80}
\setlength{\textfloatsep}{12pt plus 2pt minus 2pt}
\setlength{\intextsep}{12pt plus 2pt minus 2pt}
\captionsetup[figure]{font={small,stretch=1.1},skip=6pt}

We first examine cases when the explicit worst-case approximations meet
a prescribed error tolerance. The VaR and RVaR experiments use
Proposition~\ref{prop:bench_gap} and Example~\ref{ex:vares_gap}.
The \cite{TK92} (TK) experiment uses both parts of Theorem~\ref{thm:radius_bound}.
We then compare portfolio
decisions for the model in Section~\ref{sec:portfolio}, as in the
application of \citet{CHLW23}. Their use of aggregated moment and
support information gives an upper approximation. Our projection
identity \eqref{eq:linear_projection} is exact, so the TK portfolio
experiment focuses on the decision losses from approximating the
objective. We then use
the inter-quantile range to examine convergence under regularization,
as in Theorem~\ref{thm:regularization}(i) and
Proposition~\ref{prop:outer_reg}(ii).
The code for these experiments is available on GitHub.\footnote{\url{https://github.com/Murphy116/Robust-distortion-riskmetrics-under-Wasserstein-ambiguity}}

\subsection{Explicit approximations with error tolerance}
\label{subsec:var_numerics}

We study the error from using the explicit distribution $F^*$ in
Example~\ref{ex:vares_gap}, with $p=2$. Here and in
Subsection~\ref{subsec:tk_numerics}, reference losses are standardized
to have mean zero and variance one. Gaps, bounds, and Wasserstein
radii are measured in units of the original reference loss standard
deviation (SD).
Before choosing the probability levels, we set a common tolerance
of $0.05$, or five percent of that standard deviation, for illustration.
The radii describe the size of the Wasserstein ball and are not
calibrated from data.

For VaR, $\VaRp$, and RVaR, we first check the bounds in
Proposition~\ref{prop:bench_gap} and Example~\ref{ex:vares_gap}
against this tolerance. For $\VaRp$ and RVaR, the bound depends only
on the reference distribution. For VaR, we also include the quantile
jump and radius term shown below. In the selected examples, the bounds
guarantee errors below $0.05$ for $\VaRp$ and RVaR at every radius,
and for VaR when $0\le\epsilon\le0.001$. Thus the explicit
approximation meets the required accuracy in these cases without
computing $S_h$. We then calculate the actual errors to show how they
change with $\epsilon$ and how much smaller they can be than the
bounds. If a bound exceeds $0.05$, the bound alone does not establish
whether the approximation meets the tolerance.

\subsubsection{VaR and \texorpdfstring{$\bm{\VaRp}$}{VaR+}}

We compare the two VaR conventions at $\alpha=0.99$. Before
standardization, the first reference distribution is $\U[0,1]$ and
the second is
\begin{equation}\label{eq:num_jump}
0.99\,\U[0,0.99]+0.01\,\U[0.996,1.006].
\end{equation}
Both have continuous cdfs. The second has a gap in its support, so its
quantile jumps at $\alpha=0.99$. After standardization, the jump
$G_0^{-1+}(\alpha)-G_0^{-1}(\alpha)$ is approximately $0.02078$.
The reference bounds $\ES_\alpha(G_0)-\VaRp_\alpha(G_0)$ are
$0.01732$ and $0.01731$, respectively, both below the chosen tolerance.

We use the distribution $F^*$ from Example~\ref{ex:vares_gap} and
calculate $S_{\VaRp_\alpha}$ from \eqref{eq:wc_var}. For $p=2$ and
$\epsilon>0$, Proposition~\ref{prop:bench_gap} and
Example~\ref{ex:vares_gap} give
\[
\begin{aligned}
0\le\Delta_{\VaRp_\alpha}(\epsilon)
&=S_{\VaRp_\alpha}-\VaRp_\alpha(F^*)
\le\ES_\alpha(G_0)-\VaRp_\alpha(G_0),\\
0\le\Delta_{\VaR_\alpha}(\epsilon)
&=S_{\VaR_\alpha}-\VaR_\alpha(F^*)
\le\ES_\alpha(G_0)-\VaR_\alpha(G_0)
+\frac{\epsilon}{\sqrt{1-\alpha}}.
\end{aligned}
\]
The first bound follows directly from Proposition~\ref{prop:bench_gap}.
The second uses the relation between the two gaps in
Example~\ref{ex:vares_gap}.

\begin{figure}[!htbp]
\centering
\includegraphics[width=\textwidth]{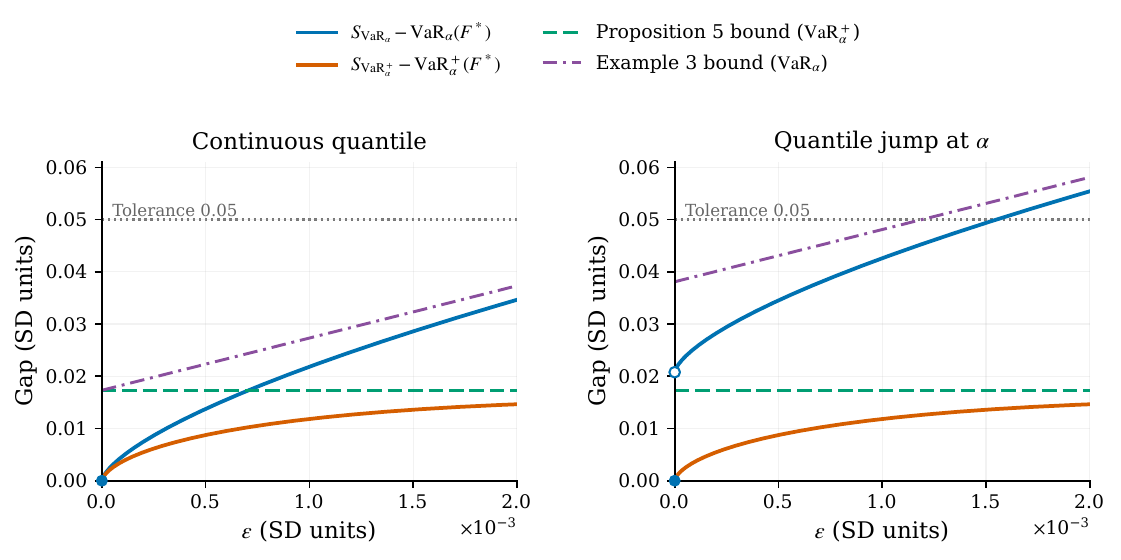}
\caption{$\VaR_\alpha$ and $\VaRp_\alpha$ approximation gaps with respect to radius $\epsilon$,
with $p=2$ and $\alpha=0.99$. The reference distributions are chosen to be the standardized uniform (left panel) or
as in \eqref{eq:num_jump} with a quantile jump (right panel).
Gaps, bounds, and Wasserstein radii are expressed in units of the original
reference loss standard deviation (SD).}
\label{fig:var_comparison}
\end{figure}

Figure~\ref{fig:var_comparison} shows the small-radius limits in
Example~\ref{ex:vares_gap}. For the continuous quantile, both
$\Delta_{\VaR_\alpha}(\epsilon)$ and $\Delta_{\VaRp_\alpha}(\epsilon)$ tend to zero.
For the reference with a quantile jump, $\Delta_{\VaRp_\alpha}(\epsilon)$
still tends to zero, whereas $\Delta_{\VaR_\alpha}(\epsilon)$ tends to
$0.02078$. At $\epsilon=0$, both errors are zero. The open point
marks the positive-radius limit.
At $\epsilon=0.001$, $\Delta_{\VaR_\alpha}(\epsilon)$ and
$\Delta_{\VaRp_\alpha}(\epsilon)$ are approximately $0.02182$ and $0.01182$
in the continuous case, and $0.04260$ and $0.01182$ in the jump case.
All four errors are below $0.05$, so the explicit distribution $F^*$
gives the required accuracy for both VaR conventions, even in the
presence of this quantile jump. The bounds guarantee this accuracy
throughout $0\le\epsilon\le0.001$. For $\VaRp$, the guarantee holds
at every radius. For VaR, it holds throughout the plotted range in
the continuous case and up to $\epsilon=0.00119$ in the jump case.
At $\epsilon=0.002$, the latter error is $0.05541$, which shows why
the radius must also be checked for VaR.

\subsubsection{RVaR}

We use standardized uniform and triangular references, with
$\alpha=0.95$ and $\beta=0.975$ for both. Before standardization,
the distributions are $\U[0,1]$ and a triangular distribution with
density $2y$ on $0\le y\le1$. Both reference losses have mean zero
and variance one after standardization. These choices give
\[
\ES_\alpha(G_0)-\RVaR_{\alpha,\beta}(G_0)
\approx 0.04330\quad\hbox{and}\quad 0.02686,
\]
so both bounds are below the tolerance. By
Proposition~\ref{prop:bench_gap} and \eqref{eq:rvar_gap},
\[
\begin{aligned}
0\le S_{\RVaR_{\alpha,\beta}}-\RVaR_{\alpha,\beta}(F^*)
=S_{\RVaR_{\alpha,\beta}}-\RVaR_{\alpha,\beta}(G_0)
-\frac{\epsilon}{\sqrt{1-\alpha}}
\le\ES_\alpha(G_0)-\RVaR_{\alpha,\beta}(G_0).
\end{aligned}
\]
We compute the worst-case value using the isotonic characterization
in Example~\ref{ex:rvar} and choose the multiplier to satisfy the
Wasserstein constraint. The probability levels remain fixed as we
vary the radius.

\begin{figure}[!htbp]
\centering
\includegraphics[width=\textwidth]{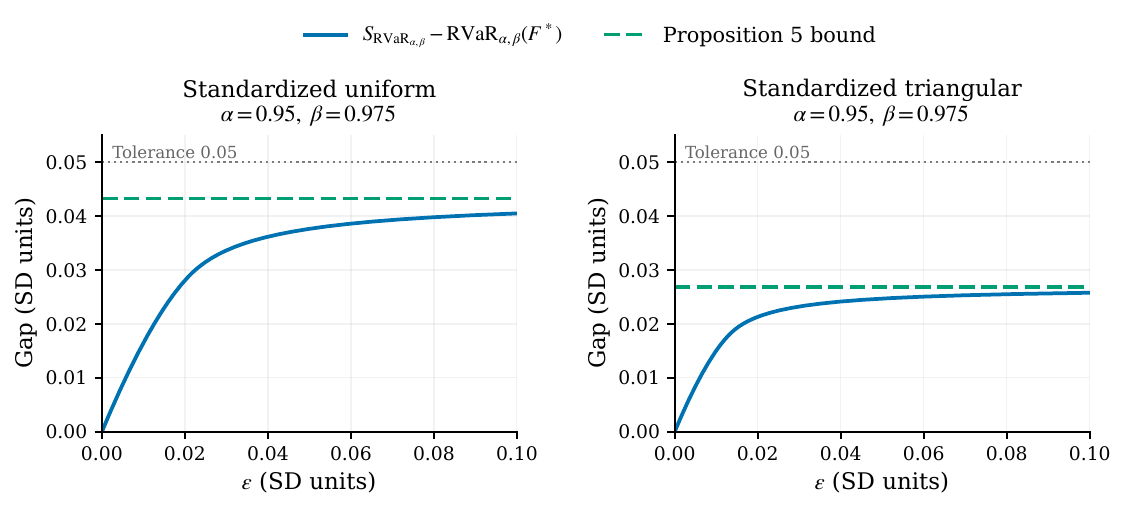}
\caption{$\RVaR_{\alpha,\beta}$ approximation gap with respect to radius $\epsilon$,
with $p=2$, $\alpha=0.95$, and $\beta=0.975$. The reference distributions
are standardized uniform (left panel) and standardized triangular (right panel).
Gaps, bounds, and Wasserstein radii are expressed in units of the original
reference loss standard deviation (SD).}
\label{fig:rvar_gaps}
\end{figure}

In Figure~\ref{fig:rvar_gaps}, both bounds are below $0.05$, so the
explicit distribution $F^*$ meets the tolerance without computing
the worst-case value. The actual gaps tend to zero as $\epsilon$
decreases to zero. At $\epsilon=0.01$, they are approximately
$0.01647$ and $0.01503$ for the uniform and triangular references.
At $\epsilon=0.10$, they are $0.04050$ and $0.02578$, close to their
respective bounds. The curves increase more slowly as the radius
grows over the plotted range, while the bounds stay constant.

\subsubsection{TK distortion}
\label{subsec:tk_numerics}

Next, we study the inverse-S-shaped distortion function considered by \cite{TK92}:
\[
h_\gamma(t)=\frac{t^\gamma}
{\bigl(t^\gamma+(1-t)^\gamma\bigr)^{1/\gamma}},
\qquad 0\le t\le1,\qquad\gamma\in(0,1).
\]
We call the distortion function above a \emph{TK distortion function} for simplicity.
We use $P_\epsilon$ from \eqref{eq:radial_candidate} to approximate the worst-case value $S_h(\epsilon)$. We keep the error tolerance at $0.05$ and set $p=2$.
Choosing $\gamma=0.7$, the TK distortion is absolutely continuous and
$\psi_h\in L^2(0,1)$. We therefore use Proposition~\ref{prop:isotonic}
to compute the numerical worst-case value $S_h(\epsilon)$ for
comparison with the explicit value $\rho_h(P_\epsilon)$.

We use a standard normal reference and a Pareto reference with
minimum $1$ and shape parameter $10$. Both reference losses are
standardized to have mean zero and variance one.
We divide the probability interval $[0,1]$ into $N=4{,}000$ cells,
with initial endpoints
\[
t_i=\frac{\sin^4(\pi i/(2N))}
{\sin^4(\pi i/(2N))+\cos^4(\pi i/(2N))},
\qquad i=0,\ldots,N.
\]
The cells are narrower near $0$ and $1$, where the reference quantile
and $\psi_h$ vary rapidly. Each cell uses the average reference
quantile and the average of $\psi_h$, as in \eqref{eq:grid_problem}.
To calculate the candidate $P_\epsilon$, we move the nearest interior
endpoint to $0.869724$ (approximately), where $\psi_{h^*}$ changes
from constant to increasing. We use this adjusted grid for both
the worst-case value and the candidate.
For each $\epsilon$, we solve the weighted isotonic regression in
Subsection~\ref{subsec:finite_grid} and adjust the multiplier until
the Wasserstein distance equals $\epsilon$. Subtracting the
candidate value gives the plotted error.
Figure~\ref{fig:tk_bound} shows $0\le\epsilon\le2.3$ on linear axes,
with the normal reference on the left and the Pareto reference on the right.

Both bounds in Theorem~\ref{thm:radius_bound} apply to this distortion.
Theorem~\ref{thm:radius_bound}(i) and
Theorem~\ref{thm:radius_bound}(ii) give, respectively,
\[
\begin{aligned}
0\le S_h(\epsilon)-\rho_h(P_\epsilon)
&\le\rho_{h^*}(G_0)+\epsilon\|\psi_{h^*}\|_2-\rho_h(P_\epsilon),\\
0\le S_h(\epsilon)-\rho_h(P_\epsilon)
&\le\rho_h(G_0)+\epsilon\|\psi_h\|_2-\rho_h(P_\epsilon).
\end{aligned}
\]
We calculate the bounds by one-dimensional integration using
\eqref{eq:radial_value}. If either bound is below $0.05$, we know
that the error from using $\rho_h(P_\epsilon)$ is below $0.05$
without calculating $S_h(\epsilon)$.

\begin{figure}[!htbp]
\centering
\includegraphics[width=\textwidth]{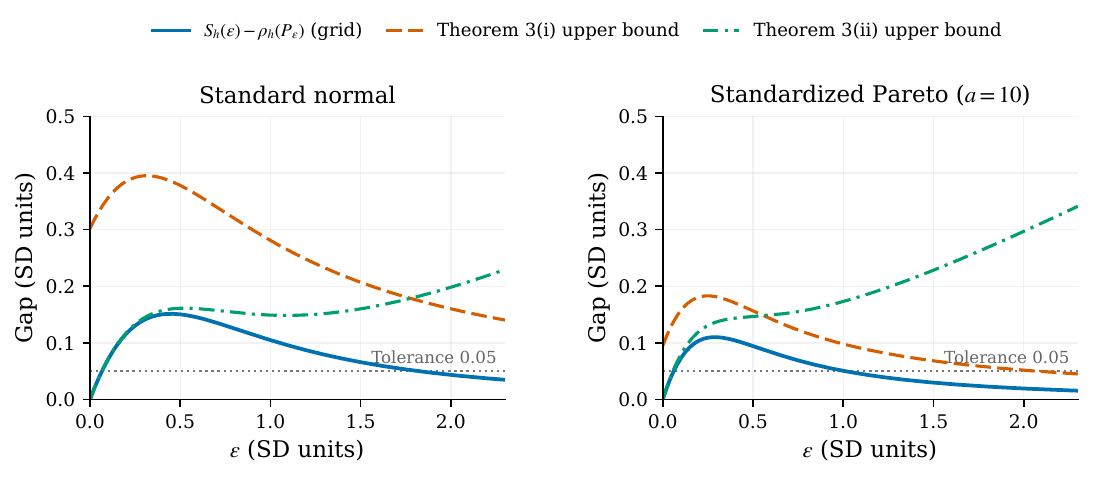}
\caption{TK approximation error with respect to radius $\epsilon$, with $p=2$ and $\gamma=0.7$.
The reference distributions are chosen to be standard normal (left panel) or standardized Pareto with shape $10$ (right panel).
Gaps, bounds, and Wasserstein radii are expressed in units of the original
reference loss standard deviation (SD).}
\label{fig:tk_bound}
\end{figure}

\needspace{3\baselineskip}
For small $\epsilon$, Theorem~\ref{thm:radius_bound}(ii) gives the
smaller upper bound. This bound tends to zero as $\epsilon$ decreases
to zero. For large $\epsilon$, Theorem~\ref{thm:radius_bound}(i)
gives the smaller upper bound. Its bound decreases towards zero,
while the bound in Theorem~\ref{thm:radius_bound}(ii) increases.
The two bounds are equal at about $\epsilon=1.776$ for normal and
$0.558$ for Pareto. For each reference, we use
Theorem~\ref{thm:radius_bound}(ii) below this value and
Theorem~\ref{thm:radius_bound}(i) above it to assess the same
approximation $\rho_h(P_\epsilon)$.

At $\epsilon=0.01$, the bounds in Theorem~\ref{thm:radius_bound}(ii)
are $0.00832$ for normal and $0.01001$ for Pareto, both below $0.05$.
We can therefore use $\rho_h(P_\epsilon)$ at this radius without
calculating $S_h(\epsilon)$.
At $\epsilon=2.3$, the Pareto error is $0.01563$ and the bound in
Theorem~\ref{thm:radius_bound}(i) is $0.04521$. This bound also
guarantees the required accuracy. The normal error is $0.03485$,
but both bounds exceed $0.05$. In this case, the bounds alone do
not confirm that the approximation meets the tolerance.

Over the plotted range, at least one bound is below $0.05$ when
$\epsilon\le0.066$ for normal, and when $\epsilon\le0.056$ or
$\epsilon\ge2.076$ for Pareto. At these radii, the bounds guarantee
the required accuracy without solving the original worst-case problem.
The largest errors on our grid are about $0.151$ for normal and
$0.110$ for Pareto, so the approximation does not meet the tolerance
at every radius. When both bounds exceed $0.05$, we need a direct
calculation to check the actual error. The choice of $\epsilon$
should reflect model uncertainty.

\subsection{Robust portfolio selection}
\label{subsec:portfolio_test}

We use the portfolio model in Section~\ref{sec:portfolio} for two
experiments. The TK experiment compares the portfolios selected by
the original and approximate objectives. The IQR experiment examines
whether regularization recovers the minimum for a discontinuous
distortion.

We consider three asset losses, set $p=r=2$, and use nonnegative
portfolio weights that sum to one. The Gaussian benchmark is $P_0=N_3(\bm\mu,\bm\Sigma_\varrho)$, where
\[
\bm\mu=(0.080,0.075,0.070)^\top,
\qquad (\sigma_1,\sigma_2,\sigma_3)=(0.12,0.20,0.32).
\]
Here $\bm\mu$ gives the mean losses, $\sigma_i$ is the standard deviation
of asset $i$'s loss, and $\varrho$ is the common correlation between
each pair of asset losses. The covariance matrix has entries
$(\bm\Sigma_\varrho)_{ii}=\sigma_i^2$ and
$(\bm\Sigma_\varrho)_{ij}=\varrho\sigma_i\sigma_j$ for $i\ne j$.
Both experiments use all weight vectors whose entries are nonnegative
multiples of $0.005$ and sum to one. This gives $20{,}301$ candidate
portfolios.

\subsubsection{TK distortion}
\label{subsec:portfolio_tk}

When we choose a portfolio using an approximation that is easy to
compute, how much do we lose compared with solving the original
problem directly? Is this loss acceptable in exchange for simpler
computation? We measure this loss by evaluating all portfolios under
the original objective.

We take $\varrho=0.3$ and
$\epsilon=c\sqrt{\operatorname{tr}(\bm\Sigma_\varrho)}$, with $c$ ranging
from $0$ to $0.20$ in steps of $0.005$.
Increasing $c$ enlarges the Wasserstein ball.
The parameter $c$ is a sensitivity parameter, not a statistical
confidence radius. We use $h(t)=1-h_\gamma(1-t)$ with $\gamma=0.7$,
the upper-tail reflection of the TK distortion in
Subsection~\ref{subsec:tk_numerics}.

We compare minimizers of the original objective $V_h$, the explicit
approximation $\widehat V_h$ in \eqref{eq:radial_outer}, and the
envelope objective $V_{h^*}$. The portfolio minimizing $V_h$ on our
weight grid provides the reference for comparison. The other two
methods use $P_\epsilon$ to approximate the inner worst-case value,
or replace $h$ by its concave envelope. They give
$\bm\omega^{\mathrm{app}}$ and $\bm\omega^{\mathrm{env}}$, respectively, as in
Section~\ref{subsec:outer_gap}.

To calculate $V_h$, we first compute the worst-case value
using Proposition~\ref{prop:isotonic} and the grid method in
Subsection~\ref{subsec:finite_grid}. Here $\psi_h(u)=h_\gamma'(u)$
is smooth on $(0,1)$. We use $N=4{,}000$ quantile cells with the
initial endpoints defined in Subsection~\ref{subsec:tk_numerics}.
Each cell uses the average reference quantile and the average of
$\psi_h$, as in \eqref{eq:grid_problem}.
We solve the weighted isotonic problem and choose the multiplier
to meet the radius constraint.
We tabulate $S_h$ at $481$ equally spaced standardized radii from
$0$ to $1.2$, with step $0.0025$. This range covers all portfolios
considered here. During the weight search, we use linear interpolation
in \eqref{eq:elliptical_inner} to evaluate $V_h$.
We then calculate $V_{h^*}$ and $\widehat V_h$ from
\eqref{eq:elliptical_value} and \eqref{eq:radial_outer}, respectively.
Both formulas use the envelope weight $\psi_{h^*}$, which is constant
up to the solution of $h_\gamma'(u)=h_\gamma(u)/u$, approximately
$0.767218$, and equals $\psi_h$ thereafter. We move the nearest
interior grid endpoint to this value, so the constant and increasing
parts lie in separate cells. All three objectives use this adjusted
partition.

Table~\ref{tab:portfolio} reports the results for $\varrho=0.3$ at
five values of $c$. For each $c$, we use the minimum of $V_h$ found
on the weight grid as the benchmark. We then evaluate the portfolios
selected by minimizing $\widehat V_h$ and $V_{h^*}$ using the same
original objective $V_h$. The closer their $V_h$ values are to the
benchmark, the better these portfolios perform.

\begin{table}[!htbp]
\centering
\caption{Selected portfolio weights and objective values under
the TK distortion for $\varrho=0.3$ on the discrete weight grid.
The selected weights $\bm\omega=(\omega_1,\omega_2,\omega_3)$
corresponding to $V_h$, $\widehat V_h$, and $V_{h^*}$ are denoted
by $\bm\omega^*$, $\bm\omega^{\mathrm{app}}$, and
$\bm\omega^{\mathrm{env}}$, respectively, following the notation
in Propositions~\ref{prop:outer_radial} and~\ref{prop:outer_gap}
for the latter two. The last column reports numerically evaluated
values of the original objective $V_h(\bm\omega)$ at these
selected weights.}
\label{tab:portfolio}
\begingroup
\small
\setstretch{1.05}
\renewcommand{\arraystretch}{1.12}
\setlength{\tabcolsep}{4.5pt}
\begin{tabular*}{0.84\linewidth}{@{\extracolsep{\fill}}ccrrrc@{}}
\toprule
\multirow{2}{*}{$c$} & \multirow{2}{*}{\shortstack{Objective used\\to select weights}} & \multicolumn{3}{c}{Selected weights} & \multirow{2}{*}{\shortstack{Original objective value\\$V_h(\bm\omega)$}} \\
\cmidrule(lr){3-5}
 & & $\omega_1$ & $\omega_2$ & $\omega_3$ & \\
\midrule
\multirow{3}{*}{0.00} & $V_h$ & 0.655 & 0.260 & 0.085 & 0.093527 \\
 & $\widehat V_h$ & 0.655 & 0.260 & 0.085 & 0.093527 \\
 & $V_{h^*}$ & 0.750 & 0.215 & 0.035 & 0.093739 \\
\midrule
\multirow{3}{*}{0.05} & $V_h$ & 0.455 & 0.345 & 0.200 & 0.109205 \\
 & $\widehat V_h$ & 0.490 & 0.340 & 0.170 & 0.109284 \\
 & $V_{h^*}$ & 0.605 & 0.295 & 0.100 & 0.110291 \\
\midrule
\multirow{3}{*}{0.10} & $V_h$ & 0.405 & 0.350 & 0.245 & 0.123795 \\
 & $\widehat V_h$ & 0.385 & 0.355 & 0.260 & 0.123823 \\
 & $V_{h^*}$ & 0.525 & 0.330 & 0.145 & 0.125102 \\
\midrule
\multirow{3}{*}{0.15} & $V_h$ & 0.385 & 0.345 & 0.270 & 0.138136 \\
 & $\widehat V_h$ & 0.345 & 0.345 & 0.310 & 0.138375 \\
 & $V_{h^*}$ & 0.480 & 0.340 & 0.180 & 0.139394 \\
\midrule
\multirow{3}{*}{0.20} & $V_h$ & 0.375 & 0.345 & 0.280 & 0.152369 \\
 & $\widehat V_h$ & 0.320 & 0.340 & 0.340 & 0.152961 \\
 & $V_{h^*}$ & 0.450 & 0.345 & 0.205 & 0.153495 \\
\bottomrule
\end{tabular*}
\endgroup
\end{table}

At $c=0.10$, the portfolio selected by $\widehat V_h$ increases
$V_h$ by about $0.000028$ relative to the minimum found on the
weight grid. The increase is about $0.001307$ for the portfolio
selected by $V_{h^*}$. At $c=0.20$, the corresponding increases are
$0.000592$ and $0.001126$. In these cases, minimizing $\widehat V_h$
gives a portfolio with an original objective value close to the
numerical minimum, while avoiding the inner worst-case optimization
when selecting the weights.

\FloatBarrier
\subsubsection{Inter-quantile range}
\label{subsec:portfolio_iqr}

We now use the inter-quantile range (IQR) \citep[see e.g.,][]{WWW20a}.
For $\alpha=0.75$, it is given by
\[
h(t)=\mathbf1_{[1-\alpha,\alpha]}(t),
\qquad
\rho_h(F)=F^{-1+}(\alpha)-F^{-1}(1-\alpha).
\]
Here $h$ is discontinuous and nonmonotone, and $\rho_h$ is nonconvex.
The objective measures the spread of losses. Since $h(1)=0$,
the mean term in \eqref{eq:elliptical_inner} drops out.
We keep the same assets and weight grid, fix $\varrho=0.3$, and take
$\epsilon=c\sqrt{\operatorname{tr}(\bm\Sigma_\varrho)}$ for
$c\in\{0.005,0.05\}$.
At each fixed $c$, we use
$\eta\in\{0.10,0.05,0.02,0.01,0.005,0.002,0.001\}$, keeping
$\epsilon$ unchanged.

For $0<\eta<0.125$, \eqref{eq:reg_density} gives
\begin{equation}\label{eq:IQR_regular}
\rho_{h_\eta}(F)
=\frac1\eta\int_\alpha^{\alpha+\eta}F^{-1}(u)\,\mathrm du
-\frac1\eta\int_{1-\alpha-\eta}^{1-\alpha}F^{-1}(u)\,\mathrm du.
\end{equation}
As $\eta$ decreases, $\rho_{h_\eta}(F)$ decreases to $\rho_h(F)$.
Theorem~\ref{thm:regularization}(i) and
Proposition~\ref{prop:outer_reg}(ii) give
$\min_{\bm\omega\in\mathcal W}V_{h_\eta}(\bm\omega)
\to\min_{\bm\omega\in\mathcal W}V_h(\bm\omega)$.
For the worst-case value of \eqref{eq:IQR_regular} with a standard normal reference, the two
averaging intervals are symmetric.
The optimal value for either tail is concave in its squared transport
budget, so allocating half to each tail is optimal. We compute twice
the worst-case $\RVaR_{\alpha,\alpha+\eta}$ value at radius
$\epsilon/\sqrt{2}$ using Proposition~\ref{prop:isotonic}, normal
quantile integrals, and one-dimensional root finding.

For each $\eta$, we compute and store $S_{h_\eta}$ for a standard
normal reference at $961$ geometrically spaced radii from approximately
$0.005920$ to $0.666566$. This range covers the standardized radii of
all candidate portfolios. We use piecewise cubic Hermite interpolation
(PCHIP) to obtain $S_{h_\eta}$ at each portfolio's standardized radius,
and then apply \eqref{eq:elliptical_inner} to calculate
$V_{h_\eta}(\bm\omega)$. We recompute the $32$ portfolios with the
lowest interpolated objective values by root finding and select the
one with the smallest value.

To check the regularized values, we compute the original IQR
independently from the worst-case $\VaRp_\alpha$ formula \eqref{eq:wc_var}.
For the standard normal reference, symmetry and convexity of the
squared transport cost imply that a fixed IQR width costs least
when its endpoints are symmetric about zero. Each tail then uses
half of the squared cost. Thus, for the inner problem at
$\alpha=0.75$,
$S_h(\epsilon)=2S_{\VaRp_\alpha}\!\left(\epsilon/\sqrt{2}\right).$
We compute the right-hand side by one-dimensional root finding
in \eqref{eq:wc_var} and use \eqref{eq:elliptical_inner} to obtain
$V_h(\bm\omega)$ at every weight on the grid.

For each $\eta>0$, Table~\ref{tab:iqr_regularization} gives the selected
weights and their values $V_{h_\eta}(\bm\omega)$.
In the table, $\eta=0$ represents the original IQR problem, with its
minimizing weights and independently computed values $V_h(\bm\omega)$.

\begin{table}[!htbp]
\centering
\caption{Selected portfolio weights and objective values for IQR with different regularization parameters.}
\label{tab:iqr_regularization}
\begingroup
\small
\setstretch{1.05}
\renewcommand{\arraystretch}{1.12}
\setlength{\tabcolsep}{5pt}
\begin{tabular*}{0.96\linewidth}{@{\extracolsep{\fill}}lcccc@{}}
\toprule
\multirow{2}{*}{$\eta$} & \multicolumn{2}{c}{$c=0.005$} & \multicolumn{2}{c}{$c=0.05$} \\
\cmidrule(lr){2-3}\cmidrule(lr){4-5}
 & Selected weights & $\min_{\bm\omega} V_{h_\eta}(\bm\omega)$ & Selected weights &  $\min_{\bm\omega} V_{h_\eta}(\bm\omega)$ \\
\midrule
0.100 & $(0.805,\,0.190,\,0.005)$ & 0.201072 & $(0.705,\,0.250,\,0.045)$ & 0.261515 \\
0.050 & $(0.795,\,0.195,\,0.010)$ & 0.183416 & $(0.690,\,0.260,\,0.050)$ & 0.258667 \\
0.020 & $(0.785,\,0.200,\,0.015)$ & 0.177233 & $(0.690,\,0.260,\,0.050)$ & 0.258667 \\
0.010 & $(0.785,\,0.200,\,0.015)$ & 0.177117 & $(0.690,\,0.260,\,0.050)$ & 0.258667 \\
0.005 & $(0.785,\,0.200,\,0.015)$ & 0.177117 & $(0.690,\,0.260,\,0.050)$ & 0.258667 \\
0.002 & $(0.785,\,0.200,\,0.015)$ & 0.177117 & $(0.690,\,0.260,\,0.050)$ & 0.258667 \\
0.001 & $(0.785,\,0.200,\,0.015)$ & 0.177117 & $(0.690,\,0.260,\,0.050)$ & 0.258667 \\
\midrule
0 & $(0.785,\,0.200,\,0.015)$ & 0.177117 & $(0.690,\,0.260,\,0.050)$ & 0.258667 \\
\bottomrule
\end{tabular*}
\endgroup
\end{table}

\Needspace{3\baselineskip}
Table~\ref{tab:iqr_regularization} shows that at $c=0.005$, the value
decreases from $0.201072$ at $\eta=0.10$
to $0.177117$ at $\eta=0.01$, matching the value in the $\eta=0$ row to
the displayed precision. The selected weights also match.
For $c=0.05$, this agreement holds for $\eta\le0.05$.
In these cases, the flat parts of the worst-case quantile make
the averages equal to the endpoint quantiles, and the same
distribution remains optimal. Thus a positive $\eta$ can already
recover the original value in this example.

\FloatBarrier

\section{Concluding remarks}

Moment constraints and Wasserstein distance express different information
about distributional uncertainty. Combining them is appropriate when the
moment information is reliable enough to be imposed on all admissible
distributions. When the benchmark's mean and dispersion are themselves
uncertain, a Wasserstein ball alone permits these features to vary jointly
with the distribution's shape. For a fixed benchmark and radius, adding
moment constraints can exclude alternatives and lower the worst-case
value. Our setup is therefore useful when the reference distribution is
an estimate and no additional moment information justifies restricting
those alternatives.

Our three main results provide three approaches to solving DRO
problems for nonconvex distortion riskmetrics. Together, they
suggest the following sequence for risk evaluation.
Direct convexification should be checked first. When the exactness
criterion in Theorem~\ref{thm:flatness} holds, existing convex methods
can be used without changing the worst-case value. Otherwise, the
convex counterpart still provides an upper bound, but a decision based
on it may reflect additional conservatism.
The regularization method in Theorem~\ref{thm:regularization} is
appropriate when a more accurate evaluation of the original risk measure
is needed after convexification fails. It recovers the exact worst-case
value as a limit of solvable problems. The cost is repeated numerical
optimization, and convergence alone does not certify the error for a
fixed regularization parameter and grid.
The explicit approximations in Theorem~\ref{thm:radius_bound} are useful
when repeated risk evaluation would be costly and a controlled objective
error is acceptable under a chosen tolerance level.
The theorem provides an explicit distribution in the Wasserstein ball and a
computable upper bound on how much its risk value underestimates the
true worst-case value. The approximation gap vanishes as
the Wasserstein radius $\epsilon$ tends to zero, and the same
conclusion holds as $\epsilon$ tends to infinity provided that $p>1$. 

Future research may extend the analysis to dynamic models with
sequential decisions and evolving information. 
One can also consider extending the linear setup to more general loss functions for a robust portfolio selection problem.
It would also be promising to explore optimization of distortion riskmetrics over ambiguity sets induced by other distributional distances beyond the Wasserstein framework. 

\subsection*{Acknowledgements}
The authors are grateful to members of the research group on financial mathematics and risk management at CUHK-Shenzhen for their useful feedback and conversations. 
Y.L. acknowledges financial support from the National Natural Science Foundation of China (Grant No.~12401624), The Chinese University of Hong Kong (Shenzhen) University Development Fund (Grant No.~UDF01003336), Guangdong Science and Technology Program (Grant No.~2024QN11X076) and Shenzhen Science and Technology Program (Grant No.~RCBS20231211090814028, JCYJ20250604141203005, 2025TC0010) and is partly supported by the Guangdong Provincial Key Laboratory of Mathematical Foundations for Artificial Intelligence (Grant No.~2023B1212010001). GPT-5.6 Sol and GPT-6 Astra were used in the development of technical results and in the production of code for our numerical experiments. The paper is human-written.

\setcounter{lemma}{0}
\renewcommand{\thelemma}{\thesection.\arabic{lemma}}
\renewcommand{\theHlemma}{\theHsection.\arabic{lemma}}
\setcounter{proposition}{0}
\renewcommand{\theproposition}{\thesection.\arabic{proposition}}
\renewcommand{\theHproposition}{\theHsection.\arabic{proposition}}
\setcounter{theorem}{0}
\renewcommand{\thetheorem}{\thesection.\arabic{theorem}}
\renewcommand{\theHtheorem}{\theHsection.\arabic{theorem}}
\setcounter{definition}{0}
\renewcommand{\thedefinition}{\thesection.\arabic{definition}}
\renewcommand{\theHdefinition}{\theHsection.\arabic{definition}}
\setcounter{corollary}{0}
\renewcommand{\thecorollary}{\thesection.\arabic{corollary}}
\renewcommand{\theHcorollary}{\theHsection.\arabic{corollary}}
\setcounter{example}{0}
\renewcommand{\theexample}{\thesection.\arabic{example}}
\renewcommand{\theHexample}{\theHsection.\arabic{example}}
\setcounter{equation}{0}
\renewcommand{\theequation}{\thesection.\arabic{equation}}
\renewcommand{\theHequation}{\theHsection.\arabic{equation}}
\renewcommand{\thefigure}{\thesection.\arabic{figure}}
\renewcommand{\thetable}{\thesection.\arabic{table}}

\clearpage
\appendix

\section{Proofs and auxiliary results}
\label{app:proofs}

We collect the complete proofs of the technical results stated in the main
text. The proofs are gathered according to the sections in which the
corresponding results appear.

\subsection{Proofs for
Section~\ref{sec:convexification}}
\label{app:convexification}

\begin{proof}[Proof of Theorem~\ref{thm:flatness}]
Sufficiency for both (i) and (ii) follows from \eqref{eq:flatness_equal} justified by Example 11 of \cite{PWW25}. We prove necessity.

For (i), let $p>1$. Suppose that there exists an interval
$C\in\mathcal I_h$ on which $G_0^{-1}$ is not constant Lebesgue-a.e. We aim to show that $S_h(\epsilon)<S_{h^*}(\epsilon).$

Write $C=(a,b)$. Because $h^*$ is linear on $(1-b,1-a)$, $\psi_{h^*}$ is constant
Lebesgue-a.e.~on $C$. It follows from \eqref{eq:bench_shift} that $r^*$ is
constant Lebesgue-a.e.~on $C$. Therefore, the quantile $(F^*)^{-1}=G_0^{-1}+r^*$ is not flat on $C$. Hence, there exist $a<u_1<u_2<b$ such that $(F^*)^{-1}(u_1)<(F^*)^{-1}(u_2)$. Let $Y^*=(F^*)^{-1}(U)$, where $U\sim \U[0,1]$. For every
$x\in((F^*)^{-1}(u_1),(F^*)^{-1}(u_2))$, the monotonicity of $(F^*)^{-1}$ implies that
$$1-b<1-u_2\le\p(Y^*>x)\le 1-u_1<1-a.$$
Define
$D_h(t)=h^*(t)-h(t)$ for $t\in[0,1]$. By definition, $D_h(0)=D_h(1)=0$ and $D_h(t)>0$ for all $t\in(1-b,1-a)$. It follows that
\begin{align}
\rho_{h^*}(F^*)-\rho_h(F^*)
=
\int_{\R} D_h\left(\p(Y^*>x)\right)\,\mathrm d x>0.
\label{eq:gap_integral}
\end{align}

Because $h\le h^*$, we have $S_h(\epsilon)\le S_{h^*}(\epsilon)$. Suppose, to the contrary, that equality holds. There exists a sequence
$\{P_n\}_{n\in\N}\subset\M_p(G_0,\epsilon)$ such that
$\rho_h(P_n)\to S_{h^*}(\epsilon)$. Because
$\rho_h(P_n)\le\rho_{h^*}(P_n)\le S_{h^*}(\epsilon)$, we also have
$\rho_{h^*}(P_n)\to S_{h^*}(\epsilon)$.
For each $n\in\N$, set $r_n=P_n^{-1}-G_0^{-1}$. By
\eqref{eq:quantile_density},
\begin{align*}
\rho_{h^*}(P_n)-\rho_{h^*}(G_0)
=
\int_0^1P_n^{-1}(u)\psi_{h^*}(u)\,\mathrm d u
-\int_0^1G_0^{-1}(u)\psi_{h^*}(u)\,\mathrm d u=
\int_0^1r_n(u)\psi_{h^*}(u)\,\mathrm d u.
\end{align*}
By \eqref{eq:bench_hstar}, as $n\to\infty$,
\begin{equation}
\int_0^1r_n(u)\psi_{h^*}(u)\,\mathrm d u
\to
S_{h^*}(\epsilon)-\rho_{h^*}(G_0)
=\epsilon\|\psi_{h^*}\|_{p'}.
\label{eq:holder_limit}
\end{equation}
Because $P_n\in\M_p(G_0,\epsilon)$, we have $\|r_n\|_p\le\epsilon$.

Define
\begin{equation*}
\varphi(u)=\frac{\psi_{h^*}(u)}{\|\psi_{h^*}\|_{p'}},
\qquad
z(u)=\frac{\operatorname{sgn}(\psi_{h^*}(u))|\psi_{h^*}(u)|^{p'-1}}
{\|\psi_{h^*}\|_{p'}^{p'-1}},
\qquad u\in(0,1).
\end{equation*}
Then $\|\varphi\|_{p'}=1$, $\|z\|_p=1$, $\int_0^1 z(u)\varphi(u)\,\mathrm d u=1$, and $r^*=\epsilon z$. Set $f_n=r_n/\epsilon$. Then $\|f_n\|_p\le1$, and
$\int_0^1f_n(u)\varphi(u)\,\mathrm d u
=(\epsilon\|\psi_{h^*}\|_{p'})^{-1}
\int_0^1r_n(u)\psi_{h^*}(u)\,\mathrm d u\to1$ as $n\to\infty$.
Moreover,
\begin{align*}
1\ge \frac{\|f_n\|_p+\|z\|_p}{2}\ge \left\|\frac{f_n+z}{2}\right\|_p
&\ge
\int_0^1\frac{f_n(u)+z(u)}{2}\varphi(u)\,\mathrm d u \\
&=
\frac12\int_0^1f_n(u)\varphi(u)\,\mathrm d u
+\frac12\int_0^1z(u)\varphi(u)\,\mathrm d u
\to 1,~~n\to\infty.
\end{align*}
It follows that
$\|(f_n+z)/2\|_p\to1$ as $n\to\infty$. Because $L^p(0,1)$ is uniformly convex and
$\|f_n\|_p\le1=\|z\|_p$, we have $\|f_n-z\|_p\to0$ as $n\to\infty$. Consequently,
$r_n\xrightarrow{L^p} r^*$ as $n\to\infty$. Because $P_n^{-1}=G_0^{-1}+r_n$ and
$(F^*)^{-1}=G_0^{-1}+r^*$, we obtain $P_n^{-1}\to (F^*)^{-1}$ in $L^p(0,1)$ as $n\to\infty$.

Let $Y_n=P_n^{-1}(U)$. Then
$Y_n\xrightarrow{L^p} Y^*$ as $n\to\infty$, and hence $Y_n\xrightarrow{\mathrm{d}} Y^*$. Because $h^*$ is continuous and $h=\widehat h$ is upper
semicontinuous, the function $D_h=h^*-h$ is lower semicontinuous. Therefore, for
Lebesgue-a.e. $x\in\R$,
$D_h\left(\p(Y^*>x)\right)
\le
\liminf_{n\to\infty}D_h\left(\p(Y_n>x)\right)$.
By Fatou's lemma and \eqref{eq:gap_integral}, we obtain
\begin{align}\label{eq:thm1_fatou}
0<
\int_{\R}D_h\left(\p(Y^*>x)\right)\,\mathrm d x
\le
\liminf_{n\to\infty}
\int_{\R}D_h\left(\p(Y_n>x)\right)\,\mathrm d x
=
\liminf_{n\to\infty}
\left(\rho_{h^*}(P_n)-\rho_h(P_n)\right).
\end{align}
Consequently,
\begin{align*}
\limsup_{n\to\infty}\rho_h(P_n)
&=
\limsup_{n\to\infty}
\left[
\rho_{h^*}(P_n)
-\left(\rho_{h^*}(P_n)-\rho_h(P_n)\right)
\right] \\
&\le
\limsup_{n\to\infty}\rho_{h^*}(P_n)
-\liminf_{n\to\infty}
\left(\rho_{h^*}(P_n)-\rho_h(P_n)\right)
<
\limsup_{n\to\infty}\rho_{h^*}(P_n) = S_{h^*}(\epsilon).
\end{align*}
This contradicts $\rho_h(P_n)\to S_{h^*}(\epsilon)$ and thus proves
$S_h(\epsilon)<S_{h^*}(\epsilon)$. The proof for (i) is complete.

For (ii), let $p=1$ and suppose that $G_0^{-1}$ is not flat on
$\mathcal I_h$. Choose $(a,b)\in\mathcal I_h$ and $a<u_1<u_2<b$ such that
$G_0^{-1}(u_1)<G_0^{-1}(u_2)$. As in (i), for every
$x\in(G_0^{-1}(u_1),G_0^{-1}(u_2))$, we have
$1-G_0(x)\in[1-u_2,1-u_1]\subset(1-b,1-a)$, where $D_h>0$.
Thus $\rho_{h^*}(G_0)-\rho_h(G_0)>0$.

Suppose, to the contrary, that there exists a sequence $\epsilon_n>0$
with $\epsilon_n\downarrow0$ and $S_h(\epsilon_n)=S_{h^*}(\epsilon_n)$.
For each $n\in\N$, choose $P_n\in\mathcal M_1(G_0,\epsilon_n)$ such that
$\rho_h(P_n)\ge S_h(\epsilon_n)-1/n$. Then
\[
0\le\rho_{h^*}(P_n)-\rho_h(P_n)
\le S_{h^*}(\epsilon_n)-\rho_h(P_n)\le\frac1n.
\]
Moreover, $W_1(P_n,G_0)\le\epsilon_n\to0$, so
$P_n(x)\to G_0(x)$ at every continuity point of $G_0$, hence for
Lebesgue-a.e.~$x\in\R$. Lower semicontinuity of $D_h$ gives
$D_h(1-G_0(x))\le\liminf_nD_h(1-P_n(x))$ at these points.
Fatou's lemma yields
\[\begin{aligned}
0<\rho_{h^*}(G_0)-\rho_h(G_0)&=\int_{\R}D_h(1-G_0(x))\,\mathrm d x \\&\le\liminf_{n\to\infty}\int_{\R}D_h(1-P_n(x))\,\mathrm d x
=\liminf_{n\to\infty}
\bigl[\rho_{h^*}(P_n)-\rho_h(P_n)\bigr]=0.
\end{aligned}\]
This contradiction proves that $S_h(\epsilon)<S_{h^*}(\epsilon)$ for all
sufficiently small $\epsilon>0$, completing the proof of (ii).
\end{proof}

\begin{proof}[Proof of Proposition~\ref{prop:wc_hstar}]
For every $F\in\mathcal M_p(G_0,\epsilon)$ we have by H\"older's inequality that
\begin{equation}\label{eq:prop1_holder}
\rho_{h^*}(F)-\rho_{h^*}(G_0)
=\int_0^1\bigl(F^{-1}(u)-G_0^{-1}(u)\bigr)\psi_{h^*}(u)\,\mathrm d u
\le\epsilon\|\psi_{h^*}\|_{p'}.
\end{equation}
If $h^*\equiv0$, then $\rho_{h^*}(F)=S_{h^*}(\epsilon)=0$.
Henceforth, we assume $h^*\not\equiv0$.

For $p>1$, because $\psi_{h^*}$ is nondecreasing, $r^*$ in \eqref{eq:bench_shift} is nondecreasing and has $L^p$-norm $\epsilon$. Thus $G_0^{-1}+r^*$ defines a valid quantile,
up to its left-continuous modification, with $W_p(F^*,G_0)=\epsilon$.
It attains equality in \eqref{eq:prop1_holder}, proving
\eqref{eq:bench_hstar}. The equality condition in H\"older's inequality
determines $r^*$ uniquely almost everywhere, so $F^*$ is unique.

For $p=1$, a feasible $F$ attains equality in
\eqref{eq:prop1_holder} if and only if $W_1(F,G_0)=\epsilon$ and
\[
\begin{gathered}
\bigl(F^{-1}(u)-G_0^{-1}(u)\bigr)\psi_{h^*}(u)
=\bigl|F^{-1}(u)-G_0^{-1}(u)\bigr|\|\psi_{h^*}\|_\infty
\quad\text{for a.e. }u\in(0,1).
\end{gathered}
\]
Thus $F^{-1}-G^{-1}_0$ must have the same sign as $\psi_{h^*}$ and vanish
where $|\psi_{h^*}|<\|\psi_{h^*}\|_\infty$. It follows that
$\{|\psi_{h^*}|=\|\psi_{h^*}\|_\infty\}$ has positive Lebesgue measure. Conversely, if 
this set has positive measure, $r^*$ in
\eqref{eq:bench_shift} is well defined, nondecreasing, and has
$L^1$-norm $\epsilon$. It therefore defines a feasible $F^*$ attaining
equality in \eqref{eq:prop1_holder}.
If this set has zero measure, there is no optimizer for \eqref{eq:bench_hstar}. Let $P_n$ be defined in \eqref{eq:bench_sequence} of Remark~\ref{rem:unique_hstar}. Because
$P_n^{-1}-G_0^{-1}$ is nondecreasing and has an $L^1$-norm $\epsilon$,
we have $P_n\in\mathcal M_1(G_0,\epsilon)$ and $W_1(P_n,G_0)=\epsilon$.
Moreover,
\[
\rho_{h^*}(P_n)-\rho_{h^*}(G_0)
=\epsilon\frac{\int_0^1|\psi_{h^*}(u)|^{n+1}\,\mathrm d u}
{\int_0^1|\psi_{h^*}(u)|^n\,\mathrm d u}
\longrightarrow\epsilon\|\psi_{h^*}\|_\infty,
\]
because the ratio lies between $\|\psi_{h^*}\|_n$ and
$\|\psi_{h^*}\|_\infty$, and $\|\psi_{h^*}\|_n\to\|\psi_{h^*}\|_\infty$.
Together with \eqref{eq:prop1_holder}, this proves
\eqref{eq:bench_hstar}.

\end{proof}

\subsection{Results and proofs for
Section~\ref{sec:regularization}}
\label{app:regularization}

For any $F\in\mathcal M_p$, monotonicity gives
\begin{equation}
\max\{|F^{-1}(u)|,|F^{-1+}(u)|\}
\le \|F^{-1}\|_p(u\wedge(1-u))^{-1/p},\qquad 0<u<1.
\label{eq:q_bound}
\end{equation}
Indeed, monotonicity gives
\[
\|F^{-1}\|_p^p\ge
u\bigl((-F^{-1}(u))_+\bigr)^p
+(1-u)\bigl((F^{-1+}(u))_+\bigr)^p,
\]
which implies \eqref{eq:q_bound} because $F^{-1}(u)\le F^{-1+}(u)$.
For $F\in\mathcal M_p(G_0,\epsilon)$, we also have
$\|F^{-1}\|_p\le\|G_0^{-1}\|_p+\epsilon$.

\begin{proof}[Proof of Proposition~\ref{prop:isotonic}]
For (i), H\"older's inequality gives finiteness of $S_h(\epsilon)$.
Take $P_n\in\mathcal M_p(G_0,\epsilon)$ such that
$\rho_h(P_n)\to S_h(\epsilon)$. For $p>1$, the feasible quantile set
$\left\{F^{-1}:F\in\mathcal M_p(G_0,\epsilon)\right\}$
is convex, norm closed, and bounded. Hence, the set is weakly compact. Therefore, there exists a subsequence $P_{n_k}^{-1}$ and
$F\in\mathcal M_p(G_0,\epsilon)$ such that
\[
\int_0^1P_{n_k}^{-1}(u)\varphi(u)\,\mathrm d u
\longrightarrow\int_0^1F^{-1}(u)\varphi(u)\,\mathrm d u
\qquad\text{for every }\varphi\in L^{p'}(0,1).
\]
Taking $\varphi=\psi_h$ shows that $F$ attains the supremum.
For $p=1$, \eqref{eq:q_bound} and Helly's selection theorem give a subsequence
$P_{n_k}^{-1}$ converging to a quantile $F^{-1}$ at its continuity points,
and hence almost everywhere.
Fatou's lemma gives $\|F^{-1}-G_0^{-1}\|_1\le\epsilon$, so
$F\in\mathcal M_1(G_0,\epsilon)$.
For each $0<\delta<1/2$, dominated convergence gives
$\int_\delta^{1-\delta}(P_{n_k}^{-1}-F^{-1})\psi_h\,\mathrm d u\to0$
as $k\to\infty$, while H\"older's inequality gives
\[
\left|\int_{(0,\delta)\cup(1-\delta,1)}
(P_{n_k}^{-1}-F^{-1})\psi_h\,\mathrm d u\right|
\le2(\|G_0^{-1}\|_1+\epsilon)
\left\|\psi_h\mathbf 1_{(0,\delta)\cup(1-\delta,1)}\right\|_\infty.
\]
By \eqref{eq:psi_endpoints}, letting $k\to\infty$ and then
$\delta\downarrow0$ gives $\rho_h(F)=S_h(\epsilon)$.

On the convex set $\{F^{-1}:F\in\mathcal M_p\}\subset L^p(0,1)$,
$F^{-1}\mapsto\int_0^1F^{-1}\psi_h\,\mathrm d u$ is continuous and
linear, and $F^{-1}\mapsto\|F^{-1}-G_0^{-1}\|_p^p$ is continuous and convex.
Since $G_0\in\mathcal M_p$ and
$\|G_0^{-1}-G_0^{-1}\|_p^p=0<\epsilon^p$,
Lagrange duality gives \eqref{eq:reg_dual} and an optimal multiplier
$\lambda^*\ge0$.
For any $F\in\mathcal M_p(G_0,\epsilon)$ and $\lambda\ge0$, we have
\[
\begin{aligned}
\rho_h(F)
\le \rho_h(F)+\frac{\lambda}{p}
\bigl[\epsilon^p-W_p^p(F,G_0)\bigr]\le \frac{\lambda\epsilon^p}{p}
+\sup_{\widetilde F\in\mathcal M_p}
\left[\rho_h(\widetilde F)
-\frac{\lambda}{p}W_p^p(\widetilde F,G_0)\right].
\end{aligned}
\]
Both inequalities are equalities if and only if
\begin{equation}\label{eq:prop1_condition}
\lambda\bigl[W_p^p(F,G_0)-\epsilon^p\bigr]=0~~\mbox{and}~~\displaystyle F\in\operatorname*{argmax}_{\widetilde F\in\mathcal M_p}
\left[\rho_h(\widetilde F)
-\frac{\lambda}{p}W_p^p(\widetilde F,G_0)\right].
\end{equation}
Together with \eqref{eq:reg_dual}, this proves the optimality criterion.
Finally, if $h(1)\neq 0$, suppose for contradiction that $\lambda^*=0$. We have 
$$S_h(\epsilon)=\sup_{F\in\mathcal M_p}\rho_h(F)\ge\sup_{c\in\R}c\,h(1)=\infty.$$
This contradicts finiteness of $S_h(\epsilon)$, and thus
$\lambda^*>0$.
For any optimal distribution $P_{\lambda^*}$,
\eqref{eq:prop1_condition} therefore gives $W_p(P_{\lambda^*},G_0)=\epsilon$.

For (ii), let $p>1$ and $\lambda>0$. H\"older's inequality gives
\[
\begin{aligned}
\frac{\lambda}{p}\|F^{-1}-G_0^{-1}\|_p^p-\rho_h(F)
\ge{}&\frac{\lambda}{p}\|F^{-1}-G_0^{-1}\|_p^p\\
&-\|\psi_h\|_{p'}
\bigl(\|F^{-1}-G_0^{-1}\|_p+\|G_0^{-1}\|_p\bigr)
\longrightarrow\infty,
\end{aligned}
\]
as $\|F^{-1}-G_0^{-1}\|_p\to\infty$.
The set $\{F^{-1}:F\in\mathcal M_p\}$ is convex and norm closed,
hence weakly closed in $L^p(0,1)$.
Weak lower semicontinuity and strict convexity give a unique minimizer
$P_\lambda$.

To obtain \eqref{eq:reg_quantile}, define $D_\lambda\in L^{p'}(0,1)$
by the following formula for $0<v<1$:
\[
D_\lambda(v)=\lambda|P_\lambda^{-1}(v)-G_0^{-1}(v)|^{p-2}
\bigl(P_\lambda^{-1}(v)-G_0^{-1}(v)\bigr)-\psi_h(v).
\]
Differentiating the inner problem of \eqref{eq:reg_dual} along $(1-t)P_\lambda^{-1}+tF^{-1}$ at $t=0+$ gives
\begin{equation}\label{eq:prop2_diff}
\int_0^1D_\lambda(v)\bigl[F^{-1}(v)-P_\lambda^{-1}(v)\bigr]\,\mathrm d v
\ge0,\qquad F\in\mathcal M_p.
\end{equation}
For $z\in\R$ and $t>0$, take respectively
\[
\begin{aligned}
F^{-1}=P_\lambda^{-1}+(t\wedge (z-P_\lambda^{-1})_+)\mathbf1_{(a,1)}~~\mbox{and}~~
F^{-1}=P_\lambda^{-1}-(t\wedge (P_\lambda^{-1}-z)_+)\mathbf1_{(0,b]}.
\end{aligned}
\]
Dividing by $t$ and letting $t\downarrow0$ yields, respectively,
\[
\begin{aligned}
\int_a^{P_\lambda(z-)}D_\lambda(v)\,\mathrm d v \ge0,~~0\le a<P_\lambda(z-),
~~\mbox{and}~~
\int_{P_\lambda(z)}^bD_\lambda(v)\,\mathrm d v\le0,~~P_\lambda(z)<b\le1.
\end{aligned}
\]
Because $P_\lambda^{-1}(v)<z$ for $v<P_\lambda(z-)$ and
$P_\lambda^{-1}(v)>z$ for $v>P_\lambda(z)$, monotonicity of
$x\mapsto|x|^{p-2}x$ and \eqref{eq:block_root} imply
$z_{a,P_\lambda(z-)}(\lambda)\le z$ for $0\le a<P_\lambda(z-)$, and
$z_{P_\lambda(z),b}(\lambda)\ge z$ for $P_\lambda(z)<b\le1$.
At any continuity point $u\in(0,1)$ of $P_\lambda^{-1}$, these bounds give
\[
\begin{aligned}
z<P_\lambda^{-1}(u)
&\ \Longrightarrow\ P_\lambda(z)<u
\ \Longrightarrow\
\sup_{0\le a<u}\inf_{u<b\le1}z_{a,b}(\lambda)
\ge\inf_{u<b\le1}z_{P_\lambda(z),b}(\lambda)\ge z,\\
z>P_\lambda^{-1}(u)
&\ \Longrightarrow\ u<P_\lambda(z-)
\ \Longrightarrow\
\sup_{0\le a<u}\inf_{u<b\le1}z_{a,b}(\lambda)
\le\sup_{0\le a<u}z_{a,P_\lambda(z-)}(\lambda)\le z.
\end{aligned}
\]
Letting $z\uparrow P_\lambda^{-1}(u)$ and
$z\downarrow P_\lambda^{-1}(u)$ proves \eqref{eq:reg_quantile}
at every continuity point, hence up to its left-continuous version.

If $h(1)\ne0$, part (i) and uniqueness of the inner optimizer give
$\operatorname*{argmax}_{F\in\mathcal M_p(G_0,\epsilon)}\rho_h(F)
=\{P_{\lambda^*}\}$, $\lambda^*>0$.
Applying $F^{-1}=P^{-1}_{\lambda^*}\pm 1$ to \eqref{eq:prop2_diff} gives
$\int_0^1D_{\lambda^*}(v)\,\mathrm d v=0$. Hence,
\[
\lambda^*
=\frac{h(1)}{\displaystyle\int_0^1
|P_{\lambda^*}^{-1}(v)-G_0^{-1}(v)|^{p-2}
\bigl(P_{\lambda^*}^{-1}(v)-G_0^{-1}(v)\bigr)\,\mathrm d v}.
\]
Thus $\lambda^*$ is unique.
\end{proof}

For $\eta>0$ and $F\in\mathcal M_p$, use the widths in
\eqref{eq:reg_widths} to define the following averaging operators:
\begin{equation}\label{eq:sec4_operator}
\begin{aligned}
&A_\eta[F^{-1}](u)
=\frac1{\delta_\eta(u)}
\int_{u-\delta_\eta(u)}^u F^{-1}(v)\,\mathrm d v,~~u\in(0,1),\\
&A_{k,\eta}^-[F^{-1}]
=\frac1{\delta_{k,\eta}^-}
\int_{\xi_k-\delta_{k,\eta}^-}^{\xi_k}F^{-1}(u)\,\mathrm d u
\qquad\mbox{and}\qquad
A_{k,\eta}^+[F^{-1}]
=\frac1{\delta_{k,\eta}^+}
\int_{\xi_k}^{\xi_k+\delta_{k,\eta}^+}F^{-1}(u)\,\mathrm d u.
\end{aligned}
\end{equation}
Retain the discontinuity points $1-\xi_k$ of $h$ and coefficients $c_k^\pm$
from Section~\ref{sec:regularization}.
Assume that
\begin{equation}
\sum_k(|c_k^-|+|c_k^+|)
(\xi_k\wedge(1-\xi_k))^{-1/p}<\infty.
\label{eq:jump_sum}
\end{equation}
Fix $\psi\in L^{p'}(0,1)$, independent of $\eta$, and define, for $F\in\mathcal M_p$,
\begin{align*}
L_0(F)
&=\int_0^1F^{-1}(u)\psi(u)\,\mathrm d u
+\sum_k\left[c_k^-F^{-1}(\xi_k)+c_k^+F^{-1+}(\xi_k)\right],\\
L_{0,\eta}(F)
&=\int_0^1F^{-1}(u)\psi(u)\,\mathrm d u
+\sum_k\left[c_k^-A_{k,\eta}^-[F^{-1}]
+c_k^+A_{k,\eta}^+[F^{-1}]\right].
\end{align*}
For $0<u<1$ and $v\in[u-\delta_\eta(u),u]$, we have
$v\wedge(1-v)\ge [u\wedge(1-u)]/2$.
Likewise, $v\wedge(1-v)\ge[\xi_k\wedge(1-\xi_k)]/2$ for
$v\in[\xi_k-\delta_{k,\eta}^-,\xi_k+\delta_{k,\eta}^+]$.
Thus \eqref{eq:q_bound} gives
\begin{equation}
\begin{aligned}
|A_\eta[F^{-1}](u)|\le2^{1/p}\|F^{-1}\|_p(u\wedge(1-u))^{-1/p}~~\mbox{and}~~
|A_{k,\eta}^{\pm}[F^{-1}]|\le2^{1/p}\|F^{-1}\|_p(\xi_k\wedge(1-\xi_k))^{-1/p}.
\end{aligned}
\label{eq:avg_bound}
\end{equation}

The following two lemmas are used in the proof of Theorem~\ref{thm:regularization}.
\begin{lemma}
\label{lem:jump_approx}
Let $1\le p<\infty$, $G_0\in\mathcal M_p$, $\epsilon>0$, and $\psi\in L^{p'}(0,1)$.
Suppose that \eqref{eq:jump_sum} holds. Then
\[
\begin{aligned}
S_{0,\eta}(\epsilon):=
\sup_{F\in\mathcal M_p(G_0,\epsilon)}L_{0,\eta}(F)<\infty,
~~\eta>0,~~\mbox{and}~~
\lim_{\eta\downarrow0}S_{0,\eta}(\epsilon)
=\sup_{F\in\mathcal M_p(G_0,\epsilon)}L_0(F)=:S_0(\epsilon)<\infty.
\end{aligned}
\]
\end{lemma}
\begin{proof}
H\"older's inequality, \eqref{eq:q_bound}, and \eqref{eq:avg_bound} give,
for every $F\in\mathcal M_p$,
\[
\begin{aligned}
|L_0(F)|\vee\sup_{\eta>0}|L_{0,\eta}(F)|
\le\|F^{-1}\|_p
\left[\|\psi\|_{p'}+2^{1/p}
\sum_k(|c_k^-|+|c_k^+|)(\xi_k\wedge(1-\xi_k))^{-1/p}\right]
<\infty.
\end{aligned}
\]
Because $\|F^{-1}\|_p\le\|G_0^{-1}\|_p+\epsilon$ on
$\mathcal M_p(G_0,\epsilon)$, we have
$|S_0(\epsilon)|\vee\sup_{\eta>0}|S_{0,\eta}(\epsilon)|<\infty$.
For every $F\in\mathcal M_p(G_0,\epsilon)$ and every $k$, we have
$A_{k,\eta}^-[F^{-1}]\to F^{-1}(\xi_k)$ and $A_{k,\eta}^+[F^{-1}]\to F^{-1+}(\xi_k)$ as $\eta\downarrow0$.
By \eqref{eq:avg_bound}, we obtain $L_{0,\eta}(F)\to L_0(F)$. Hence,
\[
\liminf_{\eta\downarrow0}S_{0,\eta}(\epsilon)
\ge\sup_{F\in\mathcal M_p(G_0,\epsilon)}L_0(F)=S_0(\epsilon).
\]

Let $\eta_n\downarrow0$ and choose
$P_n\in\mathcal M_p(G_0,\epsilon)$ with
$L_{0,\eta_n}(P_n)\ge S_{0,\eta_n}(\epsilon)-1/n$.
Fix $m$ of the discontinuity points (take all of them if their number is
finite) and choose $\delta>0$ such that
$\xi_k\in[2\delta,1-2\delta]$ for $1\le k\le m$.
For all sufficiently large $n$, we have $\eta_n<\delta$, and the intervals
$(\xi_k-\eta_n,\xi_k+\eta_n)$, $1\le k\le m$, are pairwise disjoint.
Define $\bar P_n$ for $0<u<1$ by
\[
\bar P_n^{-1}(u)=
\begin{cases}
A_{k,\eta_n}^-[P_n^{-1}],
&\xi_k-\eta_n<u\le\xi_k,\quad 1\le k\le m,\\
A_{k,\eta_n}^+[P_n^{-1}],
&\xi_k<u\le\xi_k+\eta_n,\quad 1\le k\le m,\\
P_n^{-1}(u),&\text{otherwise}.
\end{cases}
\]
The function $\bar P_n^{-1}$ is nondecreasing and left-continuous.
By Jensen's inequality, $$\|\bar P_n^{-1}\|_p\le\|P_n^{-1}\|_p\le\|G_0^{-1}\|_p+\epsilon.$$
By \eqref{eq:q_bound},
\begin{align*}
W^p_p(\bar P_n,P_n)=\sum^m_{k=1}\int^{\xi_k+\eta_n}_{\xi_k-\eta_n}|\bar P^{-1}_n(u)-P^{-1}_n(u)|^p\,\d u
&\le \left[(||P^{-1}_n||_p+||\bar P^{-1}_n||_p)\delta^{-1/p}\right]^p(2m\eta_n)\\
&
\le2^p(\|G_0^{-1}\|_p+\epsilon)^p\delta^{-1}(2m\eta_n)\to0~~\mbox{as}~n\to\infty.
\end{align*}
For $1\le k\le m$, the definition gives $\bar P_n^{-1}(\xi_k)=A_{k,\eta_n}^-[P_n^{-1}]$ and $\bar P_n^{-1+}(\xi_k)=A_{k,\eta_n}^+[P_n^{-1}]$;
for $k>m$, use \eqref{eq:q_bound} for $\bar P_n$ and
\eqref{eq:avg_bound} for $P_n$. It follows that
\begin{align*}
|L_0(\bar P_n)-S_{0,\eta_n}(\epsilon)|&\le \frac{1}{n}+|L_0(\bar P_n)-L_{0,\eta_n}(P_n)|\\
&\le \frac1n+\|\psi\|_{p'}W_p(\bar P_n,P_n)\\
&+(1+2^{1/p})(\|G_0^{-1}\|_p+\epsilon)
\sum_{k>m}(|c_k^-|+|c_k^+|)
(\xi_k\wedge(1-\xi_k))^{-1/p}.
\end{align*}
Define
\[
\widehat P_n^{-1}
=G_0^{-1}
+\frac{\epsilon}{\epsilon+W_p(\bar P_n,P_n)}
(\bar P_n^{-1}-G_0^{-1}).
\]
It follows that
\[
W_p(\widehat P_n,G_0)
=\frac{\epsilon W_p(\bar P_n,G_0)}
{\epsilon+W_p(\bar P_n,P_n)}\le\epsilon~~\mbox{and}~~|L_0(\widehat P_n)-L_0(\bar P_n)|
=\frac{W_p(\bar P_n,P_n)}{\epsilon+W_p(\bar P_n,P_n)}
|L_0(G_0)-L_0(\bar P_n)|\to0
\]
as $n\to\infty$.
Because $L_0(\widehat P_n)\le S_0(\epsilon)$, we obtain
\[
\begin{aligned}
0\le\limsup_{n\to\infty}S_{0,\eta_n}(\epsilon)-S_0(\epsilon)
\le (1+2^{1/p})(\|G_0^{-1}\|_p+\epsilon)\sum_{k>m}(|c_k^-|+|c_k^+|)
(\xi_k\wedge(1-\xi_k))^{-1/p}.
\end{aligned}
\]
Taking all discontinuities in the finite case, or letting $m\to\infty$
otherwise, gives $\lim_n S_{0,\eta_n}(\epsilon)=S_0(\epsilon)$.
Because $\eta_n\downarrow0$ is arbitrary, the result follows.
\end{proof}

\begin{lemma}
\label{lem:nonatomic_approx}
Let $\mu$ be a finite non-atomic positive measure on $(0,1)$ satisfying
$\int_{(0,1)}(u\wedge(1-u))^{-1/p}\,\mu(\mathrm d u)<\infty$.
For the averaging operator $A_\eta$ in \eqref{eq:sec4_operator},
\[
\lim_{\eta\downarrow0}
\sup_{F\in\mathcal M_p(G_0,\epsilon)}
\int_{(0,1)}
|F^{-1}(u)-A_\eta[F^{-1}](u)|
\,\mu(\mathrm d u)
=
0.
\]
\end{lemma}

\Needspace{6\baselineskip}
\begin{proof}
Suppose, to the contrary, that there exist $a>0$, $\eta_n\downarrow0$,
and $P_n\in\mathcal M_p(G_0,\epsilon)$ such that
\[
\int_{(0,1)}
|P_n^{-1}(u)-A_{\eta_n}[P_n^{-1}](u)|\,\mu(\mathrm d u)\ge a.
\]
By \eqref{eq:q_bound} and Helly's selection theorem, a subsequence,
still indexed by $n$, satisfies $P_n^{-1}(u)\to F^{-1}(u)$ at every
continuity point of $F^{-1}$. Fatou's lemma gives
$W_p(F,G_0)\le\liminf_nW_p(P_n,G_0)\le\epsilon$.
For continuity points $0<v<u<1$ of $F^{-1}$,
$v<u-\delta_{\eta_n}(u)$ for all sufficiently large $n$. Hence
$P_n^{-1}(v)\le A_{\eta_n}[P_n^{-1}](u)\le P_n^{-1}(u)$.
Letting $n\to\infty$ and then $v\uparrow u$ gives
$A_{\eta_n}[P_n^{-1}](u)\to F^{-1}(u)$.
Because $F^{-1}$ has at most countably many discontinuities and $\mu$
is non-atomic, we have
$|P_n^{-1}(u)-A_{\eta_n}[P_n^{-1}](u)|\to0$ for $\mu$-a.e.~$u\in(0,1)$.
By \eqref{eq:q_bound}, \eqref{eq:avg_bound}, and
$\|P_n^{-1}\|_p\le\|G_0^{-1}\|_p+\epsilon$, we have, for $0<u<1$,
\[
\begin{aligned}
|P_n^{-1}(u)-A_{\eta_n}[P_n^{-1}](u)|\le |P_n^{-1}(u)|+|A_{\eta_n}[P_n^{-1}](u)|\le(1+2^{1/p})(\|G_0^{-1}\|_p+\epsilon)
(u\wedge(1-u))^{-1/p}.
\end{aligned}
\]
Because $\int_{(0,1)}(u\wedge(1-u))^{-1/p}\,\mu(\mathrm d u)<\infty$,
dominated convergence gives the contradiction
\[
0<a\le\lim_{n\to\infty}
\int_{(0,1)}|P_n^{-1}(u)-A_{\eta_n}[P_n^{-1}](u)|
\,\mu(\mathrm d u)=0.\qedhere
\]
\end{proof}

\Needspace{6\baselineskip}
\begin{proof}[Proof of Theorem~\ref{thm:regularization}]
Minkowski's inequality gives
\begin{equation}\label{eq:prop2_bound}
\|\psi_\eta\|_{p'}
\le\int_{(0,1)}\delta_\eta(u)^{-1/p}\,|\nu^c|(\mathrm d u)
+\sum_k\left[|c_k^-|(\delta_{k,\eta}^-)^{-1/p}
+|c_k^+|(\delta_{k,\eta}^+)^{-1/p}\right]<\infty,
\end{equation}
by \eqref{eq:endpoint_integral} and \eqref{eq:reg_widths}.
The definition gives absolute continuity of $h_\eta$, and Fubini's
theorem gives $h_\eta(1)=h(1)$.
When $p=1$, dominated convergence in the bound
\eqref{eq:prop2_bound}, restricted to the endpoint intervals, gives $\lim_{\delta\downarrow0}
\left\|\psi_\eta\mathbf 1_{(0,\delta)\cup(1-\delta,1)}\right\|_\infty=0$.

The jump series
$\sum_k[c_k^-\mathbf1_{(1-\xi_k,1]}
+c_k^+\mathbf1_{[1-\xi_k,1]}]$ converges uniformly because
$\sum_k(|c_k^-|+|c_k^+|)<\infty$.
Subtracting this series from $h$ leaves a continuous function of bounded
variation. Stieltjes integration of this function and the two step
functions gives \eqref{eq:quantile_bv}.
By \eqref{eq:q_bound} and \eqref{eq:endpoint_integral}, we have
$\sup_{F\in\mathcal M_p(G_0,\epsilon)}|\rho_h(F)|<\infty$.
Fubini's theorem and \eqref{eq:avg_bound} give
\begin{equation}
\rho_{h_\eta}(F)
=\int_{(0,1)}A_\eta[F^{-1}](u)\,\nu^c(\mathrm d u)
+\sum_k\left[c_k^-A_{k,\eta}^-[F^{-1}]
+c_k^+A_{k,\eta}^+[F^{-1}]\right].
\label{eq:reg_objective}
\end{equation}
Proposition~\ref{prop:isotonic} gives attainment for each $\eta>0$.

(i)
Fix $\tau\ge \eta$. Apply Lemma~\ref{lem:jump_approx} with the
fixed density
\[
\psi(v)=\int_{(0,1)}
\frac{\mathbf1_{(u-\delta_\tau(u),u]}(v)}{\delta_\tau(u)}
\,\nu^c(\mathrm d u),\qquad 0<v<1.
\]
The estimate \eqref{eq:prop2_bound} with $\eta=\tau$ and $c_k^\pm=0$
gives $\psi\in L^{p'}(0,1)$.
Lemma~\ref{lem:jump_approx} gives
$S_{0,\eta}(\epsilon)\to S_0(\epsilon)$ as $\eta\downarrow0$.
By \eqref{eq:quantile_bv}, \eqref{eq:reg_objective}, and Fubini's theorem, we obtain
\[
\begin{aligned}
&|S_h(\epsilon)-S_0(\epsilon)|
\le\sup_{F\in\mathcal M_p(G_0,\epsilon)}|\rho_h(F)-L_0(F)|\le\sup_{F\in\mathcal M_p(G_0,\epsilon)}
\int_{(0,1)}|F^{-1}-A_\tau[F^{-1}]|\,\mathrm d|\nu^c|,\\
&|S_{h_\eta}(\epsilon)-S_{0,\eta}(\epsilon)|
\le\sup_{F\in\mathcal M_p(G_0,\epsilon)}
|\rho_{h_\eta}(F)-L_{0,\eta}(F)|\le\sup_{F\in\mathcal M_p(G_0,\epsilon)}
\int_{(0,1)}|A_\eta[F^{-1}]-A_\tau[F^{-1}]|\,\mathrm d|\nu^c|.
\end{aligned}
\]
Hence, by Lemma \ref{lem:jump_approx} and Lemma~\ref{lem:nonatomic_approx},
\[\begin{aligned}
\limsup_{\eta\downarrow0}|S_{h_\eta}(\epsilon)-S_h(\epsilon)|
&\le \limsup_{\eta\downarrow0}\left(|S_{h_\eta}(\epsilon)-S_{0,\eta}(\epsilon)|+|S_{0,\eta}(\epsilon)-S_0(\epsilon)|+|S_{h}(\epsilon)-S_{0}(\epsilon)|\right)\\
&\le
2\sup_{F\in\mathcal M_p(G_0,\epsilon)}
\int_{(0,1)}|F^{-1}-A_\tau[F^{-1}]|\,\mathrm d|\nu^c|
\xrightarrow{\tau\downarrow0}0.
\end{aligned}
\]
This proves \eqref{eq:reg_value}.

(ii)
Suppose that $h=\widehat h$, and take any sequence $\eta\downarrow0$.
By \eqref{eq:q_bound} and Helly's selection theorem, there exists a
subsequence $\eta_n\downarrow0$ for which $(P_n^*)^{-1}$ converges to
$(F^*)^{-1}$ at every continuity point of $(F^*)^{-1}$.
Fatou's lemma gives
\[
W_p^p(F^*,G_0)
\le\liminf_{n\to\infty}W_p^p(P_n^*,G_0)\le\epsilon^p.
\]
Thus $F^*\in\mathcal M_p(G_0,\epsilon)$ and
$P_n^*\xrightarrow{\mathrm w}F^*$.
Because $h=\widehat h$, we have $c_k^-\le0\le c_k^+$.
For continuity points $v<\xi_k<w$ of $(F^*)^{-1}$, we eventually have
\[
(P_n^*)^{-1}(v)
\le A_{k,\eta_n}^-[(P_n^*)^{-1}]
\le A_{k,\eta_n}^+[(P_n^*)^{-1}]
\le(P_n^*)^{-1}(w).
\]
Letting $n\to\infty$, then $v\uparrow\xi_k$ and
$w\downarrow\xi_k$, gives
\[
\liminf_n A_{k,\eta_n}^-[(P_n^*)^{-1}]\ge(F^*)^{-1}(\xi_k),
\qquad
\limsup_n A_{k,\eta_n}^+[(P_n^*)^{-1}]\le(F^*)^{-1+}(\xi_k).
\]
Because $c_k^-\le0\le c_k^+$, we have, for every $m$,
\begin{equation}
\begin{aligned}
&\limsup_{n\to\infty}\sum_{k=1}^m
\left[c_k^-A_{k,\eta_n}^-[(P_n^*)^{-1}]
+c_k^+A_{k,\eta_n}^+[(P_n^*)^{-1}]\right]\\
&\le\sum_{k=1}^m\left[
 c_k^-\liminf_n A_{k,\eta_n}^-[(P_n^*)^{-1}]
+c_k^+\limsup_n A_{k,\eta_n}^+[(P_n^*)^{-1}]\right]\le\sum_{k=1}^m
\left[c_k^-(F^*)^{-1}(\xi_k)
+c_k^+(F^*)^{-1+}(\xi_k)\right].
\end{aligned}
\label{eq:jump_limit}
\end{equation}
For the continuous term, Lemma~\ref{lem:nonatomic_approx} and dominated
convergence give
\begin{equation}
\int_{(0,1)}A_{\eta_n}[(P_n^*)^{-1}]\,\mathrm d\nu^c
\to\int_{(0,1)}(F^*)^{-1}\,\mathrm d\nu^c
\qquad\mbox{as }n\to\infty.
\label{eq:cont_limit}
\end{equation}
For the terms with $k>m$, \eqref{eq:avg_bound} gives
\begin{equation}
\begin{aligned}
&\sup_n\left|\sum_{k>m}
\left[c_k^-A_{k,\eta_n}^-[(P_n^*)^{-1}]
+c_k^+A_{k,\eta_n}^+[(P_n^*)^{-1}]\right]\right|\\
&\quad\le 2^{1/p}(\|G_0^{-1}\|_p+\epsilon)
\sum_{k>m}(|c_k^-|+|c_k^+|)
(\xi_k\wedge(1-\xi_k))^{-1/p}\to0
\qquad\mbox{as }m\to\infty.
\end{aligned}
\label{eq:jump_tail}
\end{equation}
Combining \eqref{eq:reg_objective}--\eqref{eq:jump_tail} with
\eqref{eq:q_bound}, \eqref{eq:quantile_bv}, and \eqref{eq:endpoint_integral} gives
\[
\begin{aligned}
\limsup_{n\to\infty}\rho_{h_{\eta_n}}(P_n^*)
\le{}&\int_{(0,1)}(F^*)^{-1}\,\mathrm d\nu^c
+\sum_{k=1}^m\left[c_k^-(F^*)^{-1}(\xi_k)
+c_k^+(F^*)^{-1+}(\xi_k)\right]\\
&+2^{1/p}(\|G_0^{-1}\|_p+\epsilon)
\sum_{k>m}(|c_k^-|+|c_k^+|)
(\xi_k\wedge(1-\xi_k))^{-1/p}\to\rho_h(F^*)~~\mbox{as}~m\to\infty.
\end{aligned}
\]
Because $\rho_{h_{\eta_n}}(P_n^*)=S_{h_{\eta_n}}(\epsilon)$ and
$F^*\in\mathcal M_p(G_0,\epsilon)$, \eqref{eq:reg_value} now gives
\[
S_h(\epsilon)=\lim_{n\to\infty}S_{h_{\eta_n}}(\epsilon)
=\lim_{n\to\infty}\rho_{h_{\eta_n}}(P_n^*)
\le\rho_h(F^*)\le S_h(\epsilon).
\]
Thus $F^*$ is the optimizer for $S_h(\epsilon)$.

(iii)
If $h(1)\ne 0$, suppose for contradiction that $W_p(F^*,G_0)<\epsilon$. Choose $c$ with
$c\,h(1)>0$ and $|c|<\epsilon-W_p(F^*,G_0)$.
The distribution with $F^{-1}=(F^*)^{-1}+c$ satisfies
\[
W_p(F,G_0)\le W_p(F^*,G_0)+|c|<\epsilon,
\qquad
\rho_h(F)=S_h(\epsilon)+c\,h(1)>S_h(\epsilon),
\]
which leads to a contradiction. Hence every optimizer satisfies $W_p(F^*,G_0)=\epsilon$.
Fatou's lemma gives
\begin{equation}\label{eq:thm2_fatou}
\epsilon=W_p(F^*,G_0)
\le\liminf_n W_p(P_n^*,G_0)
\le\limsup_n W_p(P_n^*,G_0)\le\epsilon.
\end{equation}
Because
\[
|(P_n^*)^{-1}-(F^*)^{-1}|^p\le 2^{p-1}\left(
|(P_n^*)^{-1}-G_0^{-1}|^p+|(F^*)^{-1}-G_0^{-1}|^p\right),
\]
Fatou's lemma and \eqref{eq:thm2_fatou} give
\[
\begin{aligned}
2^pW_p^p(F^*,G_0)
\le\liminf_{n\to\infty}\Bigl\{
2^{p-1}\bigl[W_p^p(P_n^*,G_0)+W_p^p(F^*,G_0)\bigr]
-W_p^p(P_n^*,F^*)\Bigr\}=2^p\epsilon^p-\limsup_{n\to\infty}W_p^p(P_n^*,F^*).
\end{aligned}
\]
Rearranging yields
\[
0\le\limsup_{n\to\infty}W_p^p(P_n^*,F^*)
\le2^p\epsilon^p-2^pW_p^p(F^*,G_0)=0,
\]
and hence $W_p(P_n^*,F^*)\to0$.

(iv)
If also $p>1$, suppose for contradiction that $P_1,P_2$ are distinct optimizers for $S_h(\epsilon)$. Define
$F^{-1}=(P_1^{-1}+P_2^{-1})/2$. Linearity and strict convexity give
\[
\rho_h(F)=\tfrac12\rho_h(P_1)+\tfrac12\rho_h(P_2)=S_h(\epsilon),
\qquad
W_p^p(F,G_0)
<\tfrac12W_p^p(P_1,G_0)+\tfrac12W_p^p(P_2,G_0)=\epsilon^p,
\]
contradicting the equality for every optimizer. Thus $F^*$ is unique.
If $W_p(P_{\lambda_\eta^*},F^*)$ did not tend to zero as $\eta\downarrow0$,
there would exist $c>0$ and a sequence $\eta_n\downarrow0$ such that
$W_p(P^*_n,F^*)\ge c$ for all $n\ge1$.
By (iii) and uniqueness of $F^*$, there exists a subsequence $\eta_{n_j}\downarrow0$ with $W_p(P^*_{n_j},F^*)\to0$ as $j\to\infty$,
leading to a contradiction. Hence,
$W_p(P_{\lambda_\eta^*},F^*)\to0$ as $\eta\downarrow0$.
\end{proof}

\begin{proof}[Proof of Proposition~\ref{prop:reg_rate}]
Because the continuous part of $h$ is Lipschitz, there exists
$C_0\in\R_+$ such that
$|\nu^c|(\mathrm d u)\le C_0\,\mathrm d u$.
All constants below are fixed and independent of $\eta$ and $F$.

(i) For $p>1$, \eqref{eq:endpoint_integral} follows from
\[
\int_0^1(u\wedge(1-u))^{-1/p}\,|\nu^c|(\mathrm d u)
\le 2C_0\int_0^{1/2}u^{-1/p}\,\mathrm d u<\infty
\]
and the fact that $h$ has finitely many discontinuity points.
For $F\in\mathcal M_p$, monotonicity and Fubini's theorem give
\[
\begin{aligned}
\int_{[2\eta,1-2\eta]}|F^{-1}-A_\eta[F^{-1}]|\,\mathrm d|\nu^c|
&\le C_0\int_{2\eta}^{1-2\eta}
\bigl[F^{-1}(u)-F^{-1}(u-\eta)\bigr]\,\mathrm d u\\
&\le C_0\eta\bigl[F^{-1+}(1-2\eta)-F^{-1}(\eta)\bigr]\le C_0(1+2^{-1/p})\|F^{-1}\|_p\eta^{1-1/p},
\end{aligned}
\]
where the last inequality follows from \eqref{eq:q_bound}.
On $(0,2\eta)\cup(1-2\eta,1)$, \eqref{eq:q_bound} and
\eqref{eq:avg_bound} bound the integral by
\[
2C_0(1+2^{1/p})\|F^{-1}\|_p\int_0^{2\eta}u^{-1/p}\,\mathrm d u
=\frac{2C_0(1+2^{1/p})2^{1-1/p}}{1-1/p}
\|F^{-1}\|_p\eta^{1-1/p}.
\]
Consequently, there exists $C_1\in\R_+$ such that
\begin{equation}
\int^1_0|F^{-1}-A_\eta[F^{-1}]|\,\mathrm d|\nu^c|
\le C_1\|F^{-1}\|_p\eta^{1-1/p},
\qquad F\in\mathcal M_p.
\label{eq:avg_rate}
\end{equation}

Choose $a>0$ so that all discontinuity points
$\xi_1,\ldots,\xi_m$ lie in $[2a,1-2a]$.
Take $\eta<a$ sufficiently small that
$(\xi_k-\eta,\xi_k+\eta)$, $1\le k\le m$, are disjoint.
For each feasible $F$, define $\bar P_\eta$ and $\widetilde P_\eta$
by the following quantile functions, for $0<u<1$:
\[
\bar P_\eta^{-1}(u)=
\begin{cases}
A_{k,\eta}^-[F^{-1}],&\xi_k-\eta<u\le\xi_k,\quad k=1,\ldots,m,\\
A_{k,\eta}^+[F^{-1}],&\xi_k<u\le\xi_k+\eta,\quad k=1,\ldots,m,\\
F^{-1}(u),&\text{otherwise},
\end{cases}
\]
\[
\widetilde P_\eta^{-1}(u)=
\begin{cases}
F^{-1}(\xi_k),&\xi_k-\eta<u\le\xi_k,\quad k=1,\ldots,m,\\
F^{-1+}(\xi_k),&\xi_k<u\le\xi_k+\eta,\quad k=1,\ldots,m,\\
F^{-1}(u),&\text{otherwise}.
\end{cases}
\]
These functions are nondecreasing and left-continuous. Their jump
contributions satisfy, respectively,
\begin{equation}\label{eq:prop3_iden}
\begin{aligned}
c_k^-\bar P_\eta^{-1}(\xi_k)+c_k^+\bar P_\eta^{-1+}(\xi_k)
&=c_k^-A_{k,\eta}^-[F^{-1}]+c_k^+A_{k,\eta}^+[F^{-1}],\\
c_k^-A_{k,\eta}^-[\widetilde P_\eta^{-1}]
+c_k^+A_{k,\eta}^+[\widetilde P_\eta^{-1}]
&=c_k^-F^{-1}(\xi_k)+c_k^+F^{-1+}(\xi_k).
\end{aligned}
\end{equation}
By \eqref{eq:q_bound}, the quantiles $F^{-1}$, $\bar P^{-1}_\eta$, $\widetilde P^{-1}_\eta$, and
$G_0^{-1}$ are uniformly bounded on $[a,1-a]$.
Because the disjoint intervals $(\xi_k-\eta,\xi_k+\eta)$, $1\le k\le m$, have total length $2m\eta$, there exists $C_2\in\R_+$ such that, for $P\in\{\bar P_\eta, \widetilde P_\eta\}$,
\begin{equation}
W_p^p(P,G_0)\le\epsilon^p+C_2\eta,
\qquad \int|P^{-1}-F^{-1}|\,\mathrm d|\nu^c|\le C_2\eta.
\label{eq:reg_replacements}
\end{equation}
For sufficiently small $\eta$, this gives
$\|P^{-1}\|_p\le\|G_0^{-1}\|_p+\epsilon+1$.
Combining \eqref{eq:avg_rate}, \eqref{eq:prop3_iden}, and \eqref{eq:reg_replacements}, gives
\[
\begin{aligned}
\max\left\{
|\rho_h(\bar P_\eta)-\rho_{h_\eta}(F)|,
|\rho_{h_\eta}(\widetilde P_\eta)-\rho_h(F)|
\right\}\le\bigl[C_1(\|G_0^{-1}\|_p+\epsilon+1)+C_2\bigr]\eta^{1-1/p}.
\end{aligned}
\]
Define a feasible distribution $P^\circ$
by
\[
(P^\circ)^{-1}
=G_0^{-1}
+\frac{\epsilon}{\max\{\epsilon,W_p(P,G_0)\}}
(P^{-1}-G_0^{-1}).
\]
Then
\[
\begin{aligned}
W_p(P^\circ,G_0)\le\epsilon,\qquad
W_p(P^\circ,P)=(W_p(P,G_0)-\epsilon)_+
\le\frac{C_2}{p\epsilon^{p-1}}\eta,\qquad 0\le1-\frac{\epsilon}{\max\{\epsilon,W_p(P,G_0)\}}
\le\frac{C_2}{p\epsilon^p}\eta.
\end{aligned}
\]
By \eqref{eq:q_bound}, \eqref{eq:avg_bound}, and
\eqref{eq:endpoint_integral}, both $\rho_h$ and $\rho_{h_\eta}$ are uniformly bounded
at $P$ and $G_0$. Their quantile representations are linear in the
factor $\epsilon/\max\{\epsilon,W_p(P,G_0)\}$, so there exists
$C_3\in\R_+$ such that
\[
|\rho_h(P^\circ)-\rho_h(P)|
+|\rho_{h_\eta}(P^\circ)-\rho_{h_\eta}(P)|\le C_3\eta.
\]
It follows that, for every feasible $F$,
\[
\begin{aligned}
\max\left\{\rho_{h_\eta}(F)-S_h(\epsilon),
\rho_h(F)-S_{h_\eta}(\epsilon)\right\}\le\bigl[C_1(\|G_0^{-1}\|_p+\epsilon+1)+C_2+C_3\bigr]\eta^{1-1/p}.
\end{aligned}
\]
Taking suprema proves \eqref{eq:reg_rate}, where $C$ can be taken as $C=C_1(\|G_0^{-1}\|_p+\epsilon+1)+C_2+C_3$.

(ii) Choose $\delta\in(0,1/2)$ so that $h'$
is Lipschitz on $(0,\delta)$ and $(1-\delta,1)$.
For $p=1$, \eqref{eq:endpoint_integral} gives
\[
\int_0^\delta\frac{|h'(t)|}{t}\,\mathrm d t
+\int_{1-\delta}^1\frac{|h'(t)|}{1-t}\,\mathrm d t<\infty.
\]
Continuity of $h'$ at the endpoints therefore implies
$h'(0+)=h'(1-)=0$. For $p>1$, \eqref{eq:endpoint_integral} can be proved by part (i), and the conditions $h'(0+)=h'(1-)=0$ are assumed.
Choose $C_4\ge C_0/\delta$ such that 
$|h'(s)-h'(t)|\le C_4|s-t|$ for all $s,t\in(0,\delta)$ and $s,t\in(1-\delta,1)$.
Together with
$|\nu^c|(\mathrm d u)\le C_0\,\mathrm d u$ on $[\delta,1-\delta]$, this gives
\[
|\nu^c|(\mathrm d u)\le C_4(u\wedge(1-u))\,\mathrm d u.
\]
For $F\in\mathcal M_p$, Fubini's theorem gives
\[
\begin{aligned}
\int|F^{-1}-A_\eta[F^{-1}]|\,\mathrm d|\nu^c|
&\le C_4\int_0^1(u\wedge(1-u))
\bigl[F^{-1}(u)-F^{-1}(u-\delta_\eta(u))\bigr]\,\mathrm d u\\
&\le C_4\int_{(0,1)}
\left[\int_v^{(v+\eta)\wedge (2v) \wedge 1}(u\wedge(1-u))\,\mathrm d u\right]
\mathrm d F^{-1}(v)\\
&\le 2C_4\eta\int_{(0,1)}(v\wedge(1-v))\,\mathrm d F^{-1}(v)
\le 2C_4\eta\|F^{-1}\|_1.
\end{aligned}
\]
For $F\in\mathcal M_p(G_0,\epsilon)$, we have
$\|F^{-1}\|_1\le\|G_0^{-1}\|_p+\epsilon$, and hence
\[
\sup_{F\in\mathcal M_p(G_0,\epsilon)}
\int|F^{-1}-A_\eta[F^{-1}]|\,\mathrm d|\nu^c|\le 2C_4(\|G_0^{-1}\|_p+\epsilon)\eta.
\]
The remaining arguments to prove \eqref{eq:reg_linear} follow those in part (i), where $C$ can be taken as $C=2C_4(\|G_0^{-1}\|_p+\epsilon+1)+C_2+C_3$ for part (ii).
\end{proof}

\phantomsection
\label{app:grid}
\begin{proof}[Proof of Proposition~\ref{prop:step_approx}]
The vector $x_i=G_N^{-1}(t_i)$ is feasible for \eqref{eq:grid_problem},
and every feasible vector satisfies
\[
\Delta_i|x_i|^p\le(\|G_N^{-1}\|_p+\epsilon)^p,
\qquad i=1,\ldots,N.
\]
Thus the feasible set is nonempty and compact, and the linear objective
attains its maximum.
For $F\in\mathcal M_p$, define $P_N$ by
\[
P_N^{-1}(u)=\frac1{\Delta_i}\int_{I_i}F^{-1}(v)\,\mathrm d v,
\qquad u\in I_i.
\]
Because $\psi_{h_N}$ and $G_N^{-1}$ are constant on each $I_i$, we have
$\rho_{h_N}(P_N)=\rho_{h_N}(F)$ and $W_p(P_N,G_N)\le W_p(F,G_N)$, where the inequality follows from Jensen's inequality.
Hence $S_{h,N}(\epsilon)=\max_{F\in\M_p(G_N,\epsilon)}\rho_{h_N}(F)$.

We next establish convergence as the mesh tends to zero.
For every continuous function $\varphi:[0,1]\to\mathbb R$ and $u\in I_i$,
\[
\begin{aligned}
P_N^{-1}(u)-F^{-1}(u)
=\frac1{\Delta_i}\int_{I_i}[F^{-1}(v)-\varphi(v)]\,\mathrm d v+\frac1{\Delta_i}\int_{I_i}[\varphi(v)-\varphi(u)]\,\mathrm d v
+\varphi(u)-F^{-1}(u).
\end{aligned}
\]
By Minkowski's inequality and Jensen's inequality,
\[
\begin{aligned}
W_p(P_N,F)
\le{}&\left(\sum_i\Delta_i\left|
\frac1{\Delta_i}\int_{I_i}[F^{-1}(v)-\varphi(v)]\,\mathrm d v
\right|^p\right)^{1/p}+\sup_{|u-v|\le\max_i\Delta_i}
|\varphi(u)-\varphi(v)|+\|\varphi-F^{-1}\|_p\\
\le{}&2\|F^{-1}-\varphi\|_p
+\sup_{|u-v|\le\max_i\Delta_i}
|\varphi(u)-\varphi(v)|.
\end{aligned}
\]
The supremum tends to zero by uniform continuity of $\varphi$.
Because continuous functions on $[0,1]$ are dense in $L^p(0,1)$,
$\|F^{-1}-\varphi\|_p$ can be made arbitrarily small.
Applying the same argument with $F=G_0$ gives, for $1\le p<\infty$,
\begin{equation}
W_p(P_N,F)\to0,\qquad W_p(G_N,G_0)\to0
\qquad\mbox{as }\max_i\Delta_i\to0.
\label{eq:cell_approx}
\end{equation}
For $F\in\mathcal M_p(G_0,\epsilon)$, Jensen's inequality gives
$W_p(P_N,G_N)\le\epsilon$. Hence, we have
\[
\liminf_{N\to\infty}S_{h,N}(\epsilon)
\ge\lim_{N\to\infty}\rho_h(P_N)=\rho_h(F),
\]
where $|\rho_h(P_N)-\rho_h(F)|\le\|\psi_h\|_{p'}W_p(P_N,F)\to0$.
Taking the supremum over $F\in\mathcal M_p(G_0,\epsilon)$ gives
$\liminf_N S_{h,N}(\epsilon)\ge S_h(\epsilon)$.

For a grid optimizer $P_{\mathbf x_N^*}$, define
\[
\widehat P_N^{-1}
=G_0^{-1}
+\frac{\epsilon}{\epsilon+W_p(G_N,G_0)}
(P_{\mathbf x_N^*}^{-1}-G_0^{-1}).
\]
Then $W_p(\widehat P_N,G_0)\le\epsilon$ and
$W_p(\widehat P_N,P_{\mathbf x_N^*})\le W_p(G_N,G_0)$.
Therefore,
\[
S_{h,N}(\epsilon)
=\rho_h(P_{\mathbf x_N^*})
\le S_h(\epsilon)+\|\psi_h\|_{p'}W_p(G_N,G_0).
\]
Together with \eqref{eq:cell_approx}, this proves the desired limit.
\end{proof}

\subsection{Proofs for Section~\ref{sec:gap_bounds}}
\label{app:gap}

\begin{proof}[Proof of Proposition~\ref{prop:bench_gap}]
Take either $r=r^*$ or
$r=P_n^{-1}-G_0^{-1}$. Each shift is nondecreasing and constant on
every $C\in\mathcal I_h$. For $U\sim\U[0,1]$, the distribution
$F_{r(U)}$ therefore equals its $\mathcal I_h$-concentration.
Lemma~\ref{lem:concentration} and comonotonic additivity give
\[
\rho_h(G_0^{-1}(U)+r(U))
=\rho_h(G_0)+\rho_h(r(U))
=\rho_h(G_0)+\int_0^1r(u)\psi_{h^*}(u)\,\mathrm d u.
\]
The last integral equals $\epsilon\|\psi_{h^*}\|_{p'}$ for $r=r^*$
and converges to that value for the sequence of shifts.
Thus \eqref{eq:candidate_value} holds, with its stated limiting
interpretation, and feasibility together with $h\le h^*$ yields
\[
\rho_h(G_0)+\epsilon\|\psi_{h^*}\|_{p'}
\le S_h(\epsilon)\le
\rho_{h^*}(G_0)+\epsilon\|\psi_{h^*}\|_{p'}.
\]
This proves both benchmark bounds.
\end{proof}

\begin{proof}[Proof of Theorem~\ref{thm:radius_bound}]
(i) For every $C\in\mathcal I_h$, $\psi_{h^*}$ and $r^*$ are constant
Lebesgue-a.e.~on $C$, and
\[
\int_C\bigl[(G_0^{\mathcal I_h})^{-1}(u)-G_0^{-1}(u)\bigr]\,\mathrm d u=0.
\]
Outside $\bigcup_{C\in\mathcal I_h}C$, the two benchmark quantiles agree. It follows that
\begin{equation}
\int_0^1\bigl[(G_0^{\mathcal I_h})^{-1}(u)-G_0^{-1}(u)\bigr]
\psi_{h^*}(u)\,\mathrm d u=0;
\label{eq:radial_identities}
\end{equation}
and by Jensen's inequality, we have, for each $C$,
\[
\begin{aligned}
&\int_C\bigl|(G_0^{\mathcal I_h})^{-1}(u)-G_0^{-1}(u)+r^*(u)\bigr|^p
\,\mathrm d u\\
&\quad\ge\lambda(C)\left|
\frac1{\lambda(C)}\int_C\bigl[(G_0^{\mathcal I_h})^{-1}(u)-G_0^{-1}(u)+r^*(u)\bigr]
\,\mathrm d u\right|^p
=\int_C|r^*(u)|^p\,\mathrm d u.
\end{aligned}
\]
Because $((F^*)^{\mathcal I_h})^{-1}=(G_0^{\mathcal I_h})^{-1}+r^*$, we have
\begin{equation*}
\begin{aligned}
W_p((F^*)^{\mathcal I_h},G_0)=\|(G_0^{\mathcal I_h})^{-1}-G_0^{-1}+r^*\|_p\ge\epsilon.
\end{aligned}
\end{equation*}
Thus \eqref{eq:radial_candidate} gives
\[
\begin{aligned}
&P_\epsilon^{-1}
=\left(1-\frac{\epsilon}{W_p((F^*)^{\mathcal I_h},G_0)}\right)G_0^{-1}
+\frac{\epsilon}{W_p((F^*)^{\mathcal I_h},G_0)}
((F^*)^{\mathcal I_h})^{-1},\\
&W_p(P_\epsilon,G_0)
=\frac{\epsilon}{W_p((F^*)^{\mathcal I_h},G_0)}
\|((F^*)^{\mathcal I_h})^{-1}-G_0^{-1}\|_p=\epsilon.
\end{aligned}
\]
Both coefficients in the first line are nonnegative and sum to one,
so $P_\epsilon^{-1}$ is nondecreasing. Since $\bigl((F^*)^{\mathcal I_h}\bigr)^{\mathcal I_h}=(F^*)^{\mathcal I_h}$,
applying Lemma~\ref{lem:concentration} to $\{(F^*)^{\mathcal I_h}\}$ gives
\[
\rho_h((F^*)^{\mathcal I_h})
=\rho_{h^*}((F^*)^{\mathcal I_h})
=\rho_{h^*}(F^*)=S_{h^*}(\epsilon).
\]
Comonotonic additivity and
\eqref{eq:bench_hstar} give \eqref{eq:radial_value}.
Because $\rho_h(P_\epsilon)\le S_h(\epsilon)\le S_{h^*}(\epsilon)$,
\eqref{eq:radial_value} proves \eqref{eq:radial_gap}.
For $p=1$, $P_n$ in \eqref{eq:bench_sequence} has
$P_n^{-1}-G_0^{-1}$ constant on each $C\in\mathcal I_h$.
The preceding Jensen inequality therefore gives
$W_1(P_n^{\mathcal I_h},G_0)\ge W_1(P_n,G_0)=\epsilon$ and
$W_1(P_{\epsilon,n},G_0)=\epsilon$.
Hence, we have
\[
\rho_h(P_{\epsilon,n})\le S_h(\epsilon)\le S_{h^*}(\epsilon)
=\rho_{h^*}(G_0)+\epsilon\|\psi_{h^*}\|_\infty.
\]

For $p>1$, \eqref{eq:bench_shift} gives
$\operatorname{sgn}(r^*)|r^*|^{p-1}
=\epsilon^{p-1}\psi_{h^*}/\|\psi_{h^*}\|_{p'}$.
By \eqref{eq:radial_identities}, we obtain
\[
\int_0^1\operatorname{sgn}(r^*(u))|r^*(u)|^{p-1}
\bigl[(G_0^{\mathcal I_h})^{-1}(u)-G_0^{-1}(u)\bigr]\,\mathrm d u=0.
\]
There exists $C_p<\infty$ such that, for all
$x,y\in\mathbb R$,
\[
0\le |x+y|^p-|x|^p-p\operatorname{sgn}(x)|x|^{p-1}y
\le C_p
\begin{cases}
|y|^p,&1<p\le2,\\
|x|^{p-2}|y|^2+|y|^p,&p>2.
\end{cases}
\]
Indeed, for $1<p\le2$, this follows by integrating
$\bigl|\operatorname{sgn}(x+ty)|x+ty|^{p-1}
-\operatorname{sgn}(x)|x|^{p-1}\bigr|\le2^{2-p}|ty|^{p-1}$
over $0\le t\le1$. For $p>2$, apply Taylor's integral remainder and
$|x+ty|^{p-2}\le2^{(p-3)_+}(|x|^{p-2}+|y|^{p-2})$.
Apply these inequalities with $x=r^*(u)$ and
$y=(G_0^{\mathcal I_h})^{-1}(u)-G_0^{-1}(u)$.
Integration, \eqref{eq:radial_identities}, and H\"older's inequality yield
\[
0\le W_p^p((F^*)^{\mathcal I_h},G_0)-\epsilon^p
\le C_p
\begin{cases}
W_p^p(G_0^{\mathcal I_h},G_0),&1<p\le2,\\
\epsilon^{p-2}W_p^2(G_0^{\mathcal I_h},G_0)
+W_p^p(G_0^{\mathcal I_h},G_0),&p>2.
\end{cases}
\]
Consequently, for $\epsilon\ge1$, we have
\[
0\le W_p((F^*)^{\mathcal I_h},G_0)-\epsilon
\le\frac{W_p^p((F^*)^{\mathcal I_h},G_0)-\epsilon^p}{p\epsilon^{p-1}}
=O\!\left(\epsilon^{-\min\{p-1,1\}}\right).
\]
Because $W_p((F^*)^{\mathcal I_h},G_0)\ge\epsilon$, we have
\[
\begin{aligned}
&\left(1-\frac{\epsilon}{W_p((F^*)^{\mathcal I_h},G_0)}\right)
\bigl[\rho_{h^*}(G_0)-\rho_h(G_0)+\epsilon\|\psi_{h^*}\|_{p'}\bigr]\\
&\quad\le\bigl[W_p((F^*)^{\mathcal I_h},G_0)-\epsilon\bigr]
\left[\frac{\rho_{h^*}(G_0)-\rho_h(G_0)}{\epsilon}
+\|\psi_{h^*}\|_{p'}\right]
=O\!\left(\epsilon^{-\min\{p-1,1\}}\right),
\end{aligned}
\]
which proves the $O(\cdot)$-bound in \eqref{eq:radial_gap} for
$\epsilon\ge1$. For $0<\epsilon\le1$, we have
\[
\begin{aligned}
&\left(1-\frac{\epsilon}{W_p((F^*)^{\mathcal I_h},G_0)}\right)
\bigl[\rho_{h^*}(G_0)-\rho_h(G_0)+\epsilon\|\psi_{h^*}\|_{p'}\bigr]\\
&\quad\le\rho_{h^*}(G_0)-\rho_h(G_0)+\|\psi_{h^*}\|_{p'}\\
&\quad\le\bigl[\rho_{h^*}(G_0)-\rho_h(G_0)+\|\psi_{h^*}\|_{p'}\bigr]
\epsilon^{-\min\{p-1,1\}}.
\end{aligned}
\]
Thus this bound holds for every $\epsilon>0$.

Finally, let $\epsilon_n\downarrow0$ and choose
$\widetilde P_n\in\mathcal M_p(G_0,\epsilon_n)$ such that
$S_h(\epsilon_n)-1/n\le\rho_h(\widetilde P_n)\le S_h(\epsilon_n)$.
H\"older's inequality gives
\[
|\rho_{h^*}(\widetilde P_n)-\rho_{h^*}(G_0)|
\le\epsilon_n\|\psi_{h^*}\|_{p'}\to 0.
\]
Because $W_p(\widetilde P_n,G_0)\le\epsilon_n\to0$, we have
$\widetilde P_n\xrightarrow{\mathrm{w}} G_0$.
The function $h^*-h$ is nonnegative and lower semicontinuous.
Fatou's lemma therefore gives
\[
\begin{aligned}
\rho_{h^*}(G_0)-\rho_h(G_0)
&=\int_{\mathbb R}\bigl[h^*(1-G_0(x))-h(1-G_0(x))\bigr]\,\mathrm d x\\
&\le\liminf_{n\to\infty}\int_{\mathbb R}
\bigl[h^*(1-\widetilde P_n(x))-h(1-\widetilde P_n(x))\bigr]\,\mathrm d x=\liminf_{n\to\infty}
\bigl[\rho_{h^*}(\widetilde P_n)-\rho_h(\widetilde P_n)\bigr].
\end{aligned}
\]
Consequently,
\[
\begin{aligned}
\limsup_{n\to\infty}S_h(\epsilon_n)=\limsup_{n\to\infty}\rho_h(\widetilde P_n)=\rho_{h^*}(G_0)-\liminf_{n\to\infty}
\bigl[\rho_{h^*}(\widetilde P_n)-\rho_h(\widetilde P_n)\bigr]
\le\rho_h(G_0).
\end{aligned}
\]
Because $S_h(\epsilon)\ge\rho_h(G_0)$, we conclude that
$S_h(\epsilon)\to\rho_h(G_0)$ as $\epsilon\downarrow0$.
By \eqref{eq:radial_value},
\[
0\le S_h(\epsilon)-\rho_h(P_\epsilon)
\le S_h(\epsilon)-\rho_h(G_0)\to 0,
\qquad\epsilon\downarrow0.
\]

(ii) For absolutely continuous $h$ with $\psi_h\in L^{p'}(0,1)$,
H\"older's inequality and \eqref{eq:radial_value} give
\[
0\le S_h(\epsilon)-\rho_h(P_\epsilon)
\le\rho_h(G_0)+\epsilon\|\psi_h\|_{p'}-\rho_h(P_\epsilon)
\le\epsilon\|\psi_h\|_{p'},
\]
which proves \eqref{eq:radial_holder}.
\end{proof}

\subsection{Proofs for Section~\ref{sec:portfolio}}
\label{app:portfolio}

For $m\in\mathbb R$, $s>0$, and $X,Y\in L^p$, the quantile
representation gives
\[
W_p(F_{m+sX},F_{m+sY})=sW_p(F_X,F_Y),\qquad
\rho_h(m+sX)=mh(1)+s\rho_h(X).
\]
Take $m=\bm\omega^\top\bm\mu$, $s=s_{\bm\omega}>0$, and
$F_Y=P_\varphi$. Then
\[
F_{m+sX}\in
\mathcal M_p(G_{0,\bm\omega},\epsilon\|\bm\omega\|_{r'})
\quad\Longleftrightarrow\quad
F_X\in\mathcal M_p\!\left(P_\varphi,
\frac{\epsilon\|\bm\omega\|_{r'}}{s_{\bm\omega}}\right).
\]

\begin{proof}[Proof of Proposition~\ref{prop:outer_flatness}]
For $\bm\omega\in\mathcal W$, $\bm\Sigma\succ0$ and
$\bm\omega\ne0$ imply $s_{\bm\omega}>0$. Because
$G_{0,\bm\omega}^{-1}
=\bm\omega^\top\bm\mu+s_{\bm\omega}P_\varphi^{-1}$,
it follows that $G_{0,\bm\omega}^{-1}$ is flat on $\mathcal I_h$ if and only if
$P_\varphi^{-1}$ is flat on $\mathcal I_h$. Under flatness,
Theorem~\ref{thm:flatness} gives
$V_h(\bm\omega)=V_{h^*}(\bm\omega)$ for every
$\bm\omega\in\mathcal W$, so the minima and minimizers coincide.

Otherwise, continuity of \eqref{eq:elliptical_value} and compactness
of $\mathcal W$ give a minimizer $\bm\omega^{\mathrm{env}}$ of $V_{h^*}$.
For $p>1$, Theorem~\ref{thm:flatness}(i) yields
\[
\inf_{\bm\omega\in\mathcal W}V_h(\bm\omega)
\le V_h(\bm\omega^{\mathrm{env}})
<V_{h^*}(\bm\omega^{\mathrm{env}})
=\min_{\bm\omega\in\mathcal W}V_{h^*}(\bm\omega).
\]
For $p=1$, Theorem~\ref{thm:flatness}(ii) gives $\delta_0>0$ such that
$S_h(\delta)<S_{h^*}(\delta)$ for $0<\delta<\delta_0$.
If
\[
0<\epsilon<
\frac{\delta_0}{\displaystyle
\max_{\bm\omega\in\mathcal W}\|\bm\omega\|_{r'}/s_{\bm\omega}},
\]
then $0<\epsilon\|\bm\omega\|_{r'}/s_{\bm\omega}<\delta_0$ for every
$\bm\omega\in\mathcal W$. By \eqref{eq:elliptical_inner},
$V_h(\bm\omega^{\mathrm{env}})<V_{h^*}(\bm\omega^{\mathrm{env}})$, and the
same chain of inequalities proves (ii).
\end{proof}

\begin{proof}[Proof of Proposition~\ref{prop:outer_reg}]
For $P_i\in\mathcal M_p(P_\varphi,\delta_i)$, $i=1,2$, and
$\theta\in[0,1]$, define
$P_\theta^{-1}=(1-\theta)P_1^{-1}+\theta P_2^{-1}$.
Minkowski's inequality and comonotonic additivity give
\[
W_p(P_\theta,P_\varphi)\le(1-\theta)\delta_1+\theta\delta_2,
\qquad
\rho_h(P_\theta)=(1-\theta)\rho_h(P_1)+\theta\rho_h(P_2).
\]
Taking suprema yields
\[
S_h((1-\theta)\delta_1+\theta\delta_2)
\ge(1-\theta)S_h(\delta_1)+\theta S_h(\delta_2).
\]
The same argument applies to $h_\eta$. Both value functions are
finite and nondecreasing, and concavity makes them continuous on
$(0,\infty)$.
Because $\bm\Sigma\succ0$, compactness of $\{\bm\omega:\|\bm\omega\|_{r'}=1\}$
gives
\[
0<a:=\min_{\|\bm\omega\|_{r'}=1}\frac{\epsilon}{s_{\bm\omega}}
\le\frac{\epsilon\|\bm\omega\|_{r'}}{s_{\bm\omega}}
\le \max_{\|\bm\omega\|_{r'}=1}\frac{\epsilon}{s_{\bm\omega}}=:b<\infty.
\]
By \eqref{eq:elliptical_inner}, $V_h$ is continuous on $\bm\omega\ne 0$, and
\[
|V_h(\bm\omega)|
\le |h(1)|\,|\bm\omega^\top\bm\mu|
+s_{\bm\omega}\max_{\delta\in[a,b]}|S_h(\delta)|
\to0\qquad\mbox{as }\bm\omega\to0.
\]
The same holds for $V_{h_\eta}$. Hence both minima on compact
$\mathcal W$ are attained.

For (i), concavity gives
\begin{equation}\label{eq:prop4_concave}
S_{h_\eta}(t)\le
S_{h_\eta}(\delta)+(t-\delta)\partial_+S_{h_\eta}(\delta),
\qquad t,\delta>0,
\end{equation}
with equality at $\delta=t$.
For each $\bm\omega\ne0$, take
$t=\epsilon\|\bm\omega\|_{r'}/s_{\bm\omega}$.
Because $s_{\bm\omega}>0$, \eqref{eq:elliptical_inner} and \eqref{eq:prop4_concave} give
\[
\begin{aligned}
V_{h_\eta}(\bm\omega)
=h(1)\bm\omega^\top\bm\mu+s_{\bm\omega}S_{h_\eta}(t)=\min_{\delta>0}\Bigl\{
 h(1)\bm\omega^\top\bm\mu
+[S_{h_\eta}(\delta)-\delta\partial_+S_{h_\eta}(\delta)]s_{\bm\omega}
+\epsilon\partial_+S_{h_\eta}(\delta)\|\bm\omega\|_{r'}\Bigr\}.
\end{aligned}
\]
The minimum identity also holds for $\bm\omega=0$, with both sides being zero.
Taking the minimum over $\bm\omega\in\mathcal W$ gives
\eqref{eq:reg_outer_solution}.
For a minimizing pair $(\delta_\eta^*,\bm\omega_\eta^*)$,
\[
\begin{aligned}
V_{h_\eta}(\bm\omega_\eta^*)
\le h(1)(\bm\omega_\eta^*)^\top\bm\mu
+[S_{h_\eta}(\delta_\eta^*)-\delta_\eta^*
\partial_+S_{h_\eta}(\delta_\eta^*)]s_{\bm\omega_\eta^*}+\epsilon\,\partial_+S_{h_\eta}(\delta_\eta^*)
\|\bm\omega_\eta^*\|_{r'}
=\min_{\bm\omega\in\mathcal W}V_{h_\eta}(\bm\omega).
\end{aligned}
\]
For (ii), Theorem~\ref{thm:regularization} gives
$S_{h_\eta}(\delta)\to S_h(\delta)$ for every $\delta>0$.
Because $S_h$ is finite and
concave on $(0,\infty)$, it is uniformly continuous on $[a,b]$.
Thus, for $a<b$ and $\zeta>0$, we can choose a partition
$a=t_0<\cdots<t_m=b$ such that
$\max_i[S_h(t_i)-S_h(t_{i-1})]<\zeta$.
Monotonicity gives
\[
\sup_{\delta\in[a,b]}|S_{h_\eta}(\delta)-S_h(\delta)|
\le\max_{0\le i\le m}|S_{h_\eta}(t_i)-S_h(t_i)|+\zeta.
\]
Letting $\zeta\downarrow0$ proves uniform convergence of $S_{h_\eta}$
to $S_h$ on $[a,b]$. If $a=b$, this follows directly from pointwise
convergence.
Because $h_\eta(1)=h(1)$, \eqref{eq:elliptical_inner} gives
\[
\sup_{\bm\omega\in\mathcal W}|V_{h_\eta}(\bm\omega)-V_h(\bm\omega)|
\le\max_{\bm\omega\in\mathcal W}s_{\bm\omega}
\sup_{\delta\in[a,b]}|S_{h_\eta}(\delta)-S_h(\delta)|\to0~~\mbox{as}~\eta\downarrow 0.
\]
Consequently,
\[
\begin{aligned}
\left|\min_{\mathcal W}V_{h_\eta}-\min_{\mathcal W}V_h\right|\le\sup_{\mathcal W}|V_{h_\eta}-V_h|\to0,\qquad
0\le V_h(\bm\omega_\eta^*)-\min_{\mathcal W}V_h
\le2\sup_{\mathcal W}|V_{h_\eta}-V_h|\to0.
\end{aligned}
\]
If $\bm\omega_{\eta_n}^*\to\bm\omega^*$, continuity gives
$V_h(\bm\omega^*)=\lim_nV_h(\bm\omega_{\eta_n}^*)
=\min_{\mathcal W}V_h$.
If the minimizer is unique, compactness makes every convergent
subsequence have this same limit, so $\bm\omega_\eta^*\to\bm\omega^*$.
\end{proof}

\Needspace{6\baselineskip}
\begin{proof}[Proof of Proposition~\ref{prop:outer_gap}]
The continuous objective \eqref{eq:elliptical_value} attains its
minimum on compact $\mathcal W$.
For $s_{\bm\omega}>0$, apply \eqref{eq:gap_h_bound} at radius
$\epsilon\|\bm\omega\|_{r'}/s_{\bm\omega}$ and use
\eqref{eq:elliptical_inner} to obtain
\[
V_{h^*}(\bm\omega)
-[\rho_{h^*}(P_\varphi)-\rho_h(P_\varphi)]s_{\bm\omega}
\le V_h(\bm\omega)\le V_{h^*}(\bm\omega).
\]
If $s_{\bm\omega}=0$, the benchmark is a point mass and
$\rho_{h^*}(G_{0,\bm\omega})-\rho_h(G_{0,\bm\omega})=0$;
the same bound follows directly from \eqref{eq:gap_h_bound}.
Taking infima and evaluating the upper bound at
$\bm\omega^{\mathrm{env}}$ gives
\[
\begin{aligned}
0\le V_h(\bm\omega^{\mathrm{env}})-\inf_{\mathcal W}V_h
\le V_{h^*}(\bm\omega^{\mathrm{env}})
-\min_{\bm\omega\in\mathcal W}
\{V_{h^*}(\bm\omega)
-[\rho_{h^*}(P_\varphi)-\rho_h(P_\varphi)]s_{\bm\omega}\}\le[\rho_{h^*}(P_\varphi)-\rho_h(P_\varphi)]
\max_{\bm\omega\in\mathcal W}s_{\bm\omega},
\end{aligned}
\]
where the last inequality uses
$V_{h^*}(\bm\omega^{\mathrm{env}})=\min_{\mathcal W}V_{h^*}$.
This proves \eqref{eq:decision_bound}.
\end{proof}

\begin{proof}[Proof of Proposition~\ref{prop:outer_radial}]
The shift in \eqref{eq:bench_shift} is proportional to the radius.
Thus \eqref{eq:radial_value} makes $\delta\mapsto\rho_h(P_\delta)$
continuous on $(0,\infty)$. Because $s_{\bm\omega}>0$ on $\mathcal W$,
\eqref{eq:radial_outer} is continuous and attains its minimum on
compact $\mathcal W$.
For $\delta=\epsilon\|\bm\omega\|_{r'}/s_{\bm\omega}$, feasibility of
$P_\delta$ and \eqref{eq:elliptical_inner} give
\[
\widehat V_h(\bm\omega)\le V_h(\bm\omega)\le V_{h^*}(\bm\omega),
\qquad
V_h(\bm\omega)-\widehat V_h(\bm\omega)
=s_{\bm\omega}[S_h(\delta)-\rho_h(P_\delta)].
\]
Consequently,
\[
\begin{aligned}
0\le V_h(\bm\omega^{\mathrm{app}})-\inf_{\bm\omega\in\mathcal W}V_h(\bm\omega)
\le V_h(\bm\omega^{\mathrm{app}})-\min_{\bm\omega\in\mathcal W}\widehat V_h(\bm\omega)=V_h(\bm\omega^{\mathrm{app}})-\widehat V_h(\bm\omega^{\mathrm{app}})\le V_{h^*}(\bm\omega^{\mathrm{app}})-\widehat V_h(\bm\omega^{\mathrm{app}}),
\end{aligned}
\]
which proves \eqref{eq:radial_decision}.
Compactness and $0\notin\mathcal W$ give
\[
0<\min_{\bm\omega\in\mathcal W}
\frac{\|\bm\omega\|_{r'}}{s_{\bm\omega}}
\le\frac{\|\bm\omega\|_{r'}}{s_{\bm\omega}}
\le\max_{\bm\omega\in\mathcal W}
\frac{\|\bm\omega\|_{r'}}{s_{\bm\omega}}<\infty,
\qquad \max_{\bm\omega\in\mathcal W}s_{\bm\omega}<\infty.
\]
Theorem~\ref{thm:radius_bound}(i) yields, for $p>1$,
$\sup_{\bm\omega\in\mathcal W}
[V_{h^*}(\bm\omega)-\widehat V_h(\bm\omega)]
=O\!\left(\epsilon^{-\min\{p-1,1\}}\right)$ as $\epsilon\to\infty.
$
Also, \eqref{eq:radial_value} gives
$\rho_h(P_\delta)\ge\rho_h(P_\varphi)$ and
$\rho_h(P_\delta)\to\rho_h(P_\varphi)$ as $\delta\downarrow0$.
Together with Theorem~\ref{thm:radius_bound}(i), this implies
$S_h(\delta)\to\rho_h(P_\varphi)$.
Monotonicity of $S_h$ therefore gives
\[
\begin{aligned}
0\le\sup_{\bm\omega\in\mathcal W}
[V_h(\bm\omega)-\widehat V_h(\bm\omega)]\le\max_{\bm\omega\in\mathcal W}s_{\bm\omega}
\left[S_h\!\left(\epsilon\max_{\bm\omega\in\mathcal W}
\frac{\|\bm\omega\|_{r'}}{s_{\bm\omega}}\right)
-\rho_h(P_\varphi)\right]\to0,
\qquad\epsilon\downarrow0.
\end{aligned}
\]
This proves the remaining assertion in (i).
For (ii), Theorem~\ref{thm:radius_bound}(ii) and
$\rho_h(P_\delta)\ge\rho_h(P_\varphi)$ give
\[
\begin{aligned}
V_h(\bm\omega)-\widehat V_h(\bm\omega)
\le h(1)\bm\omega^\top\bm\mu+\rho_h(P_\varphi)s_{\bm\omega}
+\epsilon\|\psi_h\|_{p'}\|\bm\omega\|_{r'}-\widehat V_h(\bm\omega)
\le\epsilon\|\psi_h\|_{p'}\|\bm\omega\|_{r'}
\le\epsilon\|\psi_h\|_{p'}\max_{\bm\omega\in\mathcal W}\|\bm\omega\|_{r'}.
\end{aligned}
\]
Evaluating at $\bm\omega^{\mathrm{app}}$ proves \eqref{eq:radial_decision_ac}.
\end{proof}

\end{document}